\documentclass[11pt]{article}
\usepackage{verbatim,amsmath,amsthm,eufrak}
\usepackage{cite} 
\usepackage{yhmath}
\usepackage{makeidx}
\usepackage{yhmath}
\usepackage{enumitem}
\usepackage{color,xcolor}
\usepackage[hypertexnames=false,hyperfootnotes=false,colorlinks=true,linkcolor=blue,%
citecolor=purple,filecolor=magenta,urlcolor=cyan,unicode,linktocpage=true,pagebackref=false]{hyperref}
\usepackage{scalerel,stackengine}

\renewcommand{\arraystretch}{1.2}
\makeatletter
\newdimen\normalarrayskip              
\newdimen\minarrayskip                 
\normalarrayskip\baselineskip
\minarrayskip\jot
\newif\ifold             \oldtrue            \def\new{\oldfalse}
\def\arraymode{\ifold\relax\else\displaystyle\fi} 
\def\eqnumphantom{\phantom{(\theequation)}}     
\def\@arrayskip{\ifold\baselineskip\z@\lineskip\z@
     \else
     \baselineskip\minarrayskip\lineskip2\minarrayskip\fi}
\def\@arrayclassz{\ifcase \@lastchclass \@acolampacol \or
\@ampacol \or \or \or \@addamp \or
   \@acolampacol \or \@firstampfalse \@acol \fi
\edef\@preamble{\@preamble
  \ifcase \@chnum
     \hfil$\relax\arraymode\@sharp$\hfil
     \or $\relax\arraymode\@sharp$\hfil
     \or \hfil$\relax\arraymode\@sharp$\fi}}
\def\@array[#1]#2{\setbox\@arstrutbox=\hbox{\vrule
     height\arraystretch \ht\strutbox
     depth\arraystretch \dp\strutbox
     width\z@}\@mkpream{#2}\edef\@preamble{\halign
\noexpand\@halignto
\bgroup \tabskip\z@ \@arstrut \@preamble \tabskip\z@ \cr}%
\let\@startpbox\@@startpbox \let\@endpbox\@@endpbox
  \if #1t\vtop \else \if#1b\vbox \else \vcenter \fi\fi
  \bgroup \let\par\relax
  \let\@sharp##\let\protect\relax
  \@arrayskip\@preamble}
\def\eqnarray{\stepcounter{equation}%
              \let\@currentlabel=\theequation
              \global\@eqnswtrue
              \global\@eqcnt\z@
              \tabskip\@centering
              \let\\=\@eqncr
 \halign to \displaywidth\bgroup
    \eqnumphantom\@eqnsel\hskip\@centering
    $\displaystyle \tabskip\z@ {##}$%
    \global\@eqcnt\@ne \hskip 2\arraycolsep
         $\displaystyle\arraymode{##}$\hfil
    \global\@eqcnt\tw@ \hskip 2\arraycolsep
         $\displaystyle\tabskip\z@{##}$\hfil
         \tabskip\@centering
    &{##}\tabskip\z@\cr}

\newcounter{app}

\def\app{\setcounter{equation}{0}
\def\theequation{A\Roman{app}.\arabic{equation}}\par
   \addvspace{4ex}
   \@afterindentfalse
  \secdef\@app\@dapp}
\newcommand\@app{\@startsection {app}{1}{0ex}%
                                   {-3.5ex \@plus -1ex \@minus -.2ex}%
                                   {2.3ex \@plus.2ex}%
                                   {\normalfont\Large\bf}}

\def\@dapp#1{%
{\parindent \z@ \raggedright  \bf #1}\par\nobreak}
\def\l@app#1#2{\ifnum \c@tocdepth >\z@
    \addpenalty\@secpenalty
    \addvspace{1.0em \@plus\p@}%
    \setlength\@tempdima{8.5em}%
    \begingroup
      \parindent \z@ \rightskip \@pnumwidth
      \parfillskip -\@pnumwidth
      \leavevmode \bfseries
      \advance\leftskip\@tempdima
      \hskip -\leftskip
      #1\nobreak\hfil \nobreak\hb@xt@\@pnumwidth{\hss #2}\par
    \endgroup\fi}
\newcounter{sapp}[app]

\def\sapp{\def\theequation{A\arabic{app}.\arabic{equation}}\par
   \@afterindentfalse
  \secdef\@sapp\@dsapp}
\newcommand\@sapp{\@startsection{sapp}{2}{\z@}%
                                     {-3.25ex\@plus -1ex \@minus -.2ex}%
                                     {1.5ex \@plus .2ex}%
                                     {\normalfont\large\bfseries}}

\def\@dsapp#1{%
{\parindent \z@ \raggedright  \bf #1}\par\nobreak}
\newcommand{\l@sapp}{\@dottedtocline{2}{1.5em}{3em}}
\def\draft{\oddsidemargin -.5truein
        \def\@oddfoot{\sl preliminary draft \hfil
        \rm\thepage\hfil\sl\today\quad\militarytime}
        \let\@evenfoot\@oddfoot \overfullrule 3pt
        \let\label=\draftlabel
        \let\marginnote=\draftmarginnote
   \def\@eqnnum{(\theequation)\rlap{\kern\marginparsep\tt\@eqnlabel}%
\global\let\@eqnlabel\@vacuum}  }

\def\be{\begin{eqnarray}}
\def\ee{\end{eqnarray}}
\def\nn{\nonumber}

\def\p{\partial}

\def\beq{\begin{equation}}
\def\eeq{\end{equation}}
\def\ba{\beq\new\begin{array}{c}}
\def\ea{\end{array}\eeq}
\def\be{\ba}
\def\ee{\ea}

\def\Tr{{\rm Tr}\,}
\def\dim{{\rm dim}\,}
\def\res{{\rm res}\,}
\def\diag{{\rm diag}\,}
\def\sppan{{\rm span}\,}

\newfont{\Bbbb}{msbm7 scaled 1\@ptsize00}

\newcommand{\z}{\raise-1pt\hbox{$\mbox{\Bbbb Z}$}}
\newcommand{\cc}{\raise-1pt\hbox{$\mbox{\Bbbb C}$}}
\newcommand{\rr}{\raise-1pt\hbox{$\mbox{\Bbbb R}$}}

\def\Gr{{\rm Gr}\,}

\newcommand{\<}{\left <}
\renewcommand{\>}{\right >}

\def\lvac{\left <0\right |}
\def\rvac{\left |0\right >}

\newcommand{\id}{{\mathrm{id}}}

\newfont{\alef}{msbm10 at 11pt}
\newfont {\goth}{eufm10 at 11pt}
\def\mathbb#1{\hbox{{\alef #1}}}

\DeclareMathOperator{\GL}{GL}
\DeclareMathOperator{\sll}{sl}
\DeclareMathOperator{\gl}{gl}
\DeclareMathOperator{\Ai}{Ai}
\DeclareMathOperator{\Bi}{Bi}
\DeclareMathOperator{\He}{He}

\let\@@savethanks\thanks
\def\thanks#1{\gdef\thefootnote{\alph{footnote}}\@@savethanks{#1}}

\newtheorem{theorem}{Theorem}
\newtheorem{lemma}{Lemma}[section]
\newtheorem{proposition}[lemma]{Proposition}
\newtheorem{corollary}[lemma]{Corollary}
\newtheorem{remark}{Remark}[section]

\newtheorem{conjecture}{Conjecture}[section]
\newtheorem*{theorem*}{Theorem}

\newtheorem{definition}{Definition}

\numberwithin{equation}{section}

\makeatletter
\g@addto@macro \normalsize {%
 \setlength\abovedisplayskip{14pt plus 3pt minus 3pt}%
 \setlength\belowdisplayskip{14pt plus 3pt minus 3pt}%
  \setlength\abovedisplayshortskip{11pt plus 3pt minus 3pt}%
 \setlength\belowdisplayshortskip{11pt plus 3pt minus 3pt}%
}
\makeatother

\bigskip
\bigskip

\title{
\bigskip
{\bf KP integrability of triple Hodge integrals. II.\\
Generalized Kontsevich matrix model } \vspace{.5cm}}
\author{{\bf Alexander Alexandrov}\thanks{E-mail:  {\tt alexandrovsash at gmail.com}}
\date{ } \\
{\small {\it Center for Geometry and Physics, Institute for Basic Science (IBS), Pohang 37673, Korea}}
}

\begin{document}

\setcounter{footnote}{0}

\setcounter{tocdepth}{3}

\maketitle

\vspace{-8.0cm}

\begin{center}
\end{center}

\vspace{6.5cm}

\begin{center}
\today
\end{center}

\begin{abstract} 
In this paper we introduce a new family of the KP tau-functions. This family can be described by a deformation of the generalized 
Kontsevich matrix model. We prove that the simplest representative of this family describes a generating function of the cubic Hodge integrals satisfying the Calabi-Yau condition, and claim that the whole family describes its generalization for the higher spin cases. To investigate this family we construct a new description of the Sato Grassmannian in terms of a canonical pair of the Kac-Schwarz operators. 
\end{abstract}

\bigskip


\bigskip

{\small \bf MSC 2020 Primary: 37K10, 14N35, 81R10, 14N10; Secondary:  81T32.}

\newpage

\tableofcontents

\def\thefootnote{\arabic{footnote}}
\section{Introduction}
\setcounter{equation}{0}

Methods of integrable systems play an important role in the modern mathematical physics and enumerative geometry.  Although numerous applications have proven their universality, new models often desire new, underdeveloped elements of the general theory of integrable systems. Hence, simultaneously with investigation of new applications it is necessary to develop new methods of integrable systems. Moreover, integrability of generating functions is often closely related to other universal structures which include matrix models, Virasoro constraints, and quantum spectral curves. So it is necessary to better understand the role of these elements in the general scheme of  integrability.

The main goal of this paper is to investigate a new, infinite-dimensional family of tau-functions of the Kadomtsev-Petviashvili (KP) hierarchy. This family  can be described by a deformation of the generalized Kontsevich model (GKM) with polynomial potential, so we call it {\em deformed generalized Kontsevich model}. The deformed potential is an infinite formal series of certain type, parametrized by the auxiliary parameters. The choice of deformation is motivated by a connection between its simplest representative  and the generating function of the cubic Hodge integrals satisfying the Calabi-Yau condition, which we prove in this paper.

The deformed GKM shares a lot of common properties with the polynomial GKM.  However, the standard methods of GKM suitable for the  polynomial or antipolynomial potential, are not always convenient for the cases, when the potential is a Taylor or Laurent series, convergent or formal. Hence, one has to develop new methods, suitable for the investigation of tau-functions of GKM with non-polynomials potential. While the matrix integral description for the deformed model is more sophisticated then one for the pure GKM, we can apply the methods of the KP hierarchy, in particular those associated with the Kac-Schwarz operators, to investigate this deformed GKM. 


\subsection{Sato Grassmannian and Kac-Schwarz operators}
It is well-known that the tau-functions of the KP hierarchy can be described by the elements of $\GL(\infty)$ group or its central extension. However, this is not one-to-one correspondence at all -- there are many group elements corresponding to a tau-function. According to Sato's theorem, among all these group elements there exists a canonical one, the so-called Sato's group element. This element is particularly convenient for the description of the Sato Grassmannian -- the space of solutions of the KP hierarchy. Namely, for any point of the Sato Grassmannian, ${\mathcal W}\in \Gr_+^{(0)}$, the Sato group element is of the form
\be
{\mathtt G}_{\mathcal W}\in 1+{\mathcal D}_-,
\ee
where ${\mathcal D}_- = z^{-1}{\mathbb C}[[z^{-1}]][[\frac{\p}{\p z}]]$. Existence of the Sato group element allows us to introduce a canonical pair of the Kac-Schwarz (KS) operators. Let ${\mathcal D}={\mathbb C}((z^{-1}))[[\frac{\p}{\p z}]]$.
\begin{definition}\label{defPQ}
For any ${\mathcal W}\in  \Gr^{(0)}_+$, we define a {\em canonical pair} of the KS operators associated to~${\mathcal W}$,
\be
\rho({\mathcal W}): = \left({\mathtt P}_{\mathcal W},{\mathtt Q}_{\mathcal W}\right) \in {\mathcal D}^{2},
\ee
where
\begin{equation}
\begin{split}
\label{PQdef}
{\mathtt P}_{\mathcal W} & :=  {\mathtt G}_{\mathcal W} \,\frac{\p}{\p z}\, {\mathtt G}_{\mathcal W}^{-1} ,\\
{\mathtt Q}_{\mathcal W} & :=   {\mathtt G}_{\mathcal W}\, z\,  {\mathtt G}_{\mathcal W}^{-1}. 
\end{split}
\end{equation}
\end{definition}
This pair of operators provides a complete description of the Sato Grassmannian.  Consider the space 
\be\label{GrD}
\Gr_{\mathcal D}:= \left\{\left({\mathtt P},{\mathtt Q}\right)\in {\mathcal D}^2 \Big| \left[{\mathtt P},{\mathtt Q}\right]=1, {\mathtt P}-\frac{\p}{\p z} \in z^{-1}  {\mathcal D}_-,{\mathtt Q}-z\in {\mathcal D}_-\right\}.
\ee
Then one of the main results of this paper is the following:
\begin{theorem}\label{bijecl}
The map $\rho$ describes  a bijection between $\Gr^{(0)}_+$ and $\Gr_{\mathcal D}$
\be
\Gr^{(0)}_+ \ni {\mathcal W} \mapsto 
\rho({\mathcal W})=\left({\mathtt P}_{\mathcal W},{\mathtt Q}_{\mathcal W}\right) \in \Gr_{\mathcal D}.
\ee
\end{theorem}
Operator ${\mathtt P}_{\mathcal W}$, by construction, always annihilates the wave function
\be
{\mathtt P}_{\mathcal W} \cdot \Psi=0,
\ee 
and it can be considered as the {\it quantum spectral curve} operator. Operator ${\mathtt Q}_{\mathcal W}$ is the raising operator, which generates the {\em distinguished basis} for a point of the Sato Grassmannian 
\be
{\mathcal W}=\sppan_{\cc}\{\Psi,{\mathtt Q}_{\mathcal W}\cdot \Psi,{\mathtt Q}_{\mathcal W}^2\cdot \Psi,\dots\}.
\ee


\subsection{Generalized Kontsevich model and its deformation}

GKMs were intensively investigated in the early nineties, in particular because of their relation to topological strings and intersection theory on the moduli spaces \cite{Kh,IZ,KMMMZ,KMMM,MMS,KM,AM,LG,F,KS,Sch,EY,Takasaki,W2}.  This family generalizes the Kontsevich model \cite{Konts}, which describes the intersection theory on the moduli spaces of Riemann surfaces with punctures, and helps to prove Witten's conjecture \cite{W}. Generalized Kontsevich model is defined by the asymptotic expansion of the $N\times N$ matrix integral
\be\label{mati}
Z_U(\Lambda):={\mathcal C}^{-1}\int[ d\Phi] \exp\left(-\frac{1}{\hbar}\Tr(V(\Phi)-\Phi V'(\Lambda))\right).
\ee
The model is labeled by the second derivative of the potential $V$,
\be
U(z):=\frac{1}{\hbar} V''(z).
\ee
It is easy to show that for any formal series $U(z)\in {\mathbb C}((z^{-1}))$ with non-trivial positive part, the asymptotic expansion of the matrix integral (\ref{mati}) defines a tau-function of the KP hierarchy in the Miwa parametrization. For this tau-function  we found explicitly the canonical pair of the KS operators (Lemma \ref{Qform}): 
\begin{equation}
\begin{split}
{\mathtt P}_U&=\frac{1}{\hbar}\left(V'({\mathtt Q}_U)-V'(z)\right),\\
{\mathtt Q}_U&=z+\frac{1}{U(z)}\frac{\p}{\p z} -\frac{U'(z)}{2 U(z)^2}.
\end{split}
\end{equation}
Operator ${\mathtt Q}_U$ coincides with the well-known KS operator for the GKM, and, for the monomial potential was described by Kac and Schwarz \cite{KS}. For polynomial $U$ operator ${\mathtt P}_U$ is closely related to another well known KS operator, 
\be
{\mathtt X}_U=V'(z),
\ee
namely ${\mathtt P}_U=\hbar^{-1}\left(V'({\mathtt Q}_U)-{\mathtt X}_U\right)$.

To extend the realm of GKM, for any polynomial potential 
\be
U_0=\frac{1}{\hbar}(z^{n}+b_{n-1}z^{n-1}+\dots+b_1 z)
\ee
we introduce its deformation with small parameters $w_1,\dots,w_{n+1}$:
\begin{definition}\label{DefPot}
The {\em deformed potential} is given by
\be\label{gendefor}
U=\frac{1}{\hbar}\frac{z^{n}+b_{n-1}z^{n-1}+\dots+b_1 z}{(1-w_1z)(1-w_2z)\dots(1-w_{n+1}z)} .
\ee
\end{definition}
Deformation of the matrix integral associated to the potential (\ref{gendefor}) is rather non-trivial. In particular, for the deformed GKM the matrix integral (\ref{mati}) describes a tau-function only as $N \rightarrow \infty$. However, using Sato Grassmannian description and canonical KS operators (\ref{PQdef}) one can investigate the deformed GKM and show that it always describes a tau-function of the KP hierarchy. 

The deformed GKM can be related to a pure one. Let $f(z)$ be a change of local parameter $z$ such that 
\be
V_0'(f(z))=V'(z),
\ee
with $f(z)\in z+ {\mathbb C}(z)[[{\bf w}]]$ and $\left.f(z)\right|_{{\bf w}={\bf 0}}=z$. It can be represented as
\be\label{ffun}
f(z)=e^{-\sum_{k \in {\z}} a_k z^{k+1}\frac{\p}{\p z} } \cdot z,
\ee
for some $a_k$. Let us also introduce the coefficients
\be\label{uk}
u_k=[z^k] \int^z (f(\eta)-\eta)U(\eta) d\eta.
\ee  
Consider the element of the Heisenberg-Virasoro subgroup of symmetry group of KP hierarchy
\be\label{defconjop}
\widehat{{ G}} = e^{\sum_{k \in {\z}}  u_k \widehat{J}_k}  e^{\sum_{k \in {\z}} a_k \widehat{L}_k }.
\ee
One of the main results of this paper describes a relation between pure and deformed GKMs.
\begin{theorem}\label{taufr}
\be
\tau_U= {C}_U\,  \widehat{{ G}} \cdot \tau_{U_0},
\ee
where ${C}_U$ does not depend on $\bf t$.
\end{theorem}
Using KS description, we derive a family of the Heisenberg-Virasoro constraints for the deformed GKM (Proposition \ref{lem_defvir}). Moreover, for the simplest deformation of the monomial potential, Proposition \ref{Theor1} provides an explicit form of the modified quantum spectral curve:
\be\label{pochqq}
e^{w^{n+1}\hat{x} +\sum_{k=1}^n\frac{w^k}{k}\hat{y}^k}\left(1-w\hat{y}-\frac{\chi w}{2}\right)\cdot \tilde{\Psi}=\tilde{\Psi}.
\ee


\subsection{Hodge integrals}\label{S_int_H}

The main motivation to consider the deformed GKM is its relation to the generating functions of Hodge integrals.
Mumford's profound  result \cite{M} allows us to describe Hodge integrals on the moduli spaces of Riemann surfaces in terms of the Kontsevich-Witten tau-function. This relation was extended to a more general context of the Gromov-Witten theory by Faber and Pandharipande \cite{FP}, and can be naturally described by an element of the Givental group \cite{Giv1,Giv2}.

Let $\overline{\cal M}_{g,n}$ be the Deligne-Mumford compactification of the moduli space of stable complex curves of genus $g$ with $n$ distinct marked points. We consider {\em Hodge integrals} 
\be\label{Corr}
\left<\lambda_{j_1}\lambda_{j_2}\cdots  \lambda_{j_k} \tau_{m_1}\tau_{m_2}\ldots\tau_{m_n} \right>_g=\int_{\overline{\cal M}_{g,n}}\lambda_{j_1}\lambda_{j_2}\cdots  \lambda_{j_k}\psi_1^{m_1}\psi_2^{m_2}\cdots\psi_n^{m_n} \in {\mathbb Q},
\ee
where  $\psi_i$ is the first Chern class of the line bundle corresponding to the cotangent space of the curve at the $i$-th marked point, and $\lambda_i$ is the $i$-th Chern class of the Hodge bundle ${\mathbb E}$. 
These integrals are trivial, unless the corresponding complex dimensions coincide
\be\label{dimc}
\sum_{l=1}^{k}j_l +\sum_{i=1}^n m_i =\dim_{\cc}\,\overline{\cal M}_{g,n},
\ee
where $\dim_{\cc}\,\overline{\cal M}_{g,n}=3g-3+n$. 
The moduli space $\overline{\cal M}_{g,n}$ is defined to be empty unless the stability condition
\be
2g - 2 + n > 0
\ee
is satisfied. Let 
\be
\Lambda_g(q)=\sum_{i=0}^g q^i \lambda_i
\ee
be the Chern polynomial of ${\mathbb E}$. By linearity we extend the correlation notation (\ref{Corr}) to $\Lambda_g$'s.

The simplest case -- linear Hodge integrals -- is of particular interest. Kazarian proved \cite{Kaza}, that the generating function of all linear Hodge integrals, after a simple linear transformation of variables, yields a tau-function of the KP integrable hierarchy aka Hodge tau-function. Since the $\GL(\infty)$ group acts freely on the space of solutions of the KP hierarchy, it follows that the Hodge tau-function is related to the Kontsevich-Witten tau-function by an element of $\GL(\infty)$. In \cite{A2} the author conjectured that the corresponding group element  belongs the Heisenberg-Virasoro subgroup of $\GL(\infty)$. This conjecture was proved in \cite{A1} and the group element (actually, an infinite dimensional family of equivalent group elements) was constructed explicitly up to a normalization. Moreover, matrix integral representations and the Heisenberg-Virasoro constraints for the Hodge tau-function were derived there. 

In \cite{Wang1} Liu and Wang prove the relation between two tau-functions by a different method. Namely, they find a direct relation between the corresponding elements of the Givental and $\GL(\infty)$ groups. This approach allow the authors to prove that the normalization of the $\GL(\infty)$ element is $1$, as conjectured in \cite{A1}.

In the companion paper \cite{H3_1} we consider a general relation between the Givental group of rank one and the Heisenberg-Virasoro subgroup of $\GL(\infty)$. We prove that a simple relation between the elements of two groups exists only for a two-dimensional family of the Givental operators. Moreover, this family of operators describes the generating functions of the triple Hodge integrals with the Calabi-Yau condition. In this paper we prove that this family of generating functions can be described by the deformed GKM.

Let us consider the case of cubic Hodge integrals 
\be
\< \Lambda_g (u_1) \Lambda_g (u_2) \Lambda_g (u_3) \tau_{m_1}\tau_{m_2}\ldots\tau_{m_n}\>_g=\int_{\overline{\cal M}_{g,n}} \Lambda_g (u_1) \Lambda_g (u_2) \Lambda_g (u_3) \psi_1^{m_1}\psi_2^{m_2}\cdots\psi_n^{m_n} 
\ee
with an additional Calabi-Yau condition
\be\label{specialcond}
\frac{1}{u_1}+\frac{1}{u_2}+\frac{1}{u_3}=0.
\ee
It is convenient to use the parametrization
\be\label{speccond}
u_1=-p, u_2=-q, u_3=\frac{pq}{p+q}.
\ee

Consider the generating function
\be
Z_{q,p}({\bf T})=e^{\sum_{g=0}^\infty \sum_{n=0}^\infty \hbar^{2g-2+n} {\mathcal F}_{g,n}},
\ee
where
\be
{\mathcal F}_{g,n}=\sum_{a_1,\dots,a_n}\frac{\prod T_{a_i}}{n!}\<  \Lambda_g (-q) \Lambda_g (-p) \Lambda_g (\frac{pq}{p+q}) \tau_{a_1}\tau_{a_2}\ldots\tau_{a_n}\>_g.
\ee
Let the change of variables be described by the recursion
\be\label{changeof}
T_0^{q,p}({\bf t})=t_1,\\
T_k^{q,p}({\bf t})=\left(q  \sum_{k=1}^\infty k t_k \frac{\p}{\p t_{k}}+\frac{2q+p}{\sqrt{p+q}} \sum_{k=1}^\infty k t_k \frac{\p}{\p t_{k-1}}+ \sum_{k=1}^\infty k t_k \frac{\p}{\p t_{k-2}}\right)T_{k-1}^{q,p}({\bf t}).
\ee
In the companion paper \cite{H3_1} using the identification of Givental and $\GL(\infty)$ group elements we prove that
\be
\tau_{q,p}({\bf t})=Z_{q,p}({{\bf T}^{q,p}(\bf t)})
\ee
is a tau-function of the KP hierarchy. 

Let us denote
\begin{equation}
\begin{split}
x(z)&=\frac{p+q}{pq}\log\left(1+\frac{qz}{\sqrt{p+q}}\right)-\frac{1}{p}\log(1+\sqrt{p+q}z),\\
y(z) &=\frac {\sqrt {p+q}}{p}   \left( \log  \left( 1+\sqrt {p+q}z \right)  -\log  \left(1+\frac{qz}{ \sqrt {p+q}} \right) \right)\\
\end{split}
\end{equation}
and consider the coefficients
\be
\tilde{v}_k=[z^k]\int_{0}^z (\eta-y(\eta)d x(\eta).
\ee
In this paper we prove a relation between tau-function $\tau_{q,p}({\bf t})$ and the deformed GKM with potential
\be\label{pqpot}
U=\frac{1}{\hbar}\frac{z}{(1+\sqrt{p+q}z)(1+q z/\sqrt{p+q})}.
\ee
These two tau-functions are related by a translation of variables.
\begin{theorem}\label{H3U}
Deformed GKM with potential (\ref{pqpot}) is related to the generating function of the cubic Hodge integrals with the Calabi-Yau condition by a shift of variables
\be
\tau_{q,p}({\bf t})=\tilde{C} e^{\sum \tilde{v}_k \frac{\p}{\p t_k}}\tau_{U}({\bf t}).
\ee
Here $\tilde{C}$ does not depend on ${\bf t}$.
\end{theorem}

GKM is believed to provide a universal description of the intersection theory on the moduli spaces and 2d gravity. This
paper provides another confirmation of the old idea of universality of GKM.


\subsection{Organization of the paper}

The present paper is organized as follows. In Section \ref{KSsec} we describe the space of the tau-functions of the KP hierarchy in terms of Sato Grassmannian. Using Sato's theorem we introduce a canonical pair of the Kac-Schwarz operators and prove that the space of such pairs provides an alternative description of the Sato Grassmannian. Section \ref{S:GKM} is devoted to the description of generalized Kontsevich model in the context of Sato Grassmannian. In Section \ref{Secdef} we introduce and investigate a new class of the deformed generalized Kontsevich models. In Section \ref{Sec:Hodge} we describe a relation of this deformation with the Hodge integrals. 


\subsection*{Notation}

By $\cdot$ we denote the action of the operator on the function to distinguish it from the product of operators. We denote the sets of parameters or variables, finite or infinite, by bold letters, for example ${\bf t}=\{t_1,t_2,t_3,\dots\}$. We will slightly abuse the notation and call $\GL(\infty)$ both versions of the group with and without central extension.


\subsection*{Acknowledgments}
The author is grateful to A. Mironov and S. Shadrin for useful discussions.  This work was supported by  IBS-R003-D1 and by RFBR grant 18-01-00926. 


\section{KP hierarchy, Kac-Schwarz algebra and quantum spectral curve}\label{KSsec}

In this section we give a short reminder of some basic properties of the KP hierarchy, which we will apply below. We also describe some less known or new elements of the general construction, including Sato's group element, a canonical pair of the Kac-Schwarz operators and a distinguished basis.

\subsection{KP hierarchy and Sato Grassmannian}

Let us briefly summarize by now standard Sato Grassmannian and Lax-Orlov-Schulman descriptions of the KP hierarchy; for more detail see, e.g., \cite{Sato,Sato1,Babelon,MJ,Segal}
 and references therein.

The KP hierarchy was introduced by Sato \cite{Sato}. It
can be represented in terms of tau-function $\tau({\bf t})$ by the Hirota bilinear identity
\be\label{HBE}
\oint_{\infty} e^{\xi({\bf t-t'},z)}
\tau ({\bf t}-[z^{-1}])\tau ({\bf t'}+[z^{-1}])dz =0,
\ee
which encodes all nonlinear equations of the KP hierarchy.
Here we use the standard short-hand notations
\be
{\bf t}\pm [z^{-1}]:= \bigl \{ t_1\pm   
z^{-1}, t_2\pm \frac{1}{2}z^{-2}, 
t_3 \pm \frac{1}{3}z^{-3}, \ldots \bigr \}
\ee
and
\be
\xi({\bf t},z)=\sum_{k>0} t_k z^k.
\ee

Let us consider the description of the space of solutions for the KP hierarchy, introduced by Sato in \cite{Sato} and further developed by Segal and Wilson in \cite{Segal}. We work within the formal series setup, $\tau({\bf t})\in {\mathbb C}[[t_1,t_2,t_3,\dots]]$. Hence,
we focus on Sato's version of the construction. Let us consider the space $H=H_+\oplus H_-$, where the subspaces 
\be
H_-=z^{-1}{\mathbb C}[[z^{-1}]]
\ee and 
\be
H_+={\mathbb C}[z]
\ee 
are generated by negative and nonnegative powers of $z$ respectively. Then the Sato Grassmannian $\rm{Gr}$ consists of all closed linear spaces $\mathcal{W}\in H $, which are compatible with $H_+$. Namely, an orthogonal projection $\pi_+ : \mathcal{W} \to H_+ $ should be a Fredholm operator, i.e. both the kernel ${\rm ker}\, \pi_+ \in \mathcal{W}$ and the 
cokernel ${\rm coker}\, \pi_+ \in H_+$ should be finite-dimensional vector spaces. 
The Grassmannian $\Gr$ consists of components $\Gr^{(k)}$, parametrized by an index of the operator $\pi_+$. We need only the component $\Gr^{(0)}$; other components have an equivalent description. Moreover, we will consider only the big cell $\Gr^{(0)}_+$ of $\Gr^{(0)}$, which is defined by the constraint ${\rm ker}\, \pi_+ = {\rm coker}\, \pi_+=0$. We call $\Gr^{(0)}_+$ the {\em Sato Grassmannian} for simplicity. Most of the analyses in this section can be easily extended to a general case. There exists a bijection between the points of the Sato Grassmannian ${\mathcal W}\in\Gr^{(0)}_+$ and the tau-functions with $\tau({\bf 0})=1$.  Below we put $\tau({\bf 0})=1$ for simplicity.

A point of the Sato Grassmannian ${\mathcal W}\in \rm{Gr}^{(0)}_+$ can be described by an  {\emph {admissible basis}} $\{\Phi_1^{\mathcal W},\Phi_2^{\mathcal W},\Phi_3^{\mathcal W},\dots\}$,
\be
{\mathcal W}=\sppan_{\cc}\{\Phi_1^{\mathcal W},\Phi_2^{\mathcal W},\Phi_3^{\mathcal W},\dots\}.
\ee
The crucial property of the admissible bases is that if $\{\Phi_j^{\mathcal W}\}$ and $\{ {\Phi'}_j ^{\mathcal W}\}$ are two admissible bases of ${\mathcal W}$, then the matrix which relates them is of the kind that has a determinant, or, equivalently, this matrix differs from the identity by an operator of trace class \cite{Segal}:

\begin{definition}\label{Defadm} 
 $\{\Phi_j^{\mathcal W}\}$ is an admissible basis for ${\mathcal W}\in \rm{Gr}^{(0)}_+$,  if
\begin{enumerate}
\item the linear map $H_+ \rightarrow H$ which takes $z^{j-1}$ to $\Phi_j^{\mathcal W}$ is injective and has image $\mathcal W$, and
\item the matrix, relating $\pi_+ (\Phi_j^{\mathcal W})$ to $z^{j-1}$ differs from the identity by an operator of trace class.  
\end{enumerate}
\end{definition}

\begin{remark}
We use conventions, inverse to standard (see i.e. \cite{Babelon}) and call the point of the Sato Grassmannian, what is usually called the {\em dual} (or adjoint) point of the Sato Grassmannian
and vice verse. So, all objects we consider, including Baker-Akhiezer function, Lax operator and Kac-Schwarz operators are usually called the dual ones. We apologize for any possible inconvenience. 
\end{remark}

We call the element of $H$ {\em monic} if its leading coefficient is equal to 1. Any point of the Sato Grassmannian has an admissible basis of the monic elements of the form
\be\label{goodbas}
\Phi_j^{\mathcal W}=z^{j-1}\left(1+O(z^{-1})\right);
\ee
of course, such basis is not unique. Let $\Lambda =\diag (\lambda_1,\dots,\lambda_N)$ be a diagonal matrix.
For any function $f$, dependent on the infinite set of variables ${\bf t}=(t_1,t_2,t_3,\dots)$, let
\be\label{Miwa}
f\left(\left[\Lambda^{-1}\right]\right):=f({\bf t})\Big|_{t_k=\frac{1}{k}\Tr \Lambda^{-k}}
\ee
be the {\em Miwa parametrization}. For any basis (\ref{goodbas}) the tau-function of the KP hierarchy in the Miwa parametrization is equal to the ratio of the determinants
\be\label{miwatau}
\tau_{\mathcal W}([\Lambda^{-1}])=\frac{\det_{i,j=1}^N \Phi^{\mathcal W}_i(\lambda_j)}{\Delta(\lambda)},
\ee
where 
\be
\Delta(\lambda):=\prod_{i<j} (\lambda_j-\lambda_i)
\ee
is the {\em Vandermonde determinant}. 
Moreover, if for some function $\tau_{\mathcal W}$ equation (\ref{miwatau}) holds for all $N \in {\mathbb Z}_{\geq 0}$, then  $\tau_{\mathcal W}$ is a tau-function of the KP hierarchy.

The Hirota bilinear identity (\ref{HBE}) can be reformulated as the orthogonality condition of the {\em Baker-Akhiezer} (BA) function 
\be
\label{BA}
\Psi(z,{\bf t}):=e^{- \xi({\bf t},z)}\frac{\tau_{\mathcal W}({\bf t}+ [z^{-1}])}{\tau_{\mathcal W}({\bf t})}
\ee
and its dual 
\be
\Psi^{\bot}(z,{\bf t}):=e^{ \xi({\bf t},z)}\frac{\tau_{\mathcal W}({\bf t}- [z^{-1}])}{\tau_{\mathcal W}({\bf t})},
\ee
namely
\be\label{HBI}
\oint_{\infty} \Psi (z,{\bf t}) \Psi^{\bot}(z,{\bf t'})\, dz =0.
\ee
Functions $\Psi$ and $\Psi^\bot$ give equivalent description of the solution of the KP hierarchy.

To distinguish the special role played by $t_1$, in this section we identify it with $x$, $x\equiv t_1$  when necessary, and use one or another notation interchangeably.  The KP hierarchy and Sato Grassmannian can be naturally described by the pseudodifferential operators
\be
a_0\p^n+a_1 \p^{n-1}+\dots + a_{n+1}\p^{-1} +a_{n+2}\p^{-2}+\dots, 
\ee
where $\p:=\frac{\p}{\p x}$.
Such operator is called {\em monic} if the leading coefficient $a_0=1$. For an algebra of pseudodifferential operators there is a natural anti-homomorphism, given by the formal adjoint
\be\label{adjoint}
(x^k \p^m)^*:=(-\p)^m x^k
\ee
and linearity, such that
\be\label{adjoint1}
\oint_\infty f(x) o\cdot g(x) dx = \oint_\infty g(x) o^* \cdot f(x) dx
\ee
for any functions $f$ and $g$ and arbitrary operator $o$. 

Consider the {\em Lax operator} 
\begin{align}
L_+({\bf t})&=\p+ u_2({\bf t})\p^{-1}+u_3({\bf t}) \p^{-2}+\dots
\end{align}
and its dual $L_-({\bf t})$.
This is a first order pseudodifferential operator such that
\begin{equation}
\begin{split}
L_{+}({\bf t}) \cdot \Psi(z,{\bf t})  &= z\, \Psi (z,{\bf t}),\\
L_{-}({\bf t}) \cdot \Psi^\bot(z,{\bf t})  &= z\, \Psi^{\bot}(z,{\bf t}).
\end{split}
\end{equation}	
The Lax operator can be expressed in terms of the {\em dressing operator} $S({\bf t})$,
\begin{equation}
\begin{split}
L_+({\bf t}) &=- S^*({\bf t})^{-1} \,\p \, S^*({\bf t}),\\
L_-({\bf t}) &=S({\bf t}) \,\p \, S({\bf t})^{-1},
\end{split}
\end{equation}
where $S({\bf t})$
is a monic pseudodifferential operator of order $0$.
The BA function and its dual can also be expressed in terms of the dressing operator
\begin{equation}
\begin{split}
  \Psi(z,{\bf t}) & =S^*({\bf t})^{-1} \cdot e^{-\xi({\bf t},z)},\\
 \Psi^\bot(z,{\bf t}) & =S({\bf t}) \cdot e^{\xi({\bf t},z)}.
\end{split}
\end{equation}
Using the operator $S({\bf t})$ one can also introduce the  {\em Orlov-Schulman operator} and its dual \cite{OS,Grin}:
\begin{equation}
\begin{split}
M_+({\bf t})&= -S^*({\bf t})^{-1}\left(x +\sum_{k=2}^\infty k t_k (-\p)^{k-1} \right) {S^*({\bf t})},\\
M_-({\bf t})&= S({\bf t})\left(x +\sum_{k=2}^\infty k t_k \p^{k-1} \right) {S({\bf t})}^{-1},
\end{split}
\end{equation}
such that
\begin{equation}
\begin{split}
M_{+}({\bf t}) \cdot \Psi(z,{\bf t})&=\frac{\p}{\p z}\, \Psi(z,{\bf t}),\\
M_{-}({\bf t}) \cdot \Psi^\bot(z,{\bf t})&=\frac{\p}{\p z}\, \Psi^\bot(z,{\bf t}).
\end{split}
\end{equation}
The Lax and Orlov-Schulman operators satisfy the commutation relation $[L_{\pm}({\bf t}),M_{\pm}({\bf t})]=1$.

\begin{remark}
Let us consider the tau-function dependent of the times with the opposite sign, $\tau_{\mathcal W}(-{\bf t})$. Then, for such tau-function, the BA function is closely related to the dual BA function of the original tau-function, and vice versa.
\end{remark}


\subsection{Sato's theorem and Sato's group elements}

A BA function $\Psi(z,{\bf t})$ defines a point of the Sato Grassmannian.
This point can  be described in terms of it specialization, which we also call the Baker-Akhiezer  function 
\be\label{BAspec}
\left.\Psi(z,x):=\Psi (z,{\bf t})\right|_{t_k=\delta_{k,1}x},
\ee
namely
\be\label{Grser}
\Gr^{(0)}_+\ni {\mathcal W}=\sppan_{\cc}\{\Psi(z,0), -\left.\frac{\p}{\p x}\Psi(z,x)\right|_{x=0},\dots\}.
\ee
We also introduce its dual
\be
\left.\Psi^{\bot}(z,x):=\Psi^\bot(z,{\bf t})\right|_{t_k=\delta_{k,1}x},
\ee
with
\be
\Gr^{(0)}_+\ni {\mathcal W}^{\bot}=\sppan_{\cc}\{\Psi^{\bot}(z,0), \left.\frac{\p}{\p x}\Psi^{\bot}(z,x)\right|_{x=0},\dots\}.
\ee

The {\em wave function}
\be
\Psi(z):=\Psi(z,0)
\ee
plays a special role. It is a unique element of ${\mathcal W}\cap 1+z^{-1}{\mathbb C}[[z^{-1}]]$.
The wave function equals to the {\em principal specialization} of the tau-function $\Psi(z)=\tau_{\mathcal W}([z^{-1}])$.

Let  $\left.L:=L_+({\bf t})\right|_{t_k=\delta_{k,1}x}$ be a reduction of the Lax operator.
This operator can be obtained by a conjugation of $-\p$ with the dressing operator
\be\label{pseudodif}
\Gamma=1+\sum_{k=1}^\infty  s_k(x) (-\p)^{-k},
\ee
where  $\left.\Gamma:=S^*({\bf t})^{-1}\right|_{t_k=\delta_{k,1}x}$ and $s_k(x)\in {\mathbb C}[[x]]$. Namely
\be\label{Lax}
L= - \Gamma\,  \p \, \Gamma^{-1}.
\ee
The BA function (\ref{BAspec}) obeys
\be
\Psi(z,x) =\Gamma\cdot e^{-xz}=e^{-xz}\,\left(1+ \sum_{k=1}^\infty s_k(x)z^{-k}\right).
\ee
The reduced version of the Orlov-Schulman operator $\left.M:=M_+({\bf t})\right|_{t_k=\delta_{k,1}x} $ can be obtained by a conjugation of $x$:
\be\label{OS}
M= -\Gamma \,x\,{\Gamma}^{-1}.
\ee

Let us outline a description of the KP hierarchy, obtained by a Fourier transform of this Lax-Orlov-Schulman construction. 
Consider the ring of differential operators with coefficients formal Laurent series in the variable $z^{-1}$
\be
{\mathcal D}:={\mathbb C}((z^{-1}))[[\frac{\p}{\p z}]]
\ee
and its subrings ${\mathcal D}_\pm:=H_\pm[[\frac{\p}{\p z}]]$. 
A natural direct sum decomposition holds
\be
{\mathcal D}={\mathcal D}_+ \oplus{\mathcal D}_-.
\ee

\begin{remark}
Below we will allow the coefficients of operators in ${\mathcal D}$ to depend  also on parameters including $\hbar$, $w$ etc. We will not specify this unless necessary.
\end{remark}

Ring ${\mathcal D}$ possess a filtration
\be
\dots\supset  {\mathcal D}^{(n+1)}  \supset {\mathcal D}^{(n)} \supset  {\mathcal D}^{(n-1)} \supset\dots,
\ee
where
\be
 {\mathcal D}^{(n)}=\left\{ \sum a_{km} z^{k} \frac{\p^m}{\p z^m} \, |\,  a_{km}=0 \,\,{\text{for all}} \,\, k-m>n\right\}.
\ee
We use the typewriter font for elements of ${\mathcal D}$.
We say ${\mathtt A} \in {\mathcal D}$ has {\em degree} $n$ if ${\mathtt A} \in {\mathcal D}^{(n)} \setminus {\mathcal D}^{(n-1)}$. 
In particular, 
\be
\deg\,z^{-1}=\deg\, \frac{\p}{\p z}=-1,
\ee
so that the degrees of all operators in $\mathcal D$ are restricted from above.

Let 
\be
 {\mathcal G}:= \left\{ {\mathtt G}\in {\mathcal D} \big| {\mathtt G}-1\in  {\mathcal D}_-\right\}.
\ee
be the group of monic operators of order $0$,
\be
{\mathtt G}=1+\sum_{k=0}^\infty a_k(z^{-1})\frac{\p^k}{\p z^k}.
\ee 
There is a bijection between ${\mathcal G}$ and the space of monic pseudodifferential operators (\ref{pseudodif}) of order $0$, $\Gamma \leftrightarrow {\mathtt G}
$ such that the coefficients of $\Gamma$ and $\mathtt G$ are related by  $1+\sum_{k=1}^\infty s_k(x) z^{-k} =1+  \sum_{k=0}^\infty a_k (z^{-1}) (-x)^k$.
Under this bijection
\be
 {\mathtt G} \cdot e^{-xz} = \Gamma \cdot e^{-xz}.
\ee

This bijection allows us to provide the following version of the important {\em Sato's theorem} \cite{Sato1,Motohico}, which describes the Sato Grassmannian:
\begin{theorem*}[Sato]\label{TheoremS}
There is a bijection between the Sato Grassmanninan $\Gr^{(0)}_+$  and the group ${\mathcal G}$, 
\begin{equation}
\begin{split}
{\mathcal G}& \rightarrow  \Gr^{(0)}_+\\
{\mathcal G} \ni {\mathtt G}_{\mathcal {W}} \, & \mapsto \,  \mathcal{W} = {\mathtt G}_{\mathcal W} \cdot H_+ \in  \Gr^{(0)}_+.
\end{split}
\end{equation}
\end{theorem*}

Let ${\mathtt G}_{\mathcal W} \in {\mathcal G}$ be the group element, corresponding to ${\mathcal W} \in \Gr^{(0)}_+$. We call it the {\em Sato's group element}. This element completely describes the point of the Sato Grassmannian. 
In particular, the BA function and dual BA function obey
\begin{equation}
\begin{split}\label{psiasop}
\Psi(z,x)&={\mathtt G}_{\mathcal W} \cdot e^{-xz},\\
\Psi^{\bot}(z,x)&={({\mathtt G}_{\mathcal W}^*)}^{-1} \cdot  e^{xz},
\end{split}
\end{equation}
where the adjoint operator in $z$-variable is defined in the same way as the adjoint operator in $x$, see (\ref{adjoint}).
These formulas immediately allow us to show  that the Hirota bilinear identity (\ref{HBI}) holds at $t_k=t'_k=0$ for $k\geq 2$
\begin{equation}
\begin{split}
\oint_{\infty} \Psi (z,x) \Psi^{\bot}(z,x')\, dz
&=\oint_{\infty} ({\mathtt G}_{\mathcal W} \cdot e^{-xz}) ({({\mathtt G}_{\mathcal W}^*)}^{-1}  \cdot e^{x'z})\, dz\\
&=\oint_{\infty} ({\mathtt G}_{\mathcal W}^{-1} {\mathtt G}_{\mathcal W} \cdot e^{-xz})   e^{x'z}\, dz\\
&=\oint_{\infty}  e^{(x'-x)z}   \, dz\\
&=0,
\end{split}
\end{equation}
where we have used (\ref{adjoint1}) and the identity $({\mathtt G}^*)^{-1}=({\mathtt G}^{-1})^*$.

\begin{remark}\label{rmk_ferm}
In this paper we do not discuss free fermion description of the  KP hierarchy in detail. However, let us comment on the relation of this description and Sato's group element. For any ${\mathtt G}_{\mathcal W}=:e^{{\mathtt g}_{\mathcal W}}\in {\mathcal G}$ we consider a fermionic operator 
\be
W_{\mathcal W}:=\res_z \left(z^{-1} \psi(z) z {{\mathtt g}_{\mathcal W}} z^{-1} \psi^*(z)\right).
\ee
Then the tau-function, corresponding to ${\mathtt G}_{\mathcal W}$, is $\tau_{\mathcal W}({\bf t})=\lvac e^{J_+({\bf t})} e^{W_{\mathcal W}}\rvac$.
Here we use notations of \cite{A1}. The matrix, describing this bilinear fermionic operator $W_{\mathcal W}$, is strictly upper diagonal.
\end{remark}


\subsection{Canonical Kac-Schwarz operators and quantum spectral curve}

It is well known that the tau-functions of the KP hierarchy that appear in mathematical physics and enumerative geometry often have a nice description in terms of the so-called Kac-Schwarz operators \cite{KS}.  Unfortunately, the general theory of the Kac-Schwarz algebras is not fully developed yet.  However, some elements of this theory and many interesting examples are considered in \cite{AMSM,Sch,ALS,AM,ABGW,KS,Kh,KMMMZ,KMMM,LG,MMS,ACM,A1,Aopen2,AP1,F,ACEH}.

In this section we construct a new description of the Sato Grassmannian in terms of a pair of canonical Kac-Schwarz (KS) operators. We argue that one of these operators can be considered as a quantum spectral curve operator.

\begin{definition}
For any point of the Sato Grassmannian ${\mathcal W}$ the {\em Kac-Schwarz algebra}
\be
{\mathcal A}_{\mathcal W}:=\left\{{\mathtt A}\in {\mathcal D}\,\big | \,{\mathtt A}\cdot {\mathcal W} \subset {\mathcal W} \right\}
\ee
 is the algebra of the differential operators which stabilize this point.
\end{definition}

\begin{remark}
Sometimes it may be convenient to consider more general classes of KS operators. For example, the operators from ${\mathbb C}((z^{-1}))[[z\frac{\p}{\p z}]]$ appear naturally in description of Hurwitz tau-functions \cite{A1,ACEH}.
\end{remark}

Below we use the lemma of Sato and Noumi:
\begin{lemma}[\cite{Motohico,Sato1}]\label{LemmaAHP}
An operator ${\mathtt A} \in {\mathcal D}$ belongs to ${\mathcal D}_+$ if and only if it stabilizes $H_+$, or
\be
{\mathcal A}_{H_+}={\mathcal D}_+.
\ee
\end{lemma}

The idea of the proof is the following. It is easy to see that ${\mathcal D}_+$ stabilizes $H_+$. Then, if some operator in  ${\mathcal D}$ 
stabilizes $H_+$, than its projection to ${\mathcal D}_+$ also does, thus, its projection to  ${\mathcal D}_-$ also should stabilizes $H_+$. But this is possible only if this projection is trivial. It  follows from the consideration of the action of the top degree term of this projection on $H_+$.

\begin{corollary}\label{mnegat}
For any ${\mathcal W} \in  \Gr^{(0)}_+$ there are no non-trivial KS operators in ${\mathcal D}_-$, ${\mathcal A}_{\mathcal W} \cap {\mathcal D}_- =0$.
\end{corollary}
\begin{proof}
Assume ${\mathtt A}\in {\mathcal D}_-$ is a KS operator for some ${\mathcal W} \in  \Gr^{(0)}_+$. Then ${\mathtt G}_{\mathcal W}^{-1} {\mathtt A} {\mathtt G}_{\mathcal W}$  is a KS operator for $H_+ \in  \Gr^{(0)}_+ $. But ${\mathtt G}_{\mathcal W}^{-1} {\mathtt A} {\mathtt G}_{\mathcal W} \in {\mathcal D}_-$, so it vanishes. Hence ${\mathtt A}=0$.
\end{proof}

From Definition \ref{defPQ} of operators ${\mathtt P}_{\mathcal W}$ and ${\mathtt Q}_{\mathcal W}$ it immediately follows that for any  ${\mathcal W} \in  \Gr^{(0)}_+$ these operators belong to ${\mathcal A}_{\mathcal W}$. These operators will play an important role in our consideration. By construction, 
\be\label{PQop}
({\mathtt P}_{\mathcal W},{\mathtt Q}_{\mathcal W}) \in \Gr_{\mathcal D}.
\ee 
where $\Gr_{\mathcal D}$ is given by (\ref{GrD}).

Any two points ${\mathcal W}$ and ${\mathcal W}'$ of the Sato Grassmannian are related by some operator ${\mathtt G}_{{\mathcal W} {\mathcal W}'} \in {\mathcal G}$, 
\be\label{enoth}
{\mathcal W}'={\mathtt G}_{{\mathcal W} {\mathcal W}'} \cdot {\mathcal W}.
\ee
Such an operator is unique,  ${\mathtt G}_{{\mathcal W} {\mathcal W}'}={\mathtt G}_{ {\mathcal W}'}{\mathtt G}_{{\mathcal W}}^{-1}$, or
\be\label{GGcon}
{\mathtt G}_{{\mathcal W}'}={\mathtt G}_{{\mathcal W} {\mathcal W}'} \, {\mathtt G}_{{\mathcal W}}.
\ee
The canonical pair of the KS operators for these two points are related to each other by the conjugation
\be\label{GGcon1}
\rho({\mathcal W}')= {\mathtt G}_{{\mathcal W} {\mathcal W}'} \,\rho({\mathcal W})\, {\mathtt G}_{{\mathcal W} {\mathcal W}'}^{-1}.
\ee

Operators $({\mathtt P}_{\mathcal W},{\mathtt Q}_{\mathcal W})$ can be considered as the counterparts of the Lax (\ref{Lax}) and Orlov-Schulman (\ref{OS}) operators on the Sato Grassmannian side
\begin{equation}
\begin{split}
{\mathtt P}_{\mathcal W} \cdot \Psi(z,x) & = -x\, \Psi(z,x),\\
{\mathtt Q}_{\mathcal W} \cdot \Psi(z,x) & = -\frac {\p}{\p x} \,\Psi(z,x).
\end{split}
\end{equation}
At $x=0$ the first equation reduces to
\be\label{QSCeq}
{\mathtt P}_{\mathcal W}\cdot  \Psi=0.
\ee
We call it the {\em quantum spectral curve} equation associated with ${\mathcal W}$. We also call ${\mathtt P}_{\mathcal W}$ the {\em quantum spectral curve operator}. 

\begin{lemma}\label{QSCLem}
For any ${\mathtt P}\in \frac{\p}{\p z}+z^{-1}{\mathcal D}_-$ equation 
\be\label{likeqsc}
{\mathtt P}\cdot \Psi=0
\ee
has a unique monic solution in $H$. This solution is of the form $\Psi=1+O(z^{-1})$. Hence, the quantum spectral curve equation (\ref{QSCeq}) uniquely specifies the wave function.
\end{lemma}
\begin{proof}
For any operator ${\mathtt P}\in \frac{\p}{\p z}+z^{-1}{\mathcal D}_-$ the equation (\ref{likeqsc}) has a non-trivial solution of the form $\sum_{k=-\infty}^{k_0} a_k z^k$ for some finite $k_0$ if and only if $a_0\neq 0$ and $a_k=0$ for positive $k$. Then all coefficients $a_k$ for negative $k$ can be obtained recursively in $k$. 
\end{proof}

\begin{conjecture}
We expect, that this definition of the quantum spectral curve is consistent with other concepts in the literature, see for example \cite{DHS,BE,GS,Sch2}. In particular, when the tau-function has a nice genus or topological expansion, we expect that a semi-classical limit of the quantum spectral curve equation (\ref{QSCeq}) after its conjugation with the unstable contributions (for example, see Section \ref{S_topex}) should lead to the classical spectral curve. Let us stress that in our approach the classical spectral curve appears as a dequantization of the quantum spectral curve, so the latter is more fundamental.  
\end{conjecture}
KS algebra should be useful for the investigation of the classical spectral curve, WKB expansion of the basis vectors and correlation functions and, as a result, construction of topological recursion in a way similar to  \cite{ACEH,ACEH1}. We expect, that the quantum and classical spectral curves introduced in this paper should describe the topological recursion.

From Corollary  \ref{mnegat} we have
\begin{lemma}\label{uniquecan}
For any ${\mathcal W} \in \Gr^{(0)}_+$ the canonical KS operators $({\mathtt P}_{\mathcal W},{\mathtt Q}_{\mathcal W})$ provide a unique pair of the KS operators belonging to $\Gr_{\mathcal D}$,
\be
\Gr_{\mathcal D}\cap {\mathcal A}_{\mathcal W}^2= ({\mathtt P}_{\mathcal W},{\mathtt Q}_{\mathcal W}).
\ee
\end{lemma}
\begin{proof}
Assume that for some ${\mathcal W} \in \Gr^{(0)}_+$ there is another pair $ (\tilde{\mathtt P}_{\mathcal W},\tilde{\mathtt Q}_{\mathcal W}) \in \Gr_{\mathcal D}$ which stabilizes ${\mathcal W}$, and  $ (\tilde{\mathtt P}_{\mathcal W},\tilde{\mathtt Q}_{\mathcal W})\neq  ({\mathtt P}_{\mathcal W},{\mathtt Q}_{\mathcal W})$. Then operator $\tilde{\mathtt P}_{\mathcal W} - {\mathtt P}_{\mathcal W}\in {\mathcal D}_-$ is a KS operator, so by Corollary \ref{mnegat} it vanishes. The same argument shows that $\tilde{\mathtt Q}_{\mathcal W}={\mathtt Q}_{\mathcal W}$, which contradicts the assumption. This completes the proof.
\end{proof}

\begin{proof}[{\bf Proof of Theorem \ref{bijecl}}]
Definition \ref{defPQ} describes $\rho({\mathcal W})=\left({\mathtt P}_{\mathcal W},{\mathtt Q}_{\mathcal W}\right) \in \Gr_{\mathcal D}$ for any ${\mathcal W}\in  \Gr^{(0)}_+ $. 
Let us describe $\rho^{-1}: \Gr_{\mathcal D}  \rightarrow \Gr^{(0)}_+$. Namely, for some $\left({\mathtt P},{\mathtt Q}\right) \in  \Gr_{\mathcal D} $ consider equation (\ref{likeqsc}).
From Lemma \ref{QSCLem} this equation has a unique monic solution $\Psi$ in $H$. Then 
\be
\Gr^{(0)}_+ \ni \rho^{-1}\left({\mathtt P},{\mathtt Q}\right) := \sppan_{\cc}\{\Psi, {\mathtt Q}\cdot\Psi, {\mathtt Q}^2 \cdot \Psi,\dots \}.
\ee
By construction ${\mathtt P}$ and ${\mathtt Q}$ are the KS operators for $\rho^{-1}\left({\mathtt P},{\mathtt Q}\right) $. From Lemma \ref{uniquecan} 
such pair is unique, thus $\rho \rho^{-1} = \id$ on $\Gr_{\mathcal D}$. 

Similarly, if $\rho({\mathcal W})=\left({\mathtt P}_{\mathcal W},{\mathtt Q}_{\mathcal W}\right)$, then the equation for the wave function ${\mathtt P}_{\mathcal W}\cdot \Psi=0$ defines a wave function. Then ${\mathcal W}=\sppan_{\cc}\{\Psi, {\mathtt Q}_{\mathcal W}\cdot\Psi, {\mathtt Q}_{\mathcal W}^2 \cdot \Psi,\dots \}$, hence  $\rho^{-1} \rho = \id$ on $\Gr^{(0)}_+$. This completes the proof. 
\end{proof}

From Lemma \ref{LemmaAHP} it follows that
\begin{proposition}
For any ${\mathcal W} \in  \Gr^{(0)}_+$ the canonical pair of the KS operators generate the KS algebra
\be
{\mathcal A}_{\mathcal W}={\mathbb C}[{\mathtt Q}_{\mathcal W}][[{\mathtt P}_{\mathcal W}]].
\ee
\end{proposition}
\begin{proof}
Let ${\mathtt A} \in {\mathcal A}_{\mathcal W}$. Then  ${\mathtt G}_{\mathcal W}^{-1} {\mathtt A} {\mathtt G}_{\mathcal W} \in {\mathcal D}_+$, hence ${\mathtt A} \in {\mathtt G}_{\mathcal W}  {\mathcal D}_+ {\mathtt G}_{\mathcal W}^{-1}= {\mathbb C}[{\mathtt Q}_{\mathcal W}][[{\mathtt P}_{\mathcal W}]]$.
\end{proof}

A question arises: When does the KS algebra ${\mathcal A}_{\mathcal W}$ contain the differential operators of finite order? In such cases corresponding operators, acting on the tau-functions (see Section \ref{Virsec}), are finite order differential operators in $t_k$.  In some cases the operator ${\mathtt P}_{\mathcal W}$ or ${\mathtt Q}_{\mathcal W}$  (or both) is a differential operator of finite order. A rich class of such tau-functions associated to the generalized Kontsevich model will be discussed in Sections \ref{S:GKM} and \ref{Secdef}. Unfortunately, we do not know a general answer yet. 

\begin{remark} It is also interesting to consider a more restrictive subclass of the KS operators, namely ${\mathbb C}[z,z^{-1},\frac{\p}{\p z}] \cap  {\mathcal A}_{\mathcal W}$.
This subalgebra is directly related to the polynomial solutions for the Douglas string equation \cite{Douglas}. 
\end{remark}

For the dual point of the Sato Grassmannian, ${\mathcal W}^{\bot}$, the canonical KS operators are
\begin{equation}
\begin{split}\label{PQdual}
{\mathtt P}_{{\mathcal W}^\bot} & =  ({\mathtt G}_{\mathcal W}^{-1})^* \,\frac{\p}{\p z}\, {\mathtt G}_{\mathcal W}^{*} =-{\mathtt P}_{\mathcal W}^* ,\\
{\mathtt Q}_{{\mathcal W}^\bot} & =  ({\mathtt G}_{\mathcal W}^{-1})^*\, z\,  {\mathtt G}_{\mathcal W}^{*}={\mathtt Q}_{\mathcal W}^*, 
\end{split}
\end{equation}
hence \cite{ACM}
\be
{\mathcal A}_{{\mathcal W}^{\bot}}={\mathcal A}_{\mathcal W}^*.
\ee

\subsection{Distinguished basis}
Let us introduce a convenient basis, which describes a point of the Sato Grassmannian and its dual
\begin{definition}
For ${\mathcal W}\in \Gr^{(0)}_+$ let
\be\label{canbas}
\check\Phi^{\mathcal W}_k:={\mathtt G}_{\mathcal W} \cdot z^{k-1}
\ee 
be the {\em distinguished basis}. The {\em dual distinguished basis} is given by
\be\label{canbasd}
\check\Phi^{{\mathcal W}^\bot}_k:=({\mathtt G}_{\mathcal W}^*)^{-1} \cdot z^{k-1}. 
\ee 
\end{definition}

From Sato's Theorem it follows that (\ref{canbas}) for $k>0$ constitute an admissible basis in ${\mathcal W}$. Actually, (\ref{canbas}) and (\ref{canbasd}) are well-defined for $k\in {\mathbb Z}$, so each set describes the adapted basis of $H$. They are orthonormal
\be
\frac{1}{2\pi i}\oint_{\infty} \check\Phi^{\mathcal W}_k  \check\Phi^{{\mathcal W}^\bot}_m\, dz=\frac{1}{2\pi i}\oint_{\infty} z^{k+m-2}\, dz = \delta_{k+m,1}.
\ee
 
 Using operator ${\mathtt G}_{\mathcal W}$ one can also describe the {\em canonical basis}. Namely,
 \be\label{canon}
 \Phi^{\mathcal W}_k={\mathtt G}_{\mathcal W}\cdot \pi_+ ({\mathtt G}_{\mathcal W}^{-1}\cdot z^{k-1})
 \ee
 is an admissible basis in ${\mathcal W}$ of the form (\ref{goodbas}) by construction. Moreover, this basis satisfies
 \be
  \Phi^{\mathcal W}_k=z^k+O(z^{-1}).
 \ee
 The {\em dual canonical basis} is given by
 \be\label{dcanon}
   \Phi^{{\mathcal W}^\bot}_k=({\mathtt G}_{\mathcal W}^*)^{-1}\cdot \pi_+ ({\mathtt G}_{\mathcal W}^*\cdot z^{k-1}).
 \ee

Operator ${\mathtt Q}_{\mathcal W}$ is invertible with inverse  ${\mathtt Q}_{\mathcal W}^{-1}=z^{-1}+\dots \in {\mathcal D}_-$. ${\mathtt Q}_{\mathcal W}^{-1}$ is not a KS operator, but it yields an alternative expression for the distinguished adapted basis, for $k \in {\mathbb Z}$
\be\label{phitpsi}
\check\Phi_k^{\mathcal W}(z) =  {\mathtt Q}_{\mathcal W}^{k-1}\cdot \Psi.
\ee
For the distinguished basis we have
\begin{equation}
\begin{split}
{\mathtt P}_{\mathcal W}\cdot \check\Phi_k^{\mathcal W}(z)&=(k-1) \check\Phi_{k-1}^{\mathcal W}(z),\\
{\mathtt Q}_{\mathcal W}\cdot \check\Phi_k^{\mathcal W}(z)&=  \check\Phi_{k+1}^{\mathcal W}(z).
\end{split}
\end{equation}

The BA function (\ref{psiasop}) can be obtained from the wave function
\begin{equation}
\begin{split}
\Psi(z,x)&={\mathtt G}_{\mathcal W}\, e^{-xz} \, {\mathtt G}_{\mathcal W}^{-1}\,  {\mathtt G}_{\mathcal W} \cdot 1 \\
&=e^{-x {\mathtt Q}_{\mathcal W}} \cdot \Psi(z).
\end{split}
\end{equation}
so the distinguished basis arises in expansion of the BA function in $x$ (\ref{Grser}),
\be
\check{\Phi}_k^{\mathcal W} =\left.\left(-\frac{\p}{\p x}\right)^{k-1} \Psi(z,x)\right|_{x=0}.
\ee

Distinguished basis is closely related to the coefficients of the Sato group operator. Let
\be
{\mathtt G}_{\mathcal W}=:\sum_{k=0}^\infty {\mathtt G}_k\frac{\p^k}{\p z^k}.
\ee
\begin{lemma}\label{Gtocan}
\be\label{defconstr}
{\mathtt G}_k=\sum_{m+j=k} \frac{(-z)^{m}}{m! j!} \check{\Phi}^{\mathcal W}_{j+1}.
\ee
The distinguished basis is a unique basis, such that the right hand side of this expression belongs to $H_-+\delta_{k,0}$.
\end{lemma}
\begin{proof}
From (\ref{psiasop}) one has
\be
e^{xz}\Psi(z,x)=\sum_{k=0}^\infty (-x)^k {\mathtt G}_k .
\ee
The relation between ${\mathtt G}_k$ and $ \check{\Phi}^{\mathcal W}_{j}$ follows directly from the $x$-series expansion of this identity. It is also obvious, that 
a basis, such that right hand side of (\ref{defconstr}) belongs to $H_-+\delta_{k,0}$ is unique.
\end{proof}
In particular,
\begin{equation}
\begin{split}
{\mathtt G}_0&=\check{\Phi}_1^{\mathcal W},\\
{\mathtt G}_1&=\check{\Phi}_2^{\mathcal W}-z\check{\Phi}_1^{\mathcal W},\\
{\mathtt G}_2&=\frac{1}{2}\check{\Phi}_3^{\mathcal W}-z\check{\Phi}_2^{\mathcal W}+\frac{z^2}{2}\check{\Phi}_1^{\mathcal W}.
\end{split}
\end{equation}

Let ${\mathtt B}=\sum_{k=0}^\infty {\mathtt B}_k\frac{\p^k}{\p z^k}$ be an operator, then:
\begin{lemma}\label{l_GuniQ}
If ${\mathtt B}_0=\Psi$ and ${\mathtt Q}_{\mathcal W}{\mathtt B}={\mathtt B}z$, then 
\be
{\mathtt B}={\mathtt G}_{\mathcal W}
\ee
is Sato's group element.
\end{lemma}
\begin{proof}
Let us rewrite the relation satisfied by ${\mathtt B}$ as follows
\be
\left({\mathtt Q}_{\mathcal W} -z\right){\mathtt B} =\left[{\mathtt B},z\right].
\ee
If we take a coefficient of the $\frac{\p^k}{\p z^k}$ in the both sides of this relation we get a recursive relation for $ {\mathtt B}_k$
\be
\left({\mathtt Q}_{\mathcal W} -z\right)\cdot  {\mathtt B}_{k}+\dots=(k+1) {\mathtt B}_{k+1}
\ee
where by $\dots$ we denote the terms containing $ {\mathtt B}_{k-1}$,\dots, $ {\mathtt B}_{0}$. This recursive relation completely specifies $ {\mathtt B}$ for a given initial condition $ {\mathtt B}_0$. Then the statement of the lemma follows from the identity ${\mathtt G}_0=\Psi$ and from the
second identity in (\ref{PQdef}). 
\end{proof}

\subsection{Cauchy-Baker-Akhieser Kernel}

Let us introduce the {\em Cauchy-Baker-Akhieser kernel} \cite{Grin} associated to a point of the Sato Grassmannian
\be\label{Kdefin}
K_{\mathcal W}(z,w,x):=e^{x(w-z)}\frac{\tau_{\mathcal W}(x+[z^{-1}]-[w^{-1}])}{(w-z)\,\tau_{\mathcal W}(x)}
\ee
and
\be
K_{\mathcal W}(z,w):=K_{\mathcal W}(z,w,0).
\ee
It contains complete information about the tau-function. In particular, for arbitrary KP tau-function all connected and disconnected correlation functions are certain (reguralized) polynomials of the Cauchy-Baker-Akhieser kernel \cite{Zhou,ACEH1}. Coefficients of the series expansion of $K_{\mathcal W}(z,w,0)$ coincide with the coefficients in the fermionic group element (or affine coordinates of the point of the Sato Grassmannian = the coefficients of the canonical basis (\ref{canon})), see, e.g., Section 3.3.3 of \cite{AZ}. By definition
\be\label{kerexp}
K_{\mathcal W}(z,w)-\frac{1}{w-z}\in z^{-1}w^{-1}{\mathbb C}[[z^{-1},w^{-1}]].
\ee

According to the differential Fay identity \cite{AM}
\be\label{Fay}
\p\, K_{\mathcal W}(z,w,x)=\Psi (z,x) \Psi^{\bot} (w,x).
\ee
Let ${\mathtt G}_z$ be the Sato group element ${\mathtt G}_{\mathcal W}$ acting on the $z$ variable, and ${\mathtt G}_w^*$ be the adjoint Sato group element ${\mathtt G}_{\mathcal W}^*$, acting on the $w$ variable.
\begin{proposition}\label{Kernel_pr}
\be
K_{\mathcal W}(z,w,x)= {\mathtt G}_z ({\mathtt G}_w^*)^{-1} \cdot \frac{e^{x(w-z)}}{w-z}.
\ee
\end{proposition}
\begin{proof}
From (\ref{Fay}) and (\ref{psiasop}) it follows that
\be\label{Ktop}
K_{\mathcal W}(z,w,x)= {\mathtt G}_z ({\mathtt G}_w^*)^{-1} \cdot \frac{e^{x(w-z)}}{w-z}+\tilde{K}_{\mathcal W}(z,w).
\ee
for some $\tilde{K}_{\mathcal W}$ independent of $x$. Let us prove that $\tilde{K}_{\mathcal W}$ vanishes.

Comparing this expression at $x=0$ with the expansion (\ref{kerexp}), we conclude that $\tilde{K}_{\mathcal W}(z,w)\in w^{-1}z^{-1}{\mathbb C}[[z^{-1},w^{-1}]]$. Let $\iota_{|z|>|w|}$ be the operation of the Laurent series expansion in the region $|z|>|w|$.  We can rewrite the kernel (\ref{Kdefin}) as
\begin{equation}
\begin{split}
\iota_{|z|>|w|} K_{\mathcal W}(z,w,x)&= -z^{-1}\, e^{-xz}  \frac{\tau_{\mathcal W}(x+[z^{-1}])}{\tau_{\mathcal W}(x)}\, e^{xw +\sum_{k=1}^\infty \frac {w^k}{k z^k}}\frac{\tau_{\mathcal W}(x+[z^{-1}]-[w^{-1}])}{\tau_{\mathcal W}(x+[z^{-1}])} \\
&=-z^{-1}\, \Psi (z,x)\Psi^\bot(w,x+[z^{-1}]). 
\end{split}
\end{equation}
As a function of $w$ this is an element of ${\mathcal W}^\bot$, hence 
\be
{\mathtt G}_w^* \cdot \left( \iota_{|z|>|w|} K_{\mathcal W}(z,w,x) \right) \in {\mathbb C}[[x,z,z^{-1},w]]
\ee
contains only non-negative powers of $w$. This is also true for the result of the action of operator ${\mathtt G}_w^*$ on the first term in the right hand side of (\ref{Ktop}),
$ \iota_{|z|>|w|}   \left({\mathtt G}_z  \cdot \frac{e^{x(w-z)}}{w-z}\right )  \in {\mathbb C}[[x,z,z^{-1},w]]$, but ${\mathtt G}_w^* \cdot \left( \iota_{|z|>|w|}  \tilde{K}_{\mathcal W}(z,w)\right)={\mathtt G}_w^* \cdot \tilde{K}_{\mathcal W}(z,w) \in w^{-1}z^{-1}{\mathbb C}[[z^{-1},w^{-1}]]$. Hence $\tilde{K}_{\mathcal W}$ identically vanishes, and this complets the proof.
\end{proof}

In particular, for $x=0$ we have
\be
K_{\mathcal W}(z,w)= {\mathtt G}_z ({\mathtt G}_w^*)^{-1} \cdot \frac{1}{w-z}.
\ee

For any operator $A \in{\mathcal A}_{{\mathcal W}}\oplus {\mathcal A}_{{\mathcal W}^{\bot}}$ we have
\be
A({\mathtt Q}_{{\mathcal W}},{\mathtt P}_{{\mathcal W}},{\mathtt Q}_{{\mathcal W}^\bot},{\mathtt P}_{{\mathcal W}^\bot}) \cdot K_{\mathcal W}(z,w)= {\mathtt G}_z ({\mathtt G}_w^*)^{-1} A\left(z,\frac{\p}{\p z},w,\frac{\p}{\p w}\right) \cdot \frac{1}{w-z},
\ee
where operators ${\mathtt Q}_{{\mathcal W}}$, ${\mathtt P}_{\mathcal W}$ act on the $z$ variable, while the operators ${\mathtt Q}_{{\mathcal W}^\bot}$, ${\mathtt P}_{{\mathcal W}^\bot}$ act on $w$. For example, the following equations hold
\begin{equation}
\begin{split}
\left({\mathtt P}_{{\mathcal W}}+{\mathtt P}_{{\mathcal W}^\bot}\right)  \cdot K_{\mathcal W}(z,w,x) &= 0,\\
\left({\mathtt Q}_{{\mathcal W}^\bot}-{\mathtt Q}_{{\mathcal W}}\right)  \cdot K_{\mathcal W}(z,w,x) &= \Psi(z,x)\, \Psi^{\bot} (w,x).
\end{split}
\end{equation}

Basis expansions of the Cauchy-Baker-Akhieser kernel have a simple form both for distinguished and canonical basis.
Namely, for the extended distinguished basis and its dual one has
\begin{equation}
\begin{split}
\iota_{|w|>|z|}  K_{\mathcal W}(z,w)&=\sum_{k=0}^\infty \check\Phi_k^{\mathcal W} (z)\, \check\Phi_{-k-1}^{{\mathcal W}\bot}(w),\\ 
\iota_{|z|>|w|}  K_{\mathcal W}(z,w)&=-\sum_{k=0}^\infty \check\Phi_{k}^{{\mathcal W}\bot}(w) \, \check\Phi_{-k-1}^{\mathcal W} (z).
\end{split}
\end{equation}
For the canonical basis (\ref{canon}) and its dual (\ref{dcanon}) expansion is given by
\begin{equation}
\begin{split}
\iota_{|w|>|z|}  K_{\mathcal W}(z,w)&=\sum_{k=1}^\infty \frac{1}{w^{k}} \Phi_k^{\mathcal W}(z),\\
\iota_{|z|>|w|}  K_{\mathcal W}(z,w)&=-\sum_{k=1}^\infty \frac{1}{z^{k}} \Phi_k^{{\mathcal W}^\bot}(w).
\end{split}
\end{equation}


\subsection{Symmetries of KP hierarchy and Heisenberg-Virasoro algebra}\label{Virsec}

It is well known that a certain central extension of the algebra $\gl(\infty)$ and corresponding group $\GL(\infty)$ acts on the space of KP tau-functions \cite{F,MJ}. From the boson-fermion correspondence we also know how to identify the operators, acting on the tau-function with the operators, acting on the Sato Grassmannian. In this section we will remind the reader the properties of the important Heisenberg-Virasoro subalgebra of $\gl(\infty)$; for more details see, e.g., \cite{A1}. We also discuss how the Sato group element,  the canonical pair of KS operators, and the distinguished basis are transformed under the action of certain subgroups of $\GL(\infty)$.

The {\em Heisenberg-Virasoro subalgebra} of $\gl(\infty)$ is generated by the operators
\be
\widehat{J}_k =
\begin{cases}
\displaystyle{\frac{\p}{\p t_k} \,\,\,\,\,\,\,\,\,\,\,\, \mathrm{for} \quad k>0},\\[2pt]
\displaystyle{0}\,\,\,\,\,\,\,\,\,\,\,\,\,\,\,\,\,\,\, \mathrm{for} \quad k=0,\\[2pt]
\displaystyle{-kt_{-k} \,\,\,\,\,\mathrm{for} \quad k<0,}
\end{cases}
\ee
unit, and
\be
\label{virfull}
\widehat{L}_m=\frac{1}{2} \sum_{a+b=-m}a b t_a t_b+ \sum_{k=1}^\infty k t_k \frac{\p}{\p t_{k+m}}+\frac{1}{2} \sum_{a+b=m} \frac{\p^2}{\p t_a \p t_b}.
\ee
These operators satisfy the commutation relations
\begin{align}\label{comre}
\left[\widehat{J}_k,\widehat{J}_m\right]&=k \delta_{k+m,0},\nn\\
\left[\widehat{L}_k,\widehat{J}_m\right]&=-m \widehat{J}_{k+m},\\
\left[\widehat{L}_k,\widehat{L}_m\right]&=(k-m)\widehat{L}_{k+m}+\frac{1}{12}\delta_{k,-m}(k^3-k).\nn
\end{align}
The Heisenberg-Virasoro group ${\mathcal V}$ is generated by the operators $\widehat{J}_k$,  $\widehat{L}_k$ and a unit,
\be
{\mathcal V}:=\left.\{C \, e^{\sum a_k \widehat{J}_k +b_{k}\widehat{L}_k} \right| a_k,b_k,C \in {\mathbb C}\}.
\ee

To describe the action of algebra (\ref{comre}) on the Sato Grassmannian let us introduce
\be\label{virw}
{\mathtt j}_m=z^m,\\
{\mathtt l}_m=-z^m\left(z\frac{\p}{\p z}+\frac{m+1}{2}\right).
\ee
These operators satisfy the commutation relations (\ref{comre}) with omitted central term.
Then, the general relation of the action of the $\GL(\infty)$ and its central extension leads to the following relation between their subgroups: 
 If two points ${\mathcal W}$ and $\tilde{\mathcal W}$ of the Sato Grassmannian are related by the action of the group element
\be
\tilde{\mathcal W}=e^{\sum a_k {\mathtt j}_k +b_{k}{\mathtt l}_k}\cdot {\mathcal W}, 
\ee 
then the tau-functions are related by the element of the Heisenberg-Virasoro group
\be\label{1Vir}
\tau_{\tilde{\mathcal W}}=C \, e^{\sum a_k \widehat{J}_k +b_{k}\widehat{L}_k}\cdot \tau_{\mathcal W}
\ee
for some $C$ independent of ${\bf t}$. We do not address the question of convergence, which,  in general, can arise here.

Moreover, for any operator ${\mathtt a}\in {\mathcal A}_{\mathcal W}$ there is an operator ${\widehat A}$, acting on the space of functions of ${\bf t}$ such, that the tau-function is an eigenfunction of this operator,  $\widehat{A} \cdot \tau_{\mathcal W} = \tilde{C} \tau_{\mathcal W}$. Operator ${\widehat A}$ can be constructed with the boson-fermion correspondence, see, e.g., \cite{F,A1,KS}. For the Heisenberg-Virasoro subalgebra the identification is particularly simple:
\begin{lemma}\label{Lem_Vir}
If $\sum_{k\in\z} (a_k {\mathtt j _k}+b_k {\mathtt l _k})\in {\mathcal A}_{\mathcal W}$, then
\be
\sum_{k\in\z} (a_k {\widehat J}_k+b_k {\widehat L}_k)\cdot \tau_{\mathcal W}=\tilde{C}  \tau_{\mathcal W},
\ee
where $\tilde{C}$ does not depend on ${\bf t}$.
\end{lemma}

How do the operators  ${\mathtt G}_{\mathcal W}$ and $({\mathtt P}_{\mathcal W}$, ${\mathtt Q}_{\mathcal W})$  transform under the action of the operators from $\GL(\infty)$ group? Note that ${\mathtt j}_m\in {\mathcal D}_-$ and ${\mathtt l}_{m-1}\in {\mathcal D}_-$ for $m<0$, hence a group generated by this subalgebra is a subgroup of $\mathcal G$, and a corresponding action on the Sato Grassmannian and a canonical pair of the KS operators is described by (\ref{GGcon}) and (\ref{GGcon1}). For a general operator $\tilde{\mathtt G} \in \GL(\infty)$, that does not belong to ${\mathcal G}$, a transformation corresponding to the transformation ${\tilde{\mathcal W}}=\tilde{\mathtt G} \cdot {\mathcal W}$ can be highly non-trivial. However, for a subalgebra, generated by $\frac{\p^k}{\p z^k}$ and $z\frac{\p^k}{\p z^k}$, $k\geq 0$, and corresponding group we can find explicit relations. Namely, this algebra stabilizes $H_+$, thus, if  $\tilde{\mathtt G}=\exp\left(\sum_{k \in \z_{>0}} \left(r_k+s_k z\right)\frac{\p^k}{\p z^k}\right)$, then the Sato's group elements are related to each other by conjugation 
\be
{\mathtt G}_{\tilde{\mathcal W}}=\tilde{\mathtt G} \, {\mathtt G}_{{\mathcal W}} \,  \tilde{\mathtt G}^{-1} \in {\mathcal G}.
\ee
For  $\tilde{\mathtt G}=\exp\left(\sum_{k \in \z_{>0}} r_k \frac{\p^k}{\p z^k}\right)$ one has
\begin{equation}
\begin{split}
{\mathtt P}_{\tilde{\mathcal W}}&={\mathtt G}_{\tilde{\mathcal W}} \frac{\p}{\p z} {\mathtt G}_{\tilde{\mathcal W}}^{-1}\\
&=\tilde{\mathtt G} \, {\mathtt G}_{{\mathcal W}} \,  \tilde{\mathtt G}^{-1}\frac{\p}{\p z}
\tilde{\mathtt G} \, {\mathtt G}_{{\mathcal W}}^{-1} \,  \tilde{\mathtt G}^{-1}\\
&= \tilde{\mathtt G} \,{\mathtt P}_{{\mathcal W}} \, \tilde{\mathtt G}^{-1}
\end{split}
\end{equation}
and
\begin{equation}
\begin{split}
{\mathtt Q}_{\tilde{\mathcal W}}&={\mathtt G}_{\tilde{\mathcal W}}\, z\,  {\mathtt G}_{\tilde{\mathcal W}}^{-1}\\
&=\tilde{\mathtt G} \, {\mathtt G}_{{\mathcal W}} \,  \tilde{\mathtt G}^{-1}\, z\,
\tilde{\mathtt G} \, {\mathtt G}_{{\mathcal W}}^{-1} \,  \tilde{\mathtt G}^{-1}\\
&=\tilde{\mathtt G} \, {\mathtt G}_{{\mathcal W}} \,  \left(z- \sum_{k=1}^\infty k r_k \frac{\p^{k-1}}{\p z^{k-1}}\right) \,
 {\mathtt G}_{{\mathcal W}}^{-1} \,  \tilde{\mathtt G}^{-1}\\
 &= \tilde{\mathtt G} \,\left({\mathtt Q}_{{\mathcal W}} - \sum_{k=1}^\infty k r_k {\mathtt P}_{{\mathcal W}}^{k-1} \right)\, \tilde{\mathtt G}^{-1}.
\end{split}
\end{equation}
In particular, if $\tilde{\mathtt G}=\exp\left(r \frac{\p}{\p z}\right)=\exp\left(-r \mathtt{l}_{-1}\right)$, then
\be
{\mathtt G}_{\tilde {\mathcal W}}={\mathtt G}_{{\mathcal W}}\left(z+r,\frac{\p}{\p z}\right)
\ee
and
\begin{equation}
\begin{split}\label{conn0}
{\mathtt P}_{\tilde{\mathcal W}}&= {\mathtt P}_{{\mathcal W}}\left(z+r, \frac{\p}{\p z}\right),\\
{\mathtt Q}_{\tilde{\mathcal W}}&={\mathtt Q}_{{\mathcal W}}\left(z+r, \frac{\p}{\p z}\right)- r.
\end{split}
\end{equation}
The distinguished basis transforms according to 
\begin{equation}\label{shifto}
\begin{split}
\check\Phi^{\tilde{\mathcal W}}_k&=e^{r \frac{\p}{\p z}} {\mathtt G}_{{\mathcal W}} e^{-r \frac{\p}{\p z}} \cdot z^{k-1}\\
&= e^{r \frac{\p}{\p z}} {\mathtt G}_{{\mathcal W}}\cdot (z-r)^{k-1}\\
&= \sum_{j=0}^{k-1} \frac{(k-1)!}{j!(k-1-j)!} (-r)^{k-j}
   \check\Phi_{j+1}^{\mathcal W}(z-r).
\end{split}
\end{equation}

Similarly, for $\tilde{\mathtt G}=\exp\left(\sum_{k\in \z_{>0}} s_k z \frac{\p^k}{\p z^k}\right)$ one has
\begin{equation}
\begin{split}\label{conn1}
{\mathtt P}_{\tilde{\mathcal W}}&= \tilde{\mathtt G} \,F_1({\mathtt P}_{{\mathcal W}}) \, \tilde{\mathtt G}^{-1},\\
{\mathtt Q}_{\tilde{\mathcal W}}&= \tilde{\mathtt G} \,{\mathtt Q}_{{\mathcal W}}F_2({\mathtt P}_{{\mathcal W}}) \, \tilde{\mathtt G}^{-1}
\end{split}
\end{equation}
for some $F_k(z) \in {\mathbb C}[[z]]$. In particular, if $\tilde{\mathtt G}=\exp\left(s z \frac{\p}{\p z}\right)=\exp\left(-s({\mathtt l}_0+1/2)\right)$, then
\begin{equation}
\begin{split}\label{psd}
{\mathtt P}_{\tilde{\mathcal W}}&=e^{s}\, {\mathtt P}_{{\mathcal W}}\left(e^sz, e^{-s}\frac{\p}{\p z}\right),\\
{\mathtt Q}_{\tilde{\mathcal W}}&=e^{-s}\, {\mathtt Q}_{{\mathcal W}}\left(e^sz, e^{-s}\frac{\p}{\p z}\right),
\end{split}
\end{equation}
and
\begin{equation}\label{shift1}
\begin{split}
\check\Phi^{\tilde{\mathcal W}}_k&=e^{s z \frac{\p}{\p z}} {\mathtt G}_{{\mathcal W}} e^{-s z \frac{\p}{\p z}} \cdot z^{k-1}\\
&= e^{-s(k-1)}  \check\Phi_{k}^{\mathcal W}\left(e^s z\right).
\end{split}
\end{equation}

Subalgebra $\sll(2)$ of the Virasoro algebra (\ref{virfull}) describes a particularly simple class of transformations of the tau-functions. It is generated by three operators $\widehat{L}_m$ with $m\in\left\{-1,0,1\right\}$, given by the first order differential operators: $\widehat{L}_m=\sum_{k=1}^\infty k t_k \frac{\p}{\p t_{k+m}}.$
Hence corresponding group operators act by the homogeneous linear transformations of the time variables {\bf t}.  Namely, 
\be\label{Gtr}
e^{-u \widehat{L}_m} \cdot f({\bf t})= f({\bf t}^{(m)}),
\ee
where 
\begin{equation}
\begin{split}\label{Tchant}
{\bf t}^{(-1)}_k&=\frac{1}{k}\sum_{j=0}^\infty \frac{(k+j)!(k+j) (-u)^j}{k! j!} t_{k+j},\\
{\bf t}^{(0)}_k&=e^{-ku} t_k, \\
{\bf t}^{(1)}_k&=\sum_{j=1}^k \frac{ (k-1)! (-u)^{k-j}}{(j-1)! (k-j)!} t_j.
\end{split}
\end{equation}
For $m=-1$ and $m=0$ the action of the group elements on the canonical pair of the KS operators and the distinguished basis is given by (\ref{conn0}), (\ref{conn1}), (\ref{shifto}), (\ref{shift1}). The third operator ${\mathtt l}_1$ has positive degree, hence its action on the basis vector of form (\ref{goodbas}) gives an infinite combination of the new basis vectors of the same form.

\section{Generalized Kontsevich model}\label{S:GKM}

In this section we give a brief introduction to the theory of the generalized Kontsevich model (GKM). Many ingredients of this description are well known for an interesting class of polynomial potentials, more details can be found in \cite{Kh,IZ,KMMMZ,KMMM,MMS,KM,AM,ABGW,LG,F,KS,Sch}. We will consider a more general class of potentials, that includes arbitrary Laurent formal series. The main goal of this section is to describe how GKM fits into the Sato-Kac-Schwarz picture of the previous section. In particular, we discuss
 the Sato group element and the canonical pair of the KS operators.  
 
 \subsection{Matrix integral}

For any {\emph{potential}} $V(z)$ let us introduce the generalized Kontsevich matrix integral
\be\label{GKMM}
Z_U(\Lambda):={\mathcal C}^{-1}\int[ d\Phi] \exp\left(-\frac{1}{\hbar}\Tr(V(\Phi)-\Phi V'(\Lambda))\right).
\ee
Here $\hbar$ is a small parameter, and we assume that $V$ does not depend on it, 
\be\label{NormalC}
{\mathcal C}:=\hbar^{\frac{N(N-1)}{2}}\frac{\Delta(\lambda)}{\Delta(V'(\lambda))} \sqrt{\det\left(\frac{2\pi \hbar}{V''(\Lambda)}\right)}\,e^{\frac{1}{\hbar}\Tr(\Lambda V'(\Lambda)-V(\Lambda))}
\ee
and $\Lambda=\diag (\lambda_1,\dots,\lambda_N)$ is a diagonal matrix.
\be\label{flatmeas}
[ d\Phi]:=\frac{1}{\prod_{k=1}^{N-1}k!} \prod_{i<j}d \Im \Phi_{ij} \,d \Re \Phi_{ij}\, \prod_{i=1}^N d \Phi_{ii}
\ee
is the flat measure on the space of $ N\times N$ hermitian matrices. Let
\be
U:=\frac{1}{\hbar}V''(z).
\ee
We assume that $U$ is not identically equal to zero.

\begin{remark}
Sometimes it is necessary to consider more general types of measures, in particular of the form
\be
[d\Phi]  \det \left(\frac{\Phi}{\Lambda}\right)^M e^{\sum_{k=1}^\infty s_k \left(\Tr \Phi^{-k}-\Tr \Lambda^{-k}\right)}, 
\ee
where $M$ and $s_i$ are parameters \cite{KMMM,Takasaki}. This type of generalized measures is important for the recent developments, in particular for the Kontsevich-Penner model for open intersection numbers \cite{Aopen1, Aopen2, Aopen3}. 
We expect that such generalization of the deformed GKM considered in Section \ref{Secdef} is related to the generating function of the Hodge integrals on the open moduli spaces. 
\end{remark}

\begin{remark}
It is clear that the generating function (\ref{GKMM}) does not depend on the constant and linear terms or the potential $V$. Thus, it depends only on $U=\hbar^{-1}V''$, which justifies the notation. Below we assume that, given some $U$,  $h^{-1} V'$ and  $\hbar^{-1} V$ can be reconstructed as its first and second antiderivatives such that $V$ does not contain constant and linear terms in $z$. Abusing the terminology, we will call $U$ the {\em potential}, when it will not cause any confusion.
\end{remark}

The archetypal example of the potential is a monomial, 
\be\label{monpot}
V(z)=\frac{z^k}{k(k-1)}
\ee 
with $k\geq 3$. In particular, the original Kontsevich matrix model \cite{Konts} corresponds to the simplest case $k=3$. In applications there also appear certain deformations of the monomial potentials.  In particular, for the models of two-dimensional topological gravity, the potential $V(z)$ is a polynomial with lower terms considered as a perturbation of (\ref{monpot}), \cite{AM,IZ,KMMMZ,KMMM,Konts,W,W2,KM}.

Another interesting class of GKMs, that we will not consider in this paper, corresponds to the anti-polynomial potentials. In this case the description is slightly different and is not given by  (\ref{GKMM}), see, e.g., \cite{MMS, ABGW} for more details.  Deformation of the GKM with  anti-polynomial potentials will be considered elsewhere.


\subsection{Determinant representation of GKM}

Applying the Harish-Chandra-Itzykson-Zuber formula one can reduce the GKM integral (\ref{GKMM}) to the ratio of two determinants
\be\label{detform}
Z_U(\Lambda)=\frac{\det_{i,j=1}^N \Phi^U_i(\lambda_j)}{\Delta(\lambda)}.
\ee
Here
\be\label{basis}
\Phi_k^U(z):=
\sqrt{\frac{V''(z)}{2\pi\hbar}}e^{\frac{1}{\hbar}\left(V(z)-zV'(z)\right)}\int_{\gamma(U)} d\varphi \, \varphi^{k-1}e^{-\frac{1}{\hbar}\left(V(\varphi)-\varphi V'(z)\right)}.
\ee

Here and below by definition we consider an asymptotic expansion of the integral, obtained by the steepest descent method at $|z|\to \infty$. The stationary points $\varphi_0$ are the solutions of the equation
\be\label{critpoint}
V'(\varphi_0)-V'(z)=0
\ee
and, by definition, we consider only the contribution from the vicinity of the point 
\be\label{critp}
\varphi_0=z. 
\ee 
Thus, we choose the contour $\gamma(U)$ to be a small arc of the steepest descent contour, which goes through the point $\varphi_0$.  The direction of the contour is consistent with the sign of the root in the factor (\ref{basis}) and leads to Definition \ref{Phidef} below. We expect that it is always possible to extend the contours so that they will become closed without change of the asymptotic expansion of the integrals. However, the description of the closed contours depends on the properties of the potential. Therefore we prefer to work with local, but universal contour. 

For the asymptotic expansion of the integral (\ref{basis}) in the neighbourhood of the critical point (\ref{critp}) it is convenient to introduce a new integration variable
\be
\varphi \mapsto \frac{\varphi}{\sqrt{U(z)}}+z,
\ee
then the integrals are given by
\begin{definition}\label{Phidef}
\be\label{phidef}
\Phi_k^U(z):=\frac{1}{\sqrt{2\pi}}\int_{\rr} d {\varphi} \left(z+\frac{{\varphi}}{\sqrt{U(z)}}\right)^{k-1}e^{-\frac{{\varphi}^2}{2}-\sum_{j=3}^\infty \frac{U^{(j-2)}(z)}{j!U(z)^{j/2}}{ \varphi}^j},
\ee
where the integral is considered to be a perturbation of the Gaussian one.
\end{definition}
Here $U^{(m)}(z):=\frac{\p^m}{\p z^m} U(z)$.
In particular,
\begin{equation}
\begin{split}
\Phi_1^U(z)&=\frac{1}{\sqrt{2\pi}}\int_{\rr} d {\varphi}\,\, e^{-\frac{{\varphi}^2}{2}-\sum_{j=3}^\infty \frac{U^{(j-2)}(z)}{j!U(z)^{j/2}}{\varphi}^j}\\
&=\frac{1}{\sqrt{2\pi}}\int_{\rr} d {\varphi} \,\,e^{-\frac{{\varphi}^2}{2}}\,\,\sum_{j=0}^\infty p_j({\bf t^*}) \varphi^j\\
&=1+\sum_{j=2}^\infty (2j-1)!! \,\,p_{2j} ({\bf t^*}),\\
\Phi_2^U(z)&=\frac{1}{\sqrt{2\pi}}\int_{\rr} d {\varphi}\,\,\left(z+\frac{{\varphi}}{\sqrt{U(z)}}\right) e^{-\frac{{\varphi}^2}{2}-\sum_{j=3}^\infty \frac{U^{(j-2)}(z)}{j!U(z)^{j/2}}{ \varphi}^j}\\
&=z\left(1+\sum_{j=2}^\infty (2j-1)!! \,\,p_{2j} ({\bf t^*})\right)+\frac{1}{\sqrt{U(z)}} \sum_{j=1}^\infty (2j+1)!! \,\,p_{2j+1} ({\bf t^*}),
\end{split}
\end{equation}
where 
\be
t_j^*:=
\begin{cases}
\displaystyle{-\frac{U^{(j-2)}(z)}{j! U(z)^{j/2}} \,\,\,\,\,\,\,\,\,\,\,\, \mathrm{for} \quad j>2},\\[10pt]
\displaystyle{0}\,\,\,\,\,\,\,\,\,\,\,\,\,\,\,\,\,\,\, \mathrm{for} \quad j\leq 2,
\end{cases}
\ee
and $p_j({\bf t})$ are the elementary Schur functions
\be
e^{\sum_{k=1}^\infty t_k z^k}=:\sum_{j=0}^\infty p_j({\bf t})z^j.
\ee
In general,
\be\label{phigen}
\Phi_k^U(z)=z^{k-1}\left(1+\sum_{j=2}^\infty (2j-1)!! \,\,p_{2j} ({\bf t^*})\right)+\sum_{m=1}^{k-1}\frac{z^{k-1-m}}{U(z)^{m/2}} c_m,
\ee
where $c_m\in {\mathbb C}[[{\bf t }^*]]$. So, by definition for $k\geq 1$
\be\label{phistruct}
\Phi_k^U(z)\in\mathbb{Q}[z,U(z)^{-1/2}][[\frac{U'(z)}{U(z)^{3/2}},\frac{U''(z)}{U(z)^{2}},\dots,\frac{U^{(j)}(z)}{U(z)^{j/2+1}},\dots]]
\ee
with certain additional properties (e.g., it contains only integer powers of $U(z)^{-1}$). Moreover, $\Phi_k^U$ does not depend on  $N$, the size of the integration matrix.


\subsection{GKM as a tau-function of KP hierarchy}\label{Sec_GGKM}

So far we have not specified the properties of $U$. Our goal  is to investigate the GKM within the framework of Sato Grassmannian description of the KP hierarchy. So we want (\ref{basis}) to give an admissible basis  for some point of the Sato Grassmannian. From Definition \ref{Phidef} it is easy to see that
this is the case if $U$ is a formal Laurent series of the form
\be\label{Vseries}
U(z)=\frac{1}{\hbar}\sum_{k=-\infty}^{n} b_{k} z^{k}, \,\,\,\,\,\,\,\,\,\,\,\,\,\, n \geq -1,
\ee
with $b_{n}\neq 0$. In this case (\ref{phidef}) satisfy (\ref{goodbas}).

Hence, we work with the space  of formal Laurent series in $z^{-1}$,  $U(z)\subset{\mathbb C}((z^{-1}))$. In particular, when $U(z)$ is presented by a function, we only consider its series expansion at $z=\infty$. We take (\ref{Vseries}) as a definition of $U$ for the rest of this section. In the next section we introduce a further deformation of this simple setup.

\begin{remark}
What we consider here is a natural extension of the standard case, when $V(z)$ is a polynomial. In particular, $n=-1$ case of (\ref{Vseries}) corresponds to the leading term of the potential $V(z)$
of the form $z(\log z-1)$. The subleading terms of $V(z)$ can also include $\log z$. This kind of potentials appears in particular in the Penner model for the Euler characteristics of moduli spaces \cite{Penner}, and
 the Eguchi-Yang model for the stationary sector of the Gromov-Witten invariants on ${\mathbb P}^1$ \cite{EY,EY1}. 
\end{remark}

Hence we immediately conclude
\begin{proposition}\label{lemmaphi}
 (\ref{phidef}) is an admissible basis for a point of the big cell of the Sato Grassmannian
\be\label{GKMgr}
\Gr^{(0)}_+\ni  {\mathcal W}_U:= \sppan_{\cc}\{\Phi^U_1,\Phi^U_2,\Phi^U_3,\dots\}.
\ee
\end{proposition}
Here and below we use $U$ as an index instead of ${\mathcal W}_U$ to simplify the notation.

From this Proposition and determinant formula (\ref{miwatau}) it follows that the matrix model (\ref{GKMM}) defines a tau-function $\tau_U({\bf t})$ of the KP hierarchy by
\be
\tau_U([\Lambda^{-1}]):=Z_U(\Lambda).
\ee
The tau-function $\tau_U$ depends on $\hbar$ only through the dependence of $U$.

\subsection{Canonical operators for GKM}

Let us assume that $n\geq1$ in (\ref{Vseries}). In this case one can derive general and simple expression for the Sato group element ${\mathtt G}_U$ and the canonical pair of the KS operators $({\mathtt P}_{U},{\mathtt Q}_{U})$.

Let the contour $\tilde{\gamma}(U)$ be the contour ${\gamma}(U)$ shifted by $z$, $\tilde{\gamma}(U)={\gamma}(U)-z$, so that the saddle point on the new contour is located at the origin. Then, we define
\be\label{aver}
\left<\dots\right>_U:=\sqrt{\frac{U(z)}{2\pi}}e^{\frac{1}{\hbar}V(z) }\int_{\tilde{\gamma}(U)} d\varphi \, e^{\frac{1}{\hbar}\left(\varphi V'(z)-V(\varphi+z)\right)} \dots.
\ee
Let us note that the position of $\dots$  in the integral is important, because below we consider, in particular, it to be operator-valued. 
Let us describe the distingushed basis and the Sato group element for the tau-function $\tau_U$:
\begin{lemma}\label{lemma_Gel}
Definition \ref{Phidef} describes a distinguished basis for ${\mathcal W}_U$. Sato group element is given by
\be\label{Gav}
{\mathtt G}_U=\left< e^{\varphi\frac{\p}{\p z}}\right>_U.
\ee
\end{lemma}
\begin{proof}
It is easy to see that for the basis (\ref{phidef}) the right hand side of (\ref{defconstr}) belongs to $H_-+\delta_{k,0}$ for $k\geq 1$. Hence this basis is canonical. Then (\ref{Gav}) follows from Lemma \ref{Gtocan}.
\end{proof}
Basis (\ref{phidef}) is distinguished for ${\mathcal W}_U$, therefore we will denote it by $\check{\Phi}^U_k$.
Actually, Definition \ref{Phidef} is meaningful for all $k\in{\mathbb Z}$ and it provides the distinguished adapted basis (\ref{canbas}).
Moreover, the BA function (\ref{psiasop}) is given by
\begin{equation}
\begin{split}
\Psi(z,x)=\left<e^{-x(z+\varphi)}\right>_U.
\end{split}
\end{equation}

Now we can construct the canonical pair of the KS operators:
\begin{lemma}\label{Qform}
\be\label{Qforgkm}
{\mathtt Q}_U=z+\frac{1}{U(z)}\frac{\p}{\p z} -\frac{U'(z)}{2 U(z)^2},
\ee
and, if $b_{-1}=0$, 
\be\label{Pgenn}
{\mathtt P}_U=\frac{1}{\hbar}\left(V'({\mathtt Q}_U)-V'(z)\right).
\ee
\end{lemma}
\begin{proof}
Using the integration by parts in (\ref{phidef}), it is easy to show that 
\be\label{KSaaction}
\left(z+\frac{1}{U(z)}\frac{\p}{\p z} -\frac{U'(z)}{2 U(z)^2}\right) \cdot \check\Phi_k^U(z)=\check\Phi_{k+1}^U(z),
\ee
hence ${\mathtt Q}_U\in {\mathcal A}_U$. 

Let us assume that $b_{-1}=0$. Using integration by parts in (\ref{phidef}) it is  easy to show that $\hbar^{-1}(V'({\mathtt Q}_U)-V'(z))$ is a KS operator. Operators (\ref{Qforgkm}) and (\ref{Pgenn}) constitute a pair in $\Gr_D$, (\ref{GrD}), and from Lemma \ref{uniquecan} it follows that they constitute a canonical KS pair.
\end{proof}
Operator ${\mathtt Q}_U$ is the original KS operator, introduced by Kac and Schwarz  in \cite{KS}.

\begin{remark}
If $b_{-1}\neq0$, then $\hbar^{-1}V'(z)$ contains logarithmic term $b_{-1}\log z$. Thus, to construct operator  ${\mathtt P}_U$ in this case, one should deal with the operator of multiplication by $\exp\left(\pm{b_{-1}^{-1}}V'(z)\right)$. Here we will not consider this case, so below we assume that $b_{-1}=0$.
\end{remark}

\begin{proposition}
The  point of the Sato Grassmannian, dual to the one for the GKM with potential $U,$ is given by the GKM with the inverse sign of $U$,
\be
{\mathcal W}_{-U}={\mathcal W}_U^\bot.
\ee
\end{proposition}

\begin{proof}
From (\ref{Qforgkm}) and (\ref{Pgenn}) it follows that
\begin{equation}
\begin{split}
{\mathtt Q}_{-U}&={\mathtt Q}_{U}^*,\\
{\mathtt P}_{-U}&=-{\mathtt P}_{U}^*.
\end{split}
\end{equation}
Then, from (\ref{PQdual}) and Theorem \ref{bijecl} the statement of lemma immediately follows.
\end{proof}
In particular, this proposition implies that
\be\label{Gconj}
({\mathtt G}_U^*)^{-1}={\mathtt G}_{-U}
\ee
and the dual BA function  is
\be
\Psi^\bot(z,x)=e^{x z}\sqrt{-\frac{U(z)}{2\pi}}e^{-\frac{1}{\hbar}V(z) }\int_{\tilde\gamma(-U)} d\varphi \, e^{\frac{1}{\hbar}\left(V(\varphi+z)-\varphi V'(z)\right)+x\varphi},
\ee
or $\Psi^\bot(z,x)=\left<e^{x(z+\varphi)}\right>_{-U}$.
From Proposition \ref{Kernel_pr} we immediately have the asymptotic expansion expression for the Cauchy-Baker-Akhieser kernel for the GKM 
\be
K_U(z,w,x)=\frac{\sqrt{-U(z)U(w)}}{2\pi}e^{x(w-z)+\frac{1}{\hbar}(V(z)-V(w)) }\\
\times\int_{\tilde\gamma(U)} d\varphi_1 \int_{\tilde\gamma(-U)} d\varphi_2 \, e^{-\frac{1}{\hbar}\left(V(\varphi_1+z)- V(\varphi_2+w)-\varphi_1 V'(z)+\varphi_2 V'(w)\right)}\frac{e^{x(\varphi_2-\varphi_1)}}{w+\varphi_2-z-\varphi_1}.
\ee
For the polynomial $U$ it was established in \cite{ACM}.

Let us describe the action of the group generated by the operators  $\widehat{L}_{0}$ and $\widehat{L}_{1}$ on the space of GKM tau-functions. 
From (\ref{conn0}) and Lemma \ref{uniquecan}, it follows that  $\tau_{\tilde{U}}({\bf t})=e^{-r \widehat{L}_{-1}}\cdot \tau_U({\bf t})$, where 
\be\label{trans1}
\tilde{U}(z)=U(z+r).
\ee
Similarly, from (\ref{psd}) and Lemma \ref{uniquecan}, it follows that $\tau_{\tilde{U}}({\bf t})=e^{-s \widehat{L}_0}\cdot \tau_U({\bf t})$, where 
\be\label{trans2}
\tilde{U}(z)=e^{2s}U(e^{s}z).
\ee
With this transformation one can always fix the leading coefficient of $U$.

The GKM tau-function depends only on the combinations $b_j/\hbar$, not on these variables separately. With integration by parts it is also easy to show that 
\be\label{Dilat}
\left(\hbar^2 \frac{\p}{\p \hbar}+V(z)-V({\mathtt Q}_U)+V'(z)({\mathtt Q}_U-z) \right) \cdot {\mathcal W}_U \subset {\mathcal W}_U.
\ee


\subsection{Topological expansion and spectral curve}\label{S_topex}

Let us consider the topological (or $\hbar$) expansion of the GKM tau-function.  From Definition \ref{Phidef} it follows that $\tau_U$ is a formal series in $\hbar$
\be\label{topex}
\tau_U({\bf t})=\exp\left(\sum_{m=1}^\infty \hbar^{m} S_{m+1}({\bf t})\right).
\ee

We introduce 
\be\label{unst1}
S_0(z):=zV'(z)-V(z)=\int V''(z) z \,dz ,\\
S_1(z):=-\frac{1}{2}\log V''(z).
\ee 
We will call them the {\em unstable contributions}. We claim, that these terms properly describe the unstable contributions, which appear on the $A$-side enumerative geometry problems. To include the unstable contributions one can redefine the partition function as 
\be
\tilde{Z}_U:=e^{\frac{1}{\hbar} \Tr S_0(Z)+\Tr S_1(Z)}\,\tau_U({\bf t}). 
\ee
It is also convenient to introduce the {\em modified wave function}
\be\label{unst}
\tilde{\Psi}_U := e^{\frac{1}{\hbar}S_0(z)+S_1(z)} \, {\Psi}_{U},
\ee
which does not belong to ${\mathbb C}((z^{-1}))$ anymore. 

Depending on the context, in the literature both $\Psi_U$ and $\tilde{\Psi}_U$ are called the wave functions, and the operators annihilating them are referred to as the quantum spectral curve operators. Sometimes $\tilde{\Psi}$ leads to a more natural semi-classical limit of the operator and better describes the spectral curve. 
To distinguish two possibilities we will call (\ref{unst}) the {\em modified wave function}. The {\em modified quantum spectral curve} operator is given by
\be
\tilde{\mathtt P}_U:=\hbar\,  e^{\frac{1}{\hbar}S_0(z)+S_1(z)}\, {\mathtt P}_U \, e^{-\frac{1}{\hbar}S_0(z)-S_1(z)}
\ee
so that
\be\label{modq}
\tilde{\mathtt P}_U\cdot \tilde{\Psi}_U=0.
\ee

Let $x:=V'(z)$. Then  the operators
\begin{equation}\label{xy}
\begin{split}
\hat{x}&:=x,\\
\hat{y} &:=\hbar \frac{\p}{\p x},
\end{split}
\end{equation}
satisfy $ \left[\hat{y}, \hat{x}\right]=\hbar$. Then
\be
\hat{y}= e^{\frac{1}{\hbar}S_0(z)+S_1(z)} \, {\mathtt Q}_U \, e^{-\frac{1}{\hbar}S_0(z)-S_1(z)}.
\ee
For GKM the modified quantum spectral curve operator is given by
\be\label{MQSC}
\tilde{\mathtt P}_U=V'(\hat{y})-\hat{x}
\ee
and the modified wave function is
\be\label{defpsi}
 \tilde{\Psi}_U=\frac{1}{\sqrt{2\pi\hbar }}\int_{\gamma(U)} e^{\frac{1}{\hbar}(x\varphi-V(\varphi))}.
\ee
In the semi-classical limit operator (\ref{MQSC}) reduces to the classical spectral curve
\be\label{ClassSC}
V'(y)=x.
\ee
For a general Laurent series $U$ this equation describes only a local behaviour of the spectral curve at $y=\infty$. However, in some cases, in particular, if $V'$ is a rational function, this equation defines an algebraic plane curve. 

The unstable contributions (\ref{unst1}) can be rewritten as
\begin{equation}
\begin{split}
S_0&= \int y dx,\\
S_1&=-\frac{1}{2}\log \left(\frac{\p x}{\p y}\right).
\end{split}
\end{equation}
\begin{remark}\label{hbardef}
One can consider a more general situation, when the coefficients $b_k$ depend on $\hbar$, $b_k\in {\mathbb C}[[\hbar]]$ with $\left.b_n\right|_{\hbar=0}\neq 0$. In this case the modified quantum spectral curve equation (\ref{MQSC}) acquires further corrections. 
\end{remark}

\subsection{Polynomial potential}\label{Polyc}

In this section we consider polynomial potentials
\be\label{polypot}
U=\frac{1}{\hbar} \sum_{k=0}^{n} b_k z^k
\ee
or, equivalently $V(z)=\sum_{k=0}^n \frac{b_k z^{k+2}}{(k+1)(k+2)}$. The transformations (\ref{trans1}) and (\ref{trans2}), generated by the linear changes of variables ${\bf t}$, allow us to put $b_n=1$, $b_0=0$ in order to consider
\be\label{polypar}
U=\frac{1}{\hbar}\left(z^{n}+b_{n-1}z^{n-1}+\dots+b_1 z\right).
\ee
Below we will mostly work with such normalized potentials. 
\begin{remark}
Alternative form of the normalized potential with the vanishing subleading term,  
\be\label{Udid}
U=\frac{1}{\hbar}\left(z^{n}+b_{n-2}z^{n-2}+\dots+b_0\right),
\ee
 is equivalent to (\ref{polypar}), but less convenient for our purposes.
\end{remark}

Polynomial GKM is very well investigated \cite{AM,AMMP,IZ,Kh,KM,KMMMZ,KMMM,LG} because of its close relation to the interesting models of theoretical physics and enumerative geometry. Namely, for $n=1$ it corresponds to the Kontsevich matrix model, which originates in Kontsevich's proof \cite{Konts} of Witten's conjecture \cite{W} on the generating function of the intersection numbers on the moduli spaces. We will consider this case in Section \ref{Sec_Kon}.

For higher $n$ the GKM with the monomial potential 
\be\label{MON}
U=\frac{1}{\hbar} z^{n}
\ee
corresponds to $r$-spin intersection theory with $r=n+1$.  $r$-spin intersection theory is described by  the generalized Witten conjecture \cite{W3}. 
Connection between intersection theory and corresponding topological string models on one side, and GKM on another was one of the main reasons for the development of the GKM in the 90's, but the rigorous statement follows only from the proof of the r-spin Witten's conjecture by Faber, Shadrin an Zvonkine \cite{FSZ}: 
 \begin{proposition}
GKM for $V''(z)=z^{r-1}$ describes r-spin intersection numbers.
\end{proposition} 
Matrix model approach allows us to deform the monomial potential by the lower polynomial terms \cite{LG}. From the intersection theory point of view the general polynomial case corresponds to the intersection theory with  shifted Witten class. In physics it corresponds to the deformed $(r,1)$ minimal models coupled to topological gravity aka topological Landau-Ginzburg models. In Section \ref{redvirs} below we derive the Heisenberg-Virasoro constraints for the tau-function for arbitrary polynomial potential $U$.
\begin{remark}
Shifted r-spin theory may be more convenient for the Givental construction \cite{Giv1,Giv2}, as it allows to consider  a semi-simple point of the Frobenius manifold.
\end{remark}

For the polynomial GKM the operator $ {\mathtt P}_U$, given by (\ref{Pgenn}), is a differential operator of order $n+1$.  
Moreover, operator of multiplication by
\be\label{bfromP}
{\mathtt X}_U:=V'({\mathtt Q}_U) - \hbar {\mathtt P}_U 
\ee
is a KS operator, which is particularly convenient and simple,
\be\label{KS0}
{\mathtt X}_U=V'(z)=\hbar \int_0^z U(\eta) d\eta  \in {\mathcal A}_U.
\ee
It can be identified with the Landau-Ginzburg superpotential. 
Operators ${\mathtt Q}_U$ and ${\mathtt X}_U$ satisfy the commutation relation
\be\label{KCR}
\left[{\mathtt Q}_U,{\mathtt X}_U\right]=\hbar.
\ee
We see, that  for a polynomial $U$ the pair of the KS operators $({\mathtt X}_U,\mathtt{Q}_U)$ and the canonical pair  $({\mathtt P}_U,\mathtt{Q}_U)$ provide equivalent description of the KS algebra. In particular, the pair $({\mathtt X}_U,\mathtt{Q}_U)$ generates the KS algebra and uniquely specifies the point of the Sato Grassmannian.
\begin{lemma}\label{lem_unib}
For polynomial $U$, ${\mathcal W}_U$ is a unique point in $\Gr_+^{(0)}$, which is invariant under the action of the operators ${\mathtt X}_U$ and ${\mathtt Q}_U$.
\end{lemma}
\begin{proof}
If a point of Sato Grassmannian is invariant under the action of ${\mathtt X}_U$ and ${\mathtt Q}_U$, then, according to (\ref{bfromP}), it is also invariant under the action of ${\mathtt P}_U$. From Theorem \ref{bijecl} such a point is unique, so it coincides with ${\mathcal W}_U$.
\end{proof}

\begin{remark}
KS operators correspond to the certain relations between the Lax-Orlov-Schulman operators, see e.g. \cite{ACM}. Let us consider
\begin{equation}
\begin{split}
P&=\frac{1}{\hbar}V'(L),\\
Q&=L+M \frac{1}{U(L)} -\frac{U'(L)}{2U(L)^2},
\end{split}
\end{equation}
where $L$ and $M$ are given by (\ref{Lax}) and (\ref{OS}). From the existence of  $({\mathtt X}_U,\mathtt{Q}_U) \in {\mathcal A}_W$ it follows that 
\be\label{Doug}
P_-=Q_-=0.
\ee
It is obvious, that these operators satisfy Douglas's string equation 
\be\label{douglass}
[P,Q]=1.
\ee
\end{remark}
For the polynomial case the equation for classical spectral curve (\ref{ClassSC}) is polynomial.
Equation (\ref{Dilat}) in the polynomial case reduces to 
\be\label{dilap}
\left(\hbar^2 \frac{\p}{\p \hbar}+V(z)-zV'(z)\right) \cdot {\mathcal W}_U \subset {\mathcal W}_U.
\ee

Let us consider the case of monomial potential (\ref{MON}) in more detail. 
For the monomial case the KS operator ${\mathtt X}_U$ describes the reduction of the KP hierarchy to the 
$(n+1)$-KdV  aka Gelfand-Dickey hierarchy. In this case we have
\begin{equation}
\begin{split}
{\mathtt Q}_U&= z + \frac{\hbar}{z^{n}}\frac{\p}{\p z}-\frac{n\hbar}{2 z^{n+1}}, \\
{\mathtt X}_U&=\frac{z^{n+1}}{n+1}.
\end{split}
\end{equation}
The quantum spectral curve operator is given by
\be\label{qscmon}
{\mathtt P}_U=\frac{({\mathtt Q}_U)^{n+1}-z^{n+1}}{(n+1) \hbar}.
\ee
The modified quantum spectral curve equation is the higher Airy equation
\be
\left(\frac{\hat{y}^{n+1}}{n+1}-\hat{x}\right)\cdot \tilde{\Psi}=0,
\ee
so that the classical spectral curve is given by
\be\label{rspinc}
\frac{y^{n+1}}{n +1}=x.
\ee
The BA function and its dual up to a simple prefactor depend only on the combination of $z$  and $x$ variables:
\begin{equation}
\begin{split}
\Psi(z,x)&=\sqrt{\frac{z^{n}}{\hbar^{\frac{n}{n+2}}}}e^{-\frac{z^{n+2}}{(n+2)\hbar}}F_{n+1}(z^{n+1}\hbar^{-\frac{n+1}{n+2}}/(n+1)-x\hbar^{\frac{1}{n+2}}),\\
\Psi^\bot(z,x)&=\sqrt{\frac{z^{n}}{\hbar^{\frac{n}{n+2}}}} F_{n+1}^\bot (z^{n+1}\hbar^{-\frac{n+1}{n+2}}/(n+1)-x\hbar^{\frac{1}{n+2}}),
\end{split}
\end{equation}
where $F_{n+1}$ and $F_{n+1}^\bot$ are the asymptotic expansions of the particular solutions of the higher Airy equation. Note, however, that the variables play different roles: while we expand the function into the positive powers of $x$, the coefficients of this expansion are the formal Laurent series in $z^{-1}$.


\subsection{Equivalent hierarchies}\label{S_eq}

All GKMs with polynomial potentials of the same degree can be related to each other by the action of Heisenberg-Virasoro group considered in Section \ref{Virsec}. Let us describe a relation between GKM with a general polynomial potential (\ref{polypar}) and another GKM with a monomial potential $\tilde{U}(z)=z^n/\hbar$.
This relation was investigated in \cite{LG} (some proofs are omitted there).
 \begin{remark}
 In this and the next section we do not require that $b_0=0$, so the results are valid for arbitrary $U=\frac{1}{\hbar}\left(z^{n}+b_{n-1}z^{n-1}+\dots+b_1 z+b_0\right)$.
 \end{remark}

Let $f(z)\in z +{\mathbb C}[[z^{-1}]]$ be a change of local parameter $z$ given by 
\be\label{ffun1}
f(z)=\left((n+1)V'(z)\right)^{\frac{1}{n+1}},
\ee
so that $f(z)=z+\frac{b_{n-1}}{n}+O(z^{-1})$. It can be represented as
\be
f(z):=e^{\sum_{k \in {\z}_{<0}} a_k {\mathtt l}_k } \,z \, e^{-\sum_{k \in {\z}_{<0}} a_k {\mathtt l}_k }
\ee
for some $a_k$. Let us introduce an element of the Heisenberg-Virasoro group
\be
\tilde{\mathtt G}:= e^{\int^z (f(\eta)-\eta)U(\eta) d\eta}\,  e^{\sum_{k \in {\z}_{<0}} a_k {\mathtt l}_k } .
\ee
\begin{lemma}\label{lem_equiv}
\be
{\mathcal W}_U= \tilde{\mathtt G} \cdot {\mathcal W}_{\tilde U}
\ee
\end{lemma}
\begin{proof}
From the direct computations it follows that
\begin{equation}
\begin{split}
{\mathtt Q}_{U}&= \tilde{\mathtt G}\,{\mathtt Q}_{\tilde{U}}\, \tilde{\mathtt G}^{-1},\\
{\mathtt X}_{U}&= \tilde{\mathtt G}\, {\mathtt X}_{\tilde{U}}\, \tilde{\mathtt G}^{-1}.
\end{split}
\end{equation}
Then, the couples of the KS operators $({\mathtt X}_U,\mathtt{Q}_U)$ and $({\mathtt X}_{\tilde U},\mathtt{Q}_{\tilde U})$ uniquely define ${\mathcal W}_U\in Gr_+^{(0)}$ and  ${\mathcal W}_{\tilde U}\in Gr_+^{(0)}$, which completes the proof.
\end{proof}
Let us stress that  $ \tilde{\mathtt G}\notin {\mathcal G}$. We can also rewrite $ \tilde{\mathtt G}$ as
\be
\tilde{\mathtt G}=  e^{-\frac{1}{\hbar}{S}_0(z)} \,  e^{\sum_{k \in {\z}_{<0}} a_k {\mathtt l}_k } \, e^{\frac{1}{\hbar}\tilde{S}_0(z)},
\ee
where ${S}_0$ and $\tilde{S}_0$ are given by (\ref{unst1}) for $U$ and $\tilde{U}$ correspondingly. Using ({\ref{1Vir}}) from Lemma \ref{lem_equiv} we get a relation between tau-functions
\be\label{tautot}
\tau_U= C\,  e^{\sum_{k\leq n+1}  u_k \widehat{J}_k}  e^{\sum_{k \in {{\z}_{<0}}} a_k \widehat{L}_k }  \cdot \tau_{\tilde U},
\ee
where the coefficients  $u_k$ are given by (\ref{uk}). Tau-functions for $U$ and $\tilde{U}$ are equivalent in the sense of Shiota \cite{Shiota}. 

\begin{remark}
We are unable to find a similar relation between the GKM with a monomial potential and a general Laurent series case, considered in Section \ref{Sec_GGKM}. We expect, polynomial and general formal Laurent series GKM's with the same leading term are not equivalent in the sense of Shiota \cite{Shiota}.
\end{remark}


\subsection{Heisenberg and Virasoro constraints for polynomial GKM}\label{redvirs}

It is well-known that the GKM with polynomial potential satisfies Heisenberg and Virasoro (and, actually, the whole family of the $W^{(\infty)}$ constraints, of which 
 $W^{(r)}$ part provides full independent subset). Equations were derived for the monomial potentials \cite{KS,F,KMMMZ,AM,Sch}, however, for the polynomial potentials only the string equation is known explicitly \cite{KMMMZ}. Here for completeness we provide expressions for the Heisenberg and Virasoro constraints for the GKM with arbitrary polynomial potential.

We use Lemma \ref{Lem_Vir} to derive the constraints. The multiplication operators ${\mathtt X}_U^k$ describe the reduction of the KP hierarchy and correspond to the linear combinations of $\widehat{J}_m$ operators, while the first order differential operators ${\mathtt X}_U^k {\mathtt Q}_U$ describe the Virasoro constraints, satisfied by the tau-function. For a polynomial $U$ let us consider the KS operators
\be\label{KSj}
{\mathtt j}_m^U:={\mathtt X}_U^m,\,\,\,\,\,\,\,\,\,\,\,\, m\geq 1.
\ee
Then
\be
{\mathtt j}_m^U=\sum_{j=m}^{m (n+1)} g_{m,j} z^{j},
\ee
where
\be\label{gcoefd}
g_{m,j}:=\left[z^{j}\right] V'(z)^m,
\ee
or
\be
g_{m,j}=\sum_{l_1+\dots + l_m=j-m} \frac{b_{l_1}}{l_1+1}\dots \frac{b_{l_m}}{l_m+1}.
\ee

Let us also introduce the KS operators
\be\label{KSl0}
{\mathtt l}_m^U:= -\frac{1}{\hbar}{\mathtt X}_U^{m+1} {\mathtt Q}_U -\frac{m+1}{2} {\mathtt X}_U^m,\,\, \,\,\,\,\,\,\,\,\,\,\,\, m\geq -1.
\ee
In terms of operators (\ref{virw}) we have
\be\label{KSl}
{\mathtt l}_m^U:=\sum_{j=-\infty}^{m (n+1)} h_{m,j} {\mathtt l}_j  - \frac{1}{\hbar}\sum_{j=m+1}^{(m+1)(n+1)} g_{m+1,j} {\mathtt j}_{j+1},
\ee
where $g_{0,j}:=\delta_{j,0}$ and
\be\label{hcoefd}
h_{m,j}:=[z^{j+1}] \frac{V'(z)^{m+1}}{V''(z)}.
\ee
From the canonical commutation relation between ${\mathtt Q}_U$ and ${\mathtt X}_U$ one immediately concludes that
\begin{align}
\left[{\mathtt j}_m^U,{\mathtt j}_l^U\right]&=0,\nn\\
\left[{\mathtt l}_m^U,{\mathtt j}_l^U\right]&=-l\, {\mathtt j}_{m+l},\\
\left[{\mathtt l}_m^U, {\mathtt l}_l^U\right]&=(m-l) \, {\mathtt l}_{m+l}^U.\nn
\end{align}

From Lemma \ref{Lem_Vir} we have
\begin{equation}
\begin{split}
\label{eigen}
\widehat{J}_m^U\cdot \tau_U &=c_m^{(1)} \tau_U \,\,\,\,\,\,\,\,\,\,\,\, m\geq 1,\\
\widehat{L}_m^U\cdot \tau_U &=c_m^{(2)} \tau_U \,\,\,\,\,\,\,\,\,\,\,\, m\geq -1,
\end{split}
\end{equation}
for some $c_m^{(i)}$ independent of ${\bf t}$. Here
\begin{equation}
\begin{split}\label{JLop}
\widehat{J}_m^U&:=\sum_{j=m}^{m (n+1)} g_{m,j} \frac{\p}{\p t_{j}},\\
\widehat{L}_m^U&:=\sum_{j=-\infty}^{m(n+1)} h_{m,j} \widehat{L}_j  - \frac{1}{\hbar}\sum_{j=m+1}^{(m+1)(n+1)} g_{m+1,j} \frac{\p}{\p t_{j+1}} + \theta_m.
\end{split}
\end{equation}
The constant $\theta_m$ is given by
\begin{equation}
\begin{split}\label{thet}
\theta_m&=\frac{1}{12(m+2)}\sum_{j=n+1}^{(m+1)(n+1)} (j^3-j) h_{m+1,j} h_{-1,-j}\\
&=\frac{1}{12(m+2)}\res_{z=0} \left(\frac{1}{V''}\frac{\p^3}{\p z^3}\frac{(V')^{m+1}}{V''}\right).
\end{split}
\end{equation}
It is chosen to satisfy the commutation relations
\begin{align}\label{constcom}
\left[\widehat{J}_k^U,\widehat{J}_m^U\right]&=0,\,\,\,\,\,\,\,\,\,\,\,\,\,\,\,\,\,\,\,\,\,\,\,\, &k,m\geq 1,\nn\\
\left[\widehat{L}_k^U,\widehat{J}_m^U\right]&=-m \widehat{J}_{k+m}^U,\,\,\,\,\,&k\geq -1, m\geq 1,\\
\left[\widehat{L}_k^U,\widehat{L}_m^U\right]&=(k-m)\widehat{L}_{k+m}^U,\,\,\,\,\,\,\,\,\,\,\,\,\, &k,m\geq 1,\nn
\end{align}
where $\widehat{J}_0^U:=0$. From these commutation relations it follows that 
all operators $\widehat{J}_m^U$ and $\widehat{L}_m^U$ can be obtained by the commutation of other operators from this algebra, and thus the eigenvalues in (\ref{eigen}) vanish and we have:
\begin{proposition}\label{Lem_virpol}
Tau-function of the GKM for polynomial potentials satisfies the constraints
\begin{equation}
\begin{split}
\widehat{J}_m^U\cdot \tau_U &=0 \,\,\,\,\,\,\,\,\,\,\,\, m\geq 1,\\
\widehat{L}_m^U\cdot \tau_U&=0 \,\,\,\,\,\,\,\,\,\,\,\, m\geq -1.
\end{split}
\end{equation}
\end{proposition}

For the monomial GKM (\ref{MON}) the operators (\ref{eigen}), annihilating the tau-function, are
\be\label{monomvir}
\widehat{J}_m^U:=\frac{1}{(n+1)^m}\frac{\p}{\p t_{m (n+1)}},\\
\widehat{L}_m^U:=\frac{1}{(n+1)^{m+1}}\left(\widehat{L}_{m(n+1)}  - \frac{1}{\hbar}\frac{\p}{\p t_{(m+1)(n+1)+1}} + \frac{n^2+2n}{24}\,\delta_{m,0}\right).
\ee

\begin{remark} Constraints for the polynomial case can also be obtained from those for monomial case by conjugation with the operator (\ref{tautot}).
\end{remark}

\begin{remark}
For a non-polynomial potential ${\mathtt X}_{U}\notin {\mathcal A}_U$, so we cannot construct the Heisenberg-Virasoro algebra of constraints for $\tau_U$. However, we always have a string equation $\widehat{L}_{-1}^U\cdot \tau_U=0$.
\end{remark}


\subsection{Kontsevich model}\label{Sec_Kon}
In this section we consider the simplest polynomial GKM model with $n=1$. Namely, let us consider in more detail the cubic case
\be
U=\frac{z}{\hbar},
\ee
or, equivalently, $V(z)=\frac{z^3}{3!}$.  This case describes the Kontsevich-Witten (KW) tau-function, that governs the intersection theory on the moduli space of punctured Riemann surfaces \cite{Konts,W}, see Section \ref{S_Inters}.

The main goal of this section is to describe the structure of the Sato group operator for the KW tau-function. Namely, we derive a simple recursion relations for the coefficients of the Sato operator. This structure is not obvious from the integral formula, described by Lemma \ref{lemma_Gel}. 

Let us introduce
\be
{\mathtt W}:=\frac{1}{z}\frac{\p}{\p z} -\frac{1}{2z^2}.
\ee
Then the canonical pair of the KS operators is given by
\begin{equation}\label{bKW}
\begin{split}
{\mathtt Q}_{KW}&=z+\hbar\,  {\mathtt W},\\
{\mathtt P}_{KW}&=\frac{1}{\hbar}\left(\frac{{\mathtt Q}_{KW}^2}{2}-{\mathtt X}_{KW}\right)=\frac{\p}{\p z} +\hbar \frac{{\mathtt W}^2}{2},
\end{split}
\end{equation}
where ${\mathtt X}_{KW}=z^2/2$.

The modified quantum spectral curve (\ref{MQSC}) is given by the Airy equation
\be\label{KWmqsc}
\left(\frac{\hat{y}^2}{2}-\hat{x}\right)\cdot \tilde{\Psi}_{KW}=0,
\ee
which in the semi-classical limit yields
\be\label{KWcsc}
\frac{y^2}{2}=x.
\ee
The distinguished basis is given by the integrals
\be
\check{\Phi}_k^{KW}(z)=\sqrt{\frac{z}{2 \pi \hbar}} \int_{\tilde{\gamma}(U)} d\varphi\, (\varphi+z)^{k-1} e^{-\frac{1}{\hbar}\left(\frac{\varphi^3}{3!}+\frac{z\varphi^2}{2}\right)}.
\ee
The Sato group operator is given by Lemma \ref{lemma_Gel}, as
\begin{equation}\label{cgz}
\begin{split}
{\mathtt G}_{KW}&=\sqrt{\frac{z}{2 \pi \hbar}}  \int_{\tilde{\gamma}(U)} d\varphi\,  e^{-\frac{1}{\hbar}\left(\frac{\varphi^3}{3!}+\frac{z\varphi^2}{2}\right)} \, e^{\varphi \p_z}\\
&=\frac{1}{\sqrt{2 \pi}} \int_{\rr} d\varphi\,  e^{-\frac{\varphi^2}{2}-\frac{\varphi^3 \hbar^{1/2}}{z^{3/2}3!}} \, 
\sum_{k=0}^\infty\frac{1}{k!} \left(\frac{\varphi \hbar^{1/2}}{z^{1/2}}\right)^k \p_z^k\\
&=1+\hbar \left(\frac{5}{24z^3}-\frac{1}{2z^2}\p_z +\frac{1}{2z}\p_z^2\right)+O(\hbar^2).
\end{split}
\end{equation}
\begin{remark}
It would be interesting to find a relation between this group element and the cut-and-join operator, derived in \cite{CAJ}.
\end{remark}

The wave function and its dual are given by the Airy functions
\begin{equation}
\begin{split}
\Psi(z,x)&=\frac{\sqrt{\pi z}}{(2\hbar)^{1/6}}e^{-\frac{z^3}{3\hbar}}\Bi((z^2-2\hbar x)/(2\hbar)^{2/3}),\\
\Psi^\bot(z,x)&=2\frac{\sqrt{\pi z}}{(2\hbar)^{1/6}}e^{\frac{z^3}{3\hbar}}\Ai((z^2-2\hbar x)/(2\hbar)^{2/3}).\\
\end{split}
\end{equation}

It is well-known that the kernel for the KW model, as for any other solution of the KdV hierarchy, can be easily expressed in terms of the basis vectors of the Sato Grassmannian
\begin{equation}
\begin{split}
K_{KW}(z,w)&=\frac{\tau_{KW}([z^{-1}]-[w^{-1}])}{w-z}\\
&=\frac{\tau_{KW}([z^{-1}]+[-w^{-1}])}{w-z}\\
&=\frac{\Phi_1^{KW}(z) \Phi_2^{KW}(-w)- \Phi_1^{KW}(-w) \Phi_2^{KW}(z)}{z^2-w^2}.
\end{split}
\end{equation}
There exists a similar expression in terms of the dual point of the Sato Grassmannian ${\mathcal W}^\bot$.

Consider a  KS operator
\be\label{QPKW}
{\mathtt Q}_{KW} {\mathtt P}_{KW}=\left(z+\hbar {\mathtt W}\right)\left(\p_z+\hbar \frac{{\mathtt W}^2}{2}\right)= z\p_z +\hbar {\mathtt W}_1+\hbar^2 {\mathtt W}_2
\ee
with ${\mathtt W}_1=z\frac{{\mathtt W}^2}{2}+{\mathtt W} \p_z= \frac{5}{8z^3}-\frac{3}{2z^2}\p_z +\frac{3}{2z}\p_z^2$ and ${\mathtt W}_2=\frac{{\mathtt W}^3}{2}$. Note, that the operators ${\mathtt W}_1$ and ${\mathtt W}_2$ do not commute with each other, $\left[{\mathtt W}_1,{\mathtt W}_2\right]\neq 0$.

Let us consider the expansion of the Sato group element
\be
{\mathtt G}_{KW}=1+\sum_{k=1}^\infty {\mathtt G}_{KW}^{(k)},
\ee
where $\deg {\mathtt G}_{KW}^{(k)}=-3k$.

\begin{proposition}
Operators ${\mathtt G}_{KW}^{(k)}$ satisfy the recursion relation
\be
{\mathtt G}_{KW}^{(k)}=\frac{\hbar {\mathtt W}_1\,{\mathtt G}_{KW}^{(k-1)}+\hbar^2 {\mathtt W}_2 \,{\mathtt G}_{KW}^{(k-2)}}{3k}
\ee
with ${\mathtt G}_{KW}^{(k)}=0$ for negative $k$ and ${\mathtt G}_{KW}^{(0)}=1$.
\end{proposition}
\begin{proof}

From the definition of the canonical pair of the KS operators it follows that
\be
{\mathtt Q}_{KW} \, {\mathtt P}_{KW}\, {\mathtt G}_{KW} ={\mathtt G}_{KW} \, z\frac{\p}{\p z}.
\ee
Using the commutation relation
\be
{\mathtt G}_{KW}^{(k)} z\frac{\p}{\p z}=\left(z\frac{\p}{\p z}+3k\right){\mathtt G}_{KW}^{(k)}
\ee
and (\ref{QPKW}) we get
\be
\left(\hbar {\mathtt W}_1+\hbar^2 {\mathtt W}_2\right) {\mathtt G}_{KW} -3 \sum_{k=1}^\infty k {\mathtt G}_{KW}^{(k)}=0.
\ee
Let us take a component of degree $-3k$ for some $k>0$:
\be
\hbar {\mathtt W}_1{\mathtt G}_{KW}^{(k-1)}+\hbar^2 {\mathtt W}_2{\mathtt G}_{KW}^{(k-2)}-3k {\mathtt G}_{KW}^{(k)}=0.
\ee
This completes the proof.
\end{proof}
In particular, 
\begin{equation}
\begin{split}
{\mathtt G}_{KW}^{(1)}&=\hbar \frac{{\mathtt W}_1}{3},\\
{\mathtt G}_{KW}^{(2)}&=\hbar^2\left( \frac{{\mathtt W}_1^2}{18}+\frac{{\mathtt W}_2}{6}\right).
\end{split}
\end{equation}
From the proposition we see that
\be
{\mathtt G}_{KW}\in{\mathbb Q}[[\hbar{\mathtt W}_1, \hbar^2{\mathtt W}_2]].
\ee
Using this Sato group element, one can immediately obtain ${\mathtt g}_{KW}=\log{\mathtt G}_{KW}=\hbar \frac{{\mathtt W}_1}{3}+ \hbar^2 \frac{{\mathtt W}_2}{6}+O(
\hbar^3)$, and  the fermionic group element from Remark \ref{rmk_ferm}. It would be interesting to derive it  from the integral representation (\ref{cgz}).
\begin{remark}
Similar $n+1$-term recursion relations define the Sato operators for GKM with higher monomial potentials.
\end{remark}

\subsection{Rational potential}\label{LPC}

Our setup allows us to consider more general, non-polynomial potentials of the GKM tau-functions. Potentially interesting class of models corresponds to rational $V'(z)$, 
\be
V'(z)=\frac{p_1(z)}{p_2(z)},
\ee
where $p_1$ and $p_2$ are two polynomials.
The modified quantum spectral curve (\ref{MQSC}) leads to the differential equation
\be\label{csqlar}
\left(p_1(\hat{y})-p_2(\hat{y}) \hat{x}\right) \cdot \tilde{\Psi}_U=0,
\ee
where the operator in the parenthesis is polynomial in $\hat{x}$ and $\hat{y}$. In the semi-classical limit it leads to a genus zero algebraic spectral curve
\be
p_1(y)-p_2(y)x=0.
\ee

\begin{remark}
Let us consider a simple example with 
\be
V''(z)=\sum_{k=0}^n b_{k} z^{k}+b_{-2}z^{-2}.
\ee
Then 
\be
V'(z)=\frac{\sum_{k=0}^n b_{k} \frac{z^{k+2}}{k+1}-b_{-2}}{z},
\ee
and the modified quantum spectral curve equation (\ref{csqlar}) is given by
\be\label{hyper}
\left(\sum_{k=0}^n \frac{b_{k}}{k+1} \hat{y}^{k+2 }-\hat{y} \hat{x} -b_{-2} \right)\cdot \tilde{\Psi}_U=0.
\ee

Let us note that the modified quantum spectral curve (\ref{hyper}) is similar to that for the hypermaps, 
\be
\left( \sum_{k=1}^a t_k \hat{y}^{k}-  \hat{x} \hat{y} +1 \right)\cdot \tilde{\Psi}=0
\ee
derived in \cite{Do,Dunin}. It is easy to identify semi-classical limits of two quantum spectral curves. However, quantum spectral curves are different.
One can identify two quantum spectral curves if the coefficients $b_{k}$ depend of $\hbar$. However, it is not clear if the hypermaps can be described by the GKM with the potential explicitly dependent on $\hbar$, see Remark \ref{hbardef}. 

\end{remark}


\section{Deformed GKM}\label{Secdef}

In this section we consider a new, more general class of KP tau-functions, which can be described by the GKM. They share some common properties with the classical GKM considered in the previous section. However, other properties are essentially different. These tau-functions will again be labeled by the potential $U$, which is given by a certain deformation of the polynomial potential (\ref{polypar}). In addition to the parameters $b_j$ this deformed potential
depends on auxiliary formal parameters ${\bf w}=\{w_1,w_2,\dots,w_{n+1}\}$ in such a way, that the deformed potential, as well as all other elements of the construction, is formal series in ${\bf w}$. All the objects in this section reduce to the associated object in the previous section at ${\mathbf w}={\mathbf 0}$.  As it will be discussed below in Section \ref{Sec:Hodge}, this deformation is motivated by the Hodge integrals.

\subsection{Deformed potential}

The main goal of this section is to give a motivation for Definition \ref{DefPot} of the deformed potential from the point of view of KS description. It can be skipped on first reading.
We focus on the deformation of the polynomial GKM model, considered in Section \ref{Polyc}. We denote the undeformed polynomial potential by $U_0$,
\be\label{U0}
U_0=\frac{1}{\hbar}\left(z^n+b_{n-1}z^{n-1}+\dots+b_1z\right),
\ee
where $n\geq 1$. The deformed potential $U(z)$ will be a formal series in ${\bf w}$:
\be\label{Vdoubleu}
U(z)=U_0(z)+\sum_{\vec{l}\neq\vec{0}} U_{\vec{l}}(z)w_1^{l_1}\dots w_{n+1}^{l_{n+1}},
\ee
where $\vec{l}\in {\mathbb Z}_{\geq 0}^{n+1}$.
For arbitrary $U_{\vec{l}}(z) \in {\mathbb C} ((z^{-1}))$ the series expansion of integrals (\ref{phidef}) defines 
\be\label{Gdefb}
\Phi_k^{U}(z)=\Phi_k^{U_0}(z)+\sum_{\vec{l}\neq\vec{0}} \Phi_{k,\vec{l}}^{U_0}(z)  w_1^{l_1}\dots w_{n+1}^{l_{n+1}},
\ee
where $\Phi_{k,\vec{l}}^{U_0}(z) \in {\mathbb C} ((z^{-1}))$. However, for general $U$ it is not clear if (\ref{Gdefb}) is an admissible basis for any point of the Sato Grassmannian. We will construct a class of deformations such that this basis is admissible for a certain point of the Sato Grassmannian, ${\mathcal W}_U \in \Gr_+^{(0)}$. 

Denote by ${\mathbf K} ={\mathbb C}[[{\bf w}]]$ the
ring of formal series in the $w_k$'s and by  ${\mathbf K}^*=\bigoplus_{k=1}^L \,w_k {\mathbf K}$ its ideal of formal series with trivial constant term. All objects considered below are formal power series in ${\bf w}$. 

\begin{remark}\label{trivdef}
For one particular class of deformations the integrals (\ref{phidef}) obviously define an admissible basis. Namely, let
\be
U(z)=\frac{1}{\hbar}\sum_{k=0}^{n} b_{k}({\bf w}) z^{k},
\ee
 and $b_k ({\bf w}) \in{\mathbf K}$ for all $k$ with  $\left.U(z)\right|_{{\bf w}={\bf 0}} = U_0(z)$. Then (\ref{phidef}) defines a point of the Sato Grassmannian. 
However, this type of deformations can be obtained from the undeformed case by consideration of the coefficients $b_k$  dependent of $w_j$ in all expressions of the previous section. Thus, it does not give any new solutions and is not very interesting. Below we will neglect the deformations, which can be absorbed by the redefinition of the coefficients of $U_0$.  
\end{remark}

We consider a class of deformations that leads to the natural deformations of the KS operators ${\mathtt Q}_{U_0}$ and ${\mathtt X}_{U_0}$,
given by  (\ref{Qforgkm}) and  (\ref{KS0}). 
Namely, we require that for the deformed $U$ the operators defined by the right hand sides of (\ref{Qforgkm}) and (\ref{KS0}) satisfy certain properties similar to ones for the polynomial potential considered in Section \ref{Polyc}. In particular, we want these operators to stabilize the space, generated by (\ref{phidef}). The operator
\be\label{Qdef}
\tilde{\mathtt Q}_{U}:=z+\frac{1}{U(z)}\frac{\p}{\p z} -\frac{U'(z)}{2 U(z)^2}
\ee
always stabilizes this space, because from the integration by parts we have
\be\label{Qtildac}
\tilde{\mathtt Q}_{U} \cdot \Phi_k^U(z)=\Phi_{k+1}^U(z).
\ee
We also want this operator to generate the basis from the wave function $\Psi_U$ (to be constructed below). For this purpose we require it to have $z$ as the top degree term, or
\be\label{atrec}
\tilde{\mathtt Q}_{U}-z\in {\mathcal D}^{(0)}.
\ee
From (\ref{Qdef}) it follows that this requirement is equivalent to
\be\label{req2}
\frac{1}{U} \in z\,{\mathbf K}[[z^{-1}]].
\ee
\begin{remark} In general, the operator $\tilde{\mathtt Q}_{U}$ for the deformed potential does not belong to $z+{\mathcal D}_-$,
so it does not coincide with the canonical KS operator ${\mathtt Q}_{U}$. The operator $\tilde{\mathtt Q}_{U}$ satisfies (\ref{PQop}), so it can be
a canonical KS operator ${\mathtt Q}_{U}$ for some point of the Sato Grassmannian only if $\frac{1}{U} \in z^{-1}{\mathbf K}[[z^{-1}]]$.
\end{remark}
Then, consider 
\be
{\mathtt X}_{U}=\hbar \int^z U(\eta) d\eta.
\ee 
From the integration by parts one has
\be
\left(V'(\tilde{\mathtt Q}_{U})-{\mathtt X}_{U}\right)\cdot \Phi_k^U(z)=(k-1)\Phi_{k-1}^U(z).
\ee
Hence the operator of multiplication by ${\mathtt X}_{U}$ stabilizes the space generated by $\Phi_k^{U}$ if it does not contain negative powers of $z$, or
\be\label{req1}
U \in  {\mathbf K}[[z]].
\ee

Let us consider the classes of equivalence of the deformations satisfying (\ref{req1}), that differ by the trivial deformations (see Remark \ref{trivdef}). Let
\be
{\mathcal U}_0^n:=\left.\{ z^{n}+b_{n-1}z^{n-1}+\dots+b_1z \, \right|\, b_k\in {\mathbb C}, b_{n}\neq 0   \}
\ee
and
\be
\tilde{\mathcal U}_0^n:=\left.\{\tilde b_n z^{n}+\tilde b_{n-1}z^{n-1}+\dots+\tilde b_0 \, \right|\, \tilde b_k\in {\mathbf K}^*  \}.
\ee
Consider the quotient space
\be
{\mathcal U}^n:=\left.\left.\{U\in {\mathbf K}[[z]]  \,\right|\,  \left.U\right|_{{\bf w}={\bf 0}} \in  {\mathcal U}_0^n  \}\right/\tilde{\mathcal U}_0^n.
\ee

\begin{lemma}
For any $U_0 \in {\mathcal U}_0^n$ and $U\in {\mathbf K}[[z]]$ with  $\left.U\right|_{{\bf w}={\bf 0}}=U_0$,  there exists a unique $\tilde{U} \in \tilde{\mathcal U}_0^n$ such that $U+\tilde U=\frac{U_0}{1-\Delta U}$ with $\Delta U \in z{\mathbf K}^* [[z]]$.
\end{lemma}
\begin{proof}
Let us consider $U^*:=U-U_0\in {\mathbf K}^* [[z]]$. Then $\tilde{U}$ should satisfy $1+\Delta U =\frac{1}{1+\frac{U^*+\tilde{U}}{U_0}}$.  We see that $\Delta U \in z{\mathbf K}^* [[z]]$  if and only if $\frac{U^*+\tilde{U}}{U_0}\in z{\mathbf K}^* [[z]]$. Such $\tilde{U} \in \tilde{\mathcal U}_0^n$ always exists, and moreover, it is unique.
\end{proof}

Hence ${\mathcal U}^n$ can be parametrized by $\Delta U \in z{\mathbf K}^* [[z]]$ through $U=\frac{U_0}{1-\Delta U}$. To satisfy (\ref{req2}), we should further require $\Delta U$ to be a polynomial in $z$ of degree at most $n+1$. So the deformations are parametrised by polynomials of degree $n+1$. It is convenient to define the parameters $w_i$ to be
the inverse zeros of these polynomials
\be
1+\Delta U=:(1-w_1z)(1-w_2z)\dots(1-w_{n+1}z). 
\ee
This leads us to Definition \ref{DefPot}. We can rewrite $U$ as
\be
U(z)=U_0(z) \sum_{k=0}^\infty h_k({\mathbf w}) z^k,
\ee
where $h_k$ are the complete symmetric functions
\be
\sum_{k=0}^\infty h_k({\mathbf w}) z^k=\prod_{j=1}^{n+1}\frac{1}{1-w_iz}.
\ee
If $w_i\neq w_j$ for $i\neq j$ and $w_i\neq 0$, we have
\be
V'(z)=\sum_{k=1}^\infty \chi_k({\bf w}) \log (1-w_k z),
\ee
where $\chi_k({\bf w})\in {\mathbb C}(({\bf w}))$. For the deformation of the monomial potential we can rewrite it as
\be
V'(z)=\sum_{k=0}^\infty \frac{h_k({\bf w})z^{k+n+1}}{k+n+1}.
\ee

\subsection{Admissible basis for the deformed potential }	

For the deformed potential the integrals (\ref{phidef}), $\Phi_k^U\in {\mathbb C}((z^{-1}))[[{\bf w},\hbar]]$ are series in $z$, infinite in both directions. So, contrary to the case considered in Section \ref{S:GKM}, they do not constitute a distinguished basis. In this section we will show that for a deformed potential (\ref{gendefor}) the integrals (\ref{phidef}) after a certain change of normalization form an admissible basis of a point of the Sato Grassmannian denoted by ${\mathcal W}_U$. By construction,  operators $\tilde{\mathtt Q}_U$ and ${\mathtt X}_U$ are the KS operators for this point. We prove that similar to the polynomial case these operators completely specify this point of the Sato Grassmannian. We also construct a canonical pair of the KS operators.

It is convenient to introduce one formal small parameter $w$ such that $w_i=\beta_j w$ for some $\beta_j \in {\mathbb C}$.

\begin{lemma}\label{wass}
\be
\Phi_k^{U}(z)=\Phi_k^{U_0}(z)+\sum_{j=1}^\infty \Phi_{k,j}(z) w^j,
\ee
where $\Phi_{k,j}(z) \in z^{k+j-3-n} {\mathbb C}[[z^{-1},\hbar]]$.
\end{lemma}

\begin{proof}
The statement immediately follows from the expansion of the basis vectors (\ref{phidef}).
\end{proof}

\begin{lemma}\label{lemmacanon} 
For any deformed potential there exists $\alpha \in 1+ {\mathbf K}^*[[\hbar]]$ such that $\alpha \Phi_k^U$ define an admissible basis of a point of the Sato Grassmannian ${\mathcal W}_U \in  \Gr^{(0)}_+$. 
\end{lemma}

\begin{proof}
Let us show that for the deformed potential the integrals (\ref{phidef}) are consistent with Definition \ref{Defadm}. 
First, the integrals (\ref{phidef}) define the point of the Sato Grassmannian $\rm{Gr}^{(0)}_+$. This is obvious for the case $w_k=0$, and the $w_k\neq 0$ deformation does not break the isomorphism between $\{\Phi_i\}$ and $H_+$. The multiplication by any $\alpha \in 1+{\mathbf K}^*[[\hbar]]$ also respects the isomorphism.
Hence, condition 1 of Definition \ref{Defadm} is satisfied. 

It remains to show that condition 2 of Definition \ref{Defadm} is also satisfied. Let us construct some admissible basis for ${\mathcal W}_U$. We can find the wave function recursively. Namely, we want to construct a linear combination of $\Phi_k^U$ that does not contain positive powers of $z$.
From Lemma \ref{wass} it follows that $\pi_+(z^{-1}\Phi_1^U)=c w^{n+3}+\dots$, where by $\dots$ we denote terms with $w^{n+4}$ or higher powers of $w$. Then $\Phi_1^U(z)-c w^{n+3}\Phi_2^U(z)$ does not contain positive powers of $z$ up to $w^{n+4}$. Repeating this argument we see, that there exist $v_m\in {\mathbf K}^*$ for $m>1$ such that
\be\label{psiassum}
\Phi^U_1+\sum_{m=2}^\infty v_m \Phi^U_m \in {\mathbf K}[[z^{-1}, \hbar]].
\ee
We can normalize this sum to define
\be
\Psi_U:=\alpha(\Phi^U_1+\sum_{m=2}^\infty v_m \Phi^U_m)
\ee 
with some $\alpha \in1 + {\mathbf K}^*[[\hbar]]$. Then $ \tilde{\mathtt Q}_{U}^k\cdot\Psi_U$ constitute an admissible basis 
\be\label{basisnon}
{\mathcal W}_U:= \sppan_{\cc}\{\Psi_U, \tilde{\mathtt Q}_U\cdot\Psi_U, \tilde{\mathtt Q}^2_U \cdot \Psi_U,\dots \}.
\ee
Let us find a determinant, connecting this basis to $\alpha \Phi^U_k$. From (\ref{Qtildac}) one has 
\begin{equation}
\begin{split}
\tilde{\mathtt Q}_U^{k-1} \cdot \Psi_U =\alpha(\Phi^U_{k}+\sum_{m=2}^\infty v_m \Phi^U_{k+m-1}).
\end{split}
\end{equation}
So, the matrix relating the two basis is upper-triangular. The determinant of the matrix, relating bases $\tilde{\mathtt Q}_U^{k-1} \cdot \Psi_U$ and $\alpha \Phi^U_{k}$, is equal to $1$. Hence, the basis $\alpha \Phi^U_{k}$ is admissible,
\be
{\mathcal W}_U:= \sppan_{\cc}\{\alpha \Phi_1^U, \alpha \Phi_2^U, \alpha \Phi_3^U,\dots \}  \in \Gr_+^{(0)}.
\ee
\end{proof}
\begin{remark} For the deformed potential the basis $\Phi_k^U$ is admissible. However,  the basis vectors contain positive powers of $z$,
so the right hand side of the finite determinant formula (\ref{detform}) contains positive powers of ${\bf \lambda}$. Therefore it is not equal to tau-functions in the Miwa parametrization. The tau-function appears only in the limit $N\rightarrow\infty$.
\end{remark}
From this Lemma it follows, in particular, that the wave function is a linear combination of $\Phi_k^U$, or
\be\label{Waved}
\Psi_U = \left< e^{\frac{1}{\hbar} R(\varphi+z)}\right>_U
\ee
for some $R(z)\in {\mathbf K}^*[[z,\hbar]]$, see (\ref{aver}). The normalization factor is $\alpha=e^{\frac{1}{\hbar}R(0)}$.
\begin{remark}
By construction the tau-function for the deformed potential satisfies the topological expansion (\ref{topex}). We extend the definition of the unstable contributions (\ref{unst1}) to the deformed case.
\end{remark}

Let us describe the canonical pair of the KS operators for the deformed potentials as formal series of the KS operators $\tilde{\mathtt Q}_{U}$ and ${\mathtt X}_U$. Consider an expansion
\be\label{1Uexp}
\frac{1}{U}=c_1 z+c_2 +O(z^{-1}),
\ee
where $c_1=\hbar \prod_{j=1}^{n+1}(-w_j)$ and $c_2=-c_1\left(b_{n-1}+\sum_{j=1}^{n+1} w_j^{-1}\right)$. 
Denote
\be
{\mathtt P}_U^0:=\frac{1}{\hbar}\left(V'(\tilde{\mathtt Q}_U)-{\mathtt X}_U\right) \in {\mathcal A}_U.
\ee
Then
\begin{lemma}\label{deforlemma}
KS operators $\tilde{\mathtt Q}_{U}$ and ${\mathtt X}_U$ uniquely specify ${\mathcal W}_{U}$. Moreover,
\be\label{Pgeg}
{\mathtt P}_U= \left(\frac{1}{c_1}   -\frac{\tilde{R}'(\tilde{\mathtt Q}_U)}{\hbar} \right ) e^{c_1 {\mathtt P}_U^0}   -\frac{1}{c_1},
\ee
where $\tilde{R}'(z) \in {\mathbf K}^*[[z,\hbar]]$, and
\be\label{Qgeg}
{\mathtt Q}_U=\left(\tilde{\mathtt Q}_U-c_2 {\mathtt P}_U -\frac{c_1}{2}\right)
\frac{1}{1+c_1{\mathtt P}_U}.
\ee
\end{lemma}

\begin{proof}
From  (\ref{Qdef}) and (\ref{1Uexp}) we have
\be
\tilde{\mathtt Q}_U\in z+ (c_1 z+c_2)\frac{\p}{\p z} +\frac{c_1}{2} + {\mathcal D}_-.
\ee
Then (\ref{Qgeg}) follows, because the right hand side is a KS operator belonging to $z+{\mathcal D}_-$, and from Lemma \ref{uniquecan} such operator is unique.

Let us construct the canonical KS operator ${\mathtt P}_U$ recursively as a series in $w$. Operator ${\mathtt P}^{0}_U$ reduces to the quantum spectral curve operator for the undeformed potential at $w=0$:
\be
{\mathtt P}^{0}_U\big|_{w=0}=\frac{1}{\hbar}\left(V_0'({\mathtt Q}_{U_0})-{\mathtt X}_{U_0}\right)={\mathtt P}_{U_0} \in \frac{\p}{\p z}+{\mathcal D}_-.
\ee
Moreover, $\tilde{\mathtt Q}_U\big|_{w=0}= {\mathtt Q}_{U_0}\in z+ {\mathcal D}_-$, and by definition $[w^k]V'(z)\in z^{n+k+1}{\mathbb C}[[z^{-1}]]$.
Hence, we can recursively construct
\be\label{Pdelta}
{\mathtt P}_U={\mathtt P}_U^0+ \Delta {\mathtt P}, 
\ee
where $\Delta {\mathtt P}= \sum_{k=1}^\infty w^k r_k(\tilde{\mathtt Q}_U, {\mathtt P}_U^0)$
and $r_k$ are some polynomials. Let us write $\Delta {\mathtt P}=f(\tilde{\mathtt Q}_U, {\mathtt P}_U^0)$, where
\be
f(x,y)=\sum_{k,m=0}^\infty f_{km} x^k y^m
\ee
for some $f_{km}\in {\mathbf K}^*$.

To find $\Delta {\mathtt{P}}$ let us use the commutation relation $\left[{\mathtt P}_U,{\mathtt Q}_U\right]=1$. Inserting (\ref{Qgeg}) we have
\be
\left[{\mathtt P}_U,\tilde{\mathtt Q}_U\right]=1+c_1{\mathtt P}_U.
\ee
To get an equation satisfied by $\Delta {\mathtt P}$ we substitute (\ref{Pdelta})
\be
\left[\Delta{\mathtt P},\tilde{\mathtt Q}_U\right]=c_1({\mathtt P}_U^0+ \Delta {\mathtt P}).
\ee
Here we also use the definition of ${\mathtt P}_U^0$ and the commutation relation between $\tilde{\mathtt Q}_U$ and ${\mathtt X}_U$. 
Then the series $f$ satisfies the differential equation
\be
\frac{\p}{\p y} f (x, y)= c_1 (y+f (x, y)).
\ee
The general solution is
\be
f (x, y)=\left(\frac{1}{c_1}-\frac{1}{\hbar}C(x)\right)e^{c_1 y}-\frac{1}{c_1}-y,
\ee
where $C(x)$ is arbitrary function. Here we introduce the factor $\hbar^{-1}$ for the convenience. To get the correct limit it should satisfy $\left.C(x)\right|_{w=0}=0$.
There is a unique series $C(x)$  such that ${\mathtt P}_U\in \frac{\p}{\p z}+{\mathcal D}_-$. We denote it by $\tilde{R}'(x)$, $\tilde{R}(x)\in {\mathtt C}[z][[w]]$.

We see, that the operators $\tilde{\mathtt Q}_{U}$ and ${\mathtt X}_U$ uniquely specify the canonical pair of the KS operators $({\mathtt P}_U,{\mathtt Q}_U)$, hence they uniquely specify
${\mathcal W}_{U}$. This competes the proof.
\end{proof}
The modified quantum spectral curve equation for the deformed potential is given by
\be
\left(\left(\frac{\hbar}{c_1}   -\tilde{R}'(\hat{y}) \right ) e^{\frac{c_1}{\hbar} \left(V'(\hat{y})-\hat{x}\right)}   -\frac{\hbar}{c_1}\right)\cdot \tilde{\Psi}_U=0.
\ee

Operators ${\mathtt P}_U$ and ${\mathtt Q}_U$ are deformations of the operators ${\mathtt P}_{U_0}$ and ${\mathtt Q}_{U_0}$
\begin{equation}
\begin{split}
\left.{\mathtt P}_U\right|_{{\bf w}={\bf 0}}={\mathtt P}_{U_0}, \,\,\,\,\,\,\,\,\,\, \left.{\mathtt Q}_U\right|_{{\bf w}={\bf 0}}={\mathtt Q}_{U_0}.
\end{split}
\end{equation}

\begin{remark}
For a general deformed potential we do not have an explicit expression for the series $\tilde{R}'(z)$. However, it is constructed in
Section \ref{defmon} for the simplest deformation of the monomial potential.
\end{remark}

\begin{remark}
Similar to the polynomial case, for the deformed GKM we have the relations (\ref{Doug}) and the Douglas string equation (\ref{douglass}). However, in this case the operator $V'(L)=V'(L)_+$ includes the differential operators of arbitrary large positive order.
\end{remark}

\begin{lemma}\label{lemma_unique}
In the deformed case the quantum spectral curve equation 
\be\label{defQSC}
{\mathtt P}_U \cdot  \Psi_U=0
\ee
has a unique monic solution  $ \Psi_U \in{\mathbb C}((z^{-1}))[[{\bf w}]]$. This solution is of the form $\Psi=1+O(z^{-1})$. 
\end{lemma}
\begin{proof}
Let us construct the solution recursively as a series in ${\bf w}$. At ${\bf w}={\bf 0}$ equation (\ref{defQSC}) reduces to the quantum spectral curve equation for $U_0$, ${\mathtt P}_{U_0} \cdot  \Psi_{U_0}=0$, which according to Lemma \ref{QSCLem} has a unique monic solution. Let us consider the first terms of the series expansion of the wave function and quantum spectral curve operator: $ \Psi_{U;k}:=[w_k] \Psi_U$, ${\mathtt P}_{U;k}:=[w_k] {\mathtt P}_{U}$. Then $ \Psi_{U;k}$ is defined by
\be\label{auxwp}
{\mathtt P}_{U_0} \cdot \Psi_{U;k} + {\mathtt P}_{U;k}  \cdot\Psi_{U_0}=0.
\ee
Since $ {\mathtt P}_{U;k}  \in z^{-1}{\mathcal D}_-$ and $\Psi_{U_0}-1\in H_-$, 
\be
{\mathtt P}_{U;k}  \cdot\Psi_{U_0} \in z^{-1 }H_-.
\ee
The general solution of (\ref{auxwp}) is $\Psi_{U;k}=c_k\Psi_{U_0}+\tilde{\Psi}_{U;k}$, where $\tilde{\Psi}_{U;k}\in H_-$ . It is easy to see that all coefficients of the series expansion of $\Psi_U$ have the same form. To get monic solution we have to put $c_k=0$ for all $k$ as well as for all higher coefficients of expansion of $\Psi_U$.
\end{proof}
This Lemma is similar to Lemma \ref{QSCLem}. However, in the deformed case we are looking for solutions in ${\mathbb C}((z^{-1}))[[{\bf w}]]$, and it could happen that equation (\ref{defQSC})
has other solutions in ${\mathbb C}[[{\bf w}]][[z,z^{-1}]] $. The main goal of this lemma is to exclude these potential solutions.

Functions ${R}$ and $\tilde{R}$ in (\ref{Waved}) and (\ref{Pgeg}) are related to each other:
\begin{lemma}\label{lemma_onpsi}
\be\label{Rtotil}
1-\frac{c_1}{\hbar}\tilde{R}'(z)=e^{\frac{1}{\hbar}\left(R(z-c_1)-R(z)-V(z-c_1)+V(z)-c_1V'(z)\right)}.
\ee
\end{lemma}
\begin{proof}
Using integration by parts it is easy to show that if (\ref{Rtotil}) is true, then ${\mathtt P}_U$ annihilates $\Psi_U$:
\begin{equation}
\begin{split}
{\mathtt P}_U \cdot \Psi_U &={\mathtt P}_U\cdot e^{-\frac{S_0}{\hbar}-S_1} \int_{{\gamma}(U_0)} d\varphi \, 
e^{\frac{1}{\hbar}\left(\varphi V'(z)-V(\varphi)+R(\varphi)\right)} \\
&= \frac{e^{-\frac{S_0}{\hbar}-S_1}}{c_1} \int_{{\gamma}(U_0)} d\varphi \,e^{\frac{1}{\hbar}\left(\varphi V'(z)-V(\varphi)+R(\varphi)\right)}\left(\left(1-c_1\frac{\tilde{R}'(\varphi)}{\hbar}\right)e^{\frac{c_1}{\hbar}\left(V'(\varphi)-V'(z)\right)}-1\right)\\
 &=\frac{e^{-\frac{S_0}{\hbar}-S_1}}{c_1}\int_{{\gamma}(U_0)} d\varphi \left(e^{\frac{1}{\hbar}\left((\varphi-c_1) V'(z)-V(\varphi-c_1)+R(\varphi-c_1)\right)} -e^{\frac{1}{\hbar}\left(\varphi V'(z)-V(\varphi)+R(\varphi)\right)} \right)\\
 &=0.
\end{split}
\end{equation}
Here  in the third equality we use (\ref{Rtotil}). The last equality follows from the invariance of the asymptotic expansion of the integral under the small translations of the integration contour. The wave function $\Psi_U$  is a formal series in ${\mathbb C}((z^{-1}))[[{\bf w}]]$.  Thus, from Lemma \ref{lemma_unique} it follows that such solution is unique up to a constant factor. 
\end{proof}
In terms of $R(z)$ the modified quantum spectral curve equation is given by
\be
\hbar c_1^{-1}\left(e^{\frac{1}{\hbar}\left(R(\hat{y}-c_1)-R(\hat{y})-V(\hat{y}-c_1)+V(\hat{y})-c_1V'(\hat{y})\right)}  e^{\frac{c_1}{\hbar} \left(V'(\hat{y})-\hat{x}\right)}-1  \right)\cdot \tilde{\Psi}_U=0. 
\ee

We also have a matrix integral representation of the tau-function obtained using basis (\ref{basisnon})
\begin{equation}
\begin{split}
\tau_U([\Lambda^{-1}])&=\frac{\det_{i,j=1}^N\left<e^{\frac{1}{\hbar}R(\varphi+\lambda_j)} (\lambda_i+\varphi)^{k-1}\right>_U}{\Delta(\lambda)}\\
&={\mathcal C}^{-1}\int[ d\Phi] \exp\left(-\frac{1}{\hbar}\Tr(V(\Phi)-R(\Phi)-\Phi V'(\Lambda))\right),
\end{split}
\end{equation}
where we integrate over the normal matrices with eigenvalues on the contours $\gamma(U_0)$ and the normalization factor ${\mathcal C}^{-1}$ is given by (\ref{NormalC}). 

\begin{remark}
Expansion of this matrix integral and integral (\ref{intm}) below can be described in terms of the Feynman diagram technique, which is a deformation of the well-known diagram technique for the monomial GKM. 
The geometrical interpretation of the diagrams could be interesting.
\end{remark}


\subsection{Relation between deformed and pure GKMs}

In this section we prove a relation between tau-functions for the deformed and pure GKM tau-functions, corresponding to the potentials $U$ and $U_0$, given by Theorem \ref{taufr}. 
This relation generalizes one considered in Section \ref{S_eq}. For coefficients $a_k \in {\mathbf K}^*$, defined by (\ref{ffun}), we introduce an element of Heisenberg-Virasoro group
\be\label{groupel}
\tilde{\mathtt G}:= e^{\int^z (f(\eta)-\eta)U(\eta) d\eta}\,  e^{\sum_{k \in {\z}} a_k {\mathtt l}_k}. 
\ee

\begin{lemma}\label{deformc}
\be
{\mathcal W}_{U}= \tilde{\mathtt G} \cdot {\mathcal W}_{U_0}.
\ee
\end{lemma}
\begin{proof}
From the direct computation it follows that
\begin{equation}
\begin{split}
\tilde{\mathtt Q}_U&=\tilde{\mathtt G}\, {\mathtt Q}_{U_0}\, \tilde{\mathtt G}^{-1},\\
{\mathtt X}_U&=\tilde{\mathtt G}\, {\mathtt X}_{U_0}\, \tilde{\mathtt G}^{-1}.
\end{split}
\end{equation}
Then the statement  follows from Lemmas \ref{lem_unib} and \ref{deforlemma}.
\end{proof}

Again, we can rewrite the operator (\ref{groupel}) as 
\be\label{reloper}
\tilde{\mathtt G}= e^{-\frac{1}{\hbar}S_0(z)}  e^{\sum_{k \in {\z}} a_k {\mathtt l}_k } e^{\frac{1}{\hbar}\tilde{S}_0(z)},
\ee
where $\tilde{S}_0(z)$ and $S_0(z)$ are given by (\ref{unst1}) for $U_0$ and $U$ respectively. Let us stress that $\tilde{\mathtt G}\notin {\mathcal G}$. 

\begin{proof}[{\bf Proof of Theorem \ref{taufr}}]
Theorem follows from (\ref{1Vir}) and Lemma \ref{deformc}.
\end{proof}

\begin{conjecture}\label{ConC}
\be
{C}_U=1.
\ee
\end{conjecture}

For the deformation of the monomial potential the operator (\ref{defconjop}) simplifies. 
In this case both summations  in (\ref{defconjop}) are over ${\mathbb Z}_{>0}$, and $f(z)$ is given by (\ref{ffun1}).

\begin{remark}
Kontsevich-Witten tau-function can be described by a cut-and-join type operator $\widehat{W}$ \cite{CAJ},
\be
\tau_{KW}=e^{\hbar \widehat{W}_{KW}}\cdot 1,
\ee
where
\be
\widehat{W}_{KW}=\frac{1}{3}\sum_{\substack{k,m\geq 0}}\left(2k+1\right)\left(2m+1\right)t_{2k+1}t_{2m+1}\frac{\p}{\p t_{2k+2m-1}}\\
+\frac{1}{3!}\sum_{k,m\geq 0}\left(2k+2m+5\right)
t_{2k+2m+5}\frac{\p^2}{\p t_{2k+1}\p t_{2m+1}}+\frac{t_1^3}{3!}+\frac{t_3}{8}.
\ee
Hence, for the deformation of this case from Theorem \ref{taufr} one has
\be
\tau_U={C}_{KW} e^{\widehat{W}_{U}}\cdot 1,
\ee
where 
\be
\widehat{W}_{U}:= \widehat{{ G}} \, \widehat{W}_{KW} \, \widehat{{ G}} ^{-1}.
\ee
Here we use the relation $ \widehat{{ G}}\cdot 1 =0$, which is valid the monomial potential. Note, that operator $\widehat{W}_{KW}$ is cubic in $\widehat{J}_k$, so is $\widehat{W}_{U}$.

For higher $n$ the cut-and-join description of the deformed potential can be constructed in the same way from the non-deformed one described by Zhou in \cite{ZhouCaJ}. 
\end{remark}

When some of parameters $w_j$ in the deformed potential (\ref{gendefor}) coincide with each other, by the action of the operator $e^{a\widehat{L}_1}$ on the tau-function $\tau_U$  we can reduce the degree of the denominator. Namely,
\be
 e^{a {\mathtt l}_1}{\mathtt X}_Ue^{-a {\mathtt l}_1}=\hbar  \int_{0}^{z}\frac{{U}^a_0(z) d \eta}{(1+a\eta)\prod_{j=1}^{n+1}(1-(w_j-a)\eta)},
 \ee
where ${U}_0^a=(1+az)^nU_0(z/(1+az))$ is a polynomial of degree $n$. Let $w_{\tilde{k}}=w_{\tilde{k}+1}=\dots=w_{n+1}$ for some $\tilde{k}\leq n$, then
 \be
e^{w_{\tilde{k}} {\mathtt l}_1}{\mathtt X}_Ue^{-w_{\tilde{k}} {\mathtt l}_1}=\int_{0}^{z}\frac{{U}_0^{w_{\tilde{k}}}(\eta) d \eta}{(1+w_{\tilde{k}}\eta)\prod_{j=1}^{\tilde{k}-1}(1-(w_j-w_{\tilde{k}})\eta)}
\ee

For simplicity, let us consider the deformation of the monomial potential with $w_{\tilde{k}}=w_{\tilde{k}+1}=\dots=w_{n+1}$ for some $\tilde{k}\leq n$. Then 
\be
\tilde{U}_{\tilde{k}}=\frac{1}{\hbar}\frac{z^n}{(1+w_{\tilde k}z)\prod_{j=1}^{\tilde{k}-1}(1-(w_j-w_{\tilde{k}})z)},
\ee
and
\be
u_k=[z^k] \int^z \left(\frac{\eta}{1+w_{\tilde{k}}\eta}-\eta\right)\tilde{U}_{\tilde{k}}(\eta) d\eta.
\ee
\begin{lemma}\label{lem_roots}
\be
\tau_{\tilde{U}_{\tilde{k}}}= e^{\sum_{k=n+3}^\infty u_k \widehat{J}_k} e^{w_{\tilde{k}} {\widehat{L}}_1}\cdot \tau_U.
\ee
\end{lemma}
\begin{proof}
The lemma follows from the relation
\begin{equation}
\begin{split}
{\mathtt X}_{\tilde U}&=e^{\int^z \left(\frac{\eta}{1+w_{\tilde{k}}\eta}-\eta\right)\tilde{U}_{\tilde{k}}(\eta) d\eta }e^{w_{\tilde k}{\mathtt l}_1} \,{\mathtt X}_{ U}\, e^{-w_{\tilde k}{\mathtt l}_1}e^{-\int^z \left(\frac{\eta}{1+w_{\tilde{k}}\eta}-\eta\right)\tilde{U}_{\tilde{k}}(\eta) d\eta },\\
\tilde{{\mathtt Q}}_{\tilde U}&=e^{\int^z \left(\frac{\eta}{1+w_{\tilde{k}}\eta}-\eta\right)\tilde{U}_{\tilde{k}}(\eta) d\eta }e^{w_{\tilde k}{\mathtt l}_1}\, \tilde{{\mathtt Q}}_{ U} \,e^{-w_{\tilde k}{\mathtt l}_1}e^{-\int^z \left(\frac{\eta}{1+w_{\tilde{k}}\eta}-\eta\right)\tilde{U}_{\tilde{k}}(\eta) d\eta },
\end{split}
\end{equation}
and the uniqueness, given by Lemma \ref{deforlemma}.
\end{proof}


\subsection{Heisenberg-Virasoro constraints}

Heisenberg-Virasoro constraints for the deformed GKM can be obtained similarly to the constraints for the undeformed case, considered in Section \ref{redvirs}.

Namely, for the deformed potential (\ref{gendefor}) we consider the KS operators (\ref{KSj}) and (\ref{KSl0}) with ${\mathtt Q}_U$ substituted by $\tilde{\mathtt Q}_U$.
The coefficients of the expansion of these operators in the ${\mathtt j}_k$ and ${\mathtt l}_k$ basis are again given by (\ref{gcoefd}) and  (\ref{hcoefd}), where now the range of summation over the index $j$, for which the coefficients are non-trivial, is not limited from above. So, the operators  
\be
\widehat{J}_m^U:=\sum_{j=m}^{\infty} g_{m,j} \frac{\p}{\p t_{j}},\\
\widehat{L}_m^U:=\sum_{j=-\infty}^{\infty} h_{m,j} \widehat{L}_j  - \frac{1}{\hbar}\sum_{j=m+1}^{\infty} g_{m+1,j} \frac{\p}{\p t_{j+1}} + \theta_m,
\ee
where the constants $\theta_m$ are chosen to satisfy the commutation relations (\ref{constcom}),
are the deformations of the operators for the polynomial case
\begin{equation}
\begin{split}
\widehat{J}_m^U\Big|_{{\bf w}={\bf 0}} &= \widehat{J}_m^{U_0},\\
\widehat{L}_m^U\Big|_{{\bf w}={\bf 0}} &=\widehat{L}_m^{U_0}.
\end{split}
\end{equation}
Applying the KS arguments we conclude that these operators annihilate the deformed tau-function:
\begin{proposition}\label{lem_defvir}
The deformed tau-function with the potential (\ref{gendefor}) satisfies the constraints
\begin{equation}
\begin{split}
\widehat{J}_m^U\cdot \tau_U &=0 \,\,\,\,\,\,\,\,\,\,\,\, m\geq 1,\\
\widehat{L}_m^U\cdot \tau_U&=0 \,\,\,\,\,\,\,\,\,\,\,\, m\geq -1.
\end{split}
\end{equation}
\end{proposition}

If $U_0$ is a monomial, then 
 the constraints for the deformed case are the deformation of the constraints (\ref{monomvir}):
\be\label{constr_d}
\widehat{J}_m^U:=\frac{1}{(n+1)^m}\frac{\p}{\p t_{m (n+1)}}+ \sum_{j=m(n+1)+1}^{\infty} g_{m,j} \frac{\p}{\p t_{j}},\\
\widehat{L}_m^U:=\frac{1}{(n+1)^{m+1}}\left(\widehat{L}_{m(n+1)}  - \frac{1}{\hbar}\frac{\p}{\p t_{(m+1)(n+1)+1}} + \frac{n^2+2n}{24}\,\delta_{m,0}\right)\\
+\sum_{j=m(n+1)+1}^{\infty} h_{m,j} \widehat{L}_j    - \sum_{j=(m+1)(n+1)+1}^{\infty} g_{m+1,j} \frac{\p}{\p t_{j+1}}+\delta_{m,-1}\theta_{-1}.
\ee

\begin{lemma}\label{lem_theta}
For the  deformation of the monomial potential 
\be
\theta_{-1}=-\frac{1}{24}\left(\prod_{k=1}^{n+1}(-w_k)+\sum_{k=1}^{n+1}w_k \prod_{j\neq k}(w_j-w_k)\right).
\ee 
\end{lemma}
\begin{proof}
With integration by parts the last line of (\ref{thet}) for $m=-1$ can be simplified,
\be
\theta_{-1}=-\frac{1}{24}\res_{z=0} \frac{(V^{(3)})^2}{(V'')^3}.
\ee
For the deformation of the monomial potential one has 
\be
\frac{V^{(3)}}{V''}=\frac{\p}{\p z} \log U =\frac{n}{z} +\sum_{k=1}^{n+1} \frac{w_k}{1-w_k z}.
\ee
Hence we can rewrite the residue as a sum of the residues at $z=\infty$ and $z=\frac{1}{w_k}$,
\be
\res_{z=0} \frac{(V^{(3)})^2}{(V'')^3}=-\left(\res_{z=0}+\sum_{k=1}^{n+1} \res_{z=w_k^{-1}}\right) \frac{\prod_{k=1}^{n+1}(1-w_k z)}{z^n}\left(\frac{n}{z} +\sum_{k=1}^{n+1} \frac{w_k}{1-w_k z}\right)^2.
\ee
The residue at $z=\infty$ yields $\prod_{k=1}^{n+1}(-w_k)$, the residue at $z=w_k^{-1}$ yields $w_k \prod_{j\neq k}(w_k-w_j)$. This completes the proof.
\end{proof}

Let $e_k({\bf w})$ be the elementary symmetric functions of $w_j$,
\be
\sum_{k=0}^{n+1} e_k({\bf w}) z^k =\prod_{j=1}^{n+1}(1+w_iz).
\ee
The simplest Heisenberg-Virasoro constraints for the deformation of the monomial potential are given by the operators

\begin{equation}
\begin{split}\label{Stringd}
\widehat{J}_{1}^U&:=\sum_{k=0}^\infty \frac{h_k({\bf w})}{k+n+1} \frac{\p}{\p t_{k+n+1}},\\
\widehat{L}_{-1}^U&:= \widehat{L}_{-n-1}-e_1({\bf w})\widehat{L}_{-n}+\dots \pm e_{n+1}({\bf w})\widehat{L}_{0}-\frac{1}{\hbar}\frac{\p}{\p t_1} + \theta_{-1}.
\end{split}
\end{equation}
\begin{remark}
Using the KS operators it is also possible to construct a complete family of the $W^{(n+1)}$ constraints for the deformed potential.
\end{remark}


\subsection{Sato's group element}\label{czero}

In this section we assume that
 \be
 c_1=0
 \ee 
 or equivalently, that at least one of the $w_k$ vanishes. In this case the deformed GKM simplifies, in particular,  Lemma \ref{lemma_onpsi} yields
 \be
 R(z)=\tilde{R}(z)
 \ee
and the canonical pair of the KS operators is given by
\begin{equation}\label{PQredd}
\begin{split}
{\mathtt P}_U&=\frac{1}{\hbar}\left(V'(\tilde{\mathtt Q}_U)-{R}'(\tilde{\mathtt Q}_U)-{\mathtt X}_U\right),\\
{\mathtt Q}_U&=\tilde{\mathtt Q}_U-c_2 {\mathtt P}_U.
\end{split}
\end{equation}
The modified quantum spectral curve is given by
\be
\left(V'(\hat{y})-{R}'(\hat{y})-\hat{x} \right)\cdot \tilde{\Psi}_U=0
\ee
with the semi-classical limit
\be
V'(y)-\left.{R}'(y)\right|_{\hbar=0}=x.
\ee

Let
\be\label{Boper}
{\mathtt B}_U:= \left<e^{\frac{1}{\hbar}R(\varphi+z)} e^{\varphi \frac{\p}{\p z}}\right>_U e^{-\frac{c_2}{2}\frac{\p^2}{\p z^2}}.
\ee
At ${\bf w}={\bf 0}$ operator ${\mathtt B}_U$ reduces to the undeformed Sato's group element $\left< e^{\varphi \frac{\p}{\p z}}\right>_{U_0}$ (see Lemma \ref{lemma_Gel}.) The following operator identities hold true  for operators (\ref{PQredd}):
\begin{lemma}\label{lemmaBasG}
\begin{equation}
\begin{split}
{\mathtt P}_U\,{\mathtt B}_U&={\mathtt B}_U\, \frac{\p}{\p z},\\
{\mathtt Q}_U\,{\mathtt B}_U&={\mathtt B}_U\, z.
\end{split}
\end{equation}
\end{lemma}

\begin{proof} 
From the definition (\ref{aver}) we have
\be
\tilde{\mathtt Q}_U \left< \dots \right>_U=\left<\left(\frac{1}{U(z)}\left(\frac{\p}{\p z}+ \frac{V'(z)-V'(z+\varphi)}{\hbar }\right) +z +\varphi \right) \dots \right>_U,
\ee
hence
\begin{equation}
\begin{split}
&\tilde{\mathtt Q}_U{\mathtt B}_U-{\mathtt B}_U \left(z+c_2\frac{\p}{\p z}\right)
=\tilde{\mathtt Q}_U  \left<e^{\frac{1}{\hbar}R(\varphi+z)} e^{\varphi \frac{\p}{\p z}}\right>_U e^{-\frac{c_2}{2}\frac{\p^2}{\p z^2}}
-  \left<e^{\frac{1}{\hbar}R(\varphi+z)} e^{\varphi \frac{\p}{\p z}} z\right>_U e^{-\frac{c_2}{2}\frac{\p^2}{\p z^2}} \\
&=\frac{1}{U(z)}\left< e^{\frac{1}{\hbar}R(\varphi+z)}  \left(\frac{\p}{\p z}+\frac{V'(z)- V'(z+\varphi) +R'(z+\varphi)}{\hbar}\right)e^{\varphi \frac{\p}{\p z}}\right>_{U} e^{-\frac{c_2}{2}\frac{\p^2}{\p z^2}}\\
&=0,
\end{split}
\end{equation}
where in the last line we have used the integration by parts in (\ref{Boper}). Using this relation we have
\begin{equation}
\begin{split}
{\mathtt P}_U\,{\mathtt B}_U-{\mathtt B}_U\, \frac{\p}{\p z}& 
=-\frac{1}{\hbar}{\mathtt X}_U{\mathtt B}_U +{\mathtt X}_U\, \left( \frac{V'(z+c_2\frac{\p}{\p z}) -R'(z+c_2\frac{\p}{\p z})}{\hbar}-\frac{\p}{\p z}\right)\\
&=-\left< e^{\frac{1}{\hbar}R(\varphi+z)}  \left(\frac{\p}{\p z}+\frac{V'(z)- V'(z+\varphi)+R'(z+\varphi)}{\hbar}\right)e^{\varphi \frac{\p}{\p z}}\right>_{U} e^{-\frac{c_2}{2}\frac{\p^2}{\p z^2}}\\
&=0,
\end{split}
\end{equation}
where the last line again follows from the integration by parts. The second identity of Lemma follows from linear combination of these two, this completes the proof.
\end{proof}

We see, that ${\mathtt B}_U$ defines the canonical KS operators for ${\mathcal W}_U$ by (\ref{PQdef}). Let us prove that  ${\mathtt X}_U$ is indeed the Sato group operator.
\begin{proposition}\label{Th_BasG}
Operator ${\mathtt B}_U$ is a Sato group operator for the deformed potential:
\be
{\mathtt G}_U={\mathtt B}_U.
\ee
\end{proposition}
\begin{proof}
Consider the expansion ${\mathtt B}_U=\sum_{k=0}^\infty {\mathtt B}_k \frac{\p^k}{\p z^k}$. From the comparison of (\ref{Waved}) and (\ref{Boper}) we see that 
\be\label{inc}
{\mathtt B}_0=\Psi_U.
\ee 
The statement of the proposition follows from the initial condition (\ref{inc}) and Lemmas \ref{lemmaBasG} and \ref{l_GuniQ}.
\end{proof}

Let $\He_k(z)$ be the Hermite polynomials
\be
\He_k(z):=e^{-\frac{1}{2}\frac{\p^2}{\p z^2}} \cdot z^k.
\ee
Then the following corollary immediately follows
\begin{corollary}
For the deformed potential with $c_1=0$ the distinguished basis is given by
\be
\check{\Phi}_k^U=  \left<e^{\frac{1}{\hbar}R(\varphi+z)} c_2^{(k-1)/2}\He_{k-1} \left(\frac{z+\varphi}{\sqrt{c_2}}\right)\right>_U.
\ee
\end{corollary}
In particular, if $c_2=0$
\be\label{canbss}
\check{\Phi}_k^U=  \left<e^{\frac{1}{\hbar}R(\varphi+z)} (z+\varphi)^{k-1}\right>_U.
\ee

\begin{remark}
Because the operators ${\mathtt X}_U$ and ${\tilde {\mathtt Q}}_U$ uniquely define the point of the Sato Grassmannian, (\ref{Gconj}) is true for the deformed GKM. 
Then one can use Proposition \ref{Th_BasG} and Proposition \ref{Kernel_pr} to derive an integral representation for the Cauchy-Baker-Akhieser kernel for the deformed case with $c_1=0$, to therefore get affine coordinates of the point of the Sato Grassmannian. 
\end{remark}


\subsection{Deformation of the monomial GKM}\label{defmon}

In this section we consider the simplest 1-parameter deformation of the monomial potentials (\ref{MON}):
\be\label{nmondef}
U=\frac{1}{\hbar}\frac{z^n}{1-wz}.
\ee
For this deformation $c_1=0$, so the results of the previous section are applicable. It remains to find $R(z)$.
The KS operators $\tilde{\mathtt Q}_U$ and ${\mathtt X}_U$ are given by
\begin{equation}
\begin{split}
\tilde{\mathtt Q}_U&=z+\hbar\left(\frac{1-wz}{z^n}\frac{\p}{\p z} -\frac{1-wz}{2z^n}\left(\frac{n}{z}+\frac{w}{1-wz}\right)\right),\\
{\mathtt X}_U&=-\frac{1}{w^{n+1}}\left(\log(1-wz)+\sum_{k=1}^n \frac{(wz)^k}{k}\right).
\end{split}
\end{equation}

Let us define
\be
{\mathcal N}:=z^{-1} {\mathcal D}_-  \cap H_-[\hbar, w,\frac{1-wz}{z}\frac{\p}{\p z}].
\ee
It is closed under commutator, $\left[{\mathcal N},{\mathcal N}\right] \subset {\mathcal N}$. Let 
\be
\chi:=w^{n+1}\hbar.
\ee
We  define ${\mathtt \Omega}$ by
\be\label{Rex}
e^{w^{n+1}{\mathtt X}_U+\sum_{k=1}^n \frac{(w\tilde{\mathtt Q}_U)^k}{k}}=e^{-\chi\frac{\p}{\p z}+\delta_{n,1}\hbar\frac{w}{z}\frac{\p}{\p z}-\frac{\chi}{2z}+{\mathtt \Omega}}\,\frac{1}{1-wz}.
\ee

\begin{lemma}
\be
{\mathtt \Omega} \in z^{-1}{\mathcal D}_-.
\ee
\end{lemma}
\begin{proof}
Let
\be
\tilde {\mathtt \Omega} :=
w^{n+1}{\mathtt X}_U+\sum_{k=1}^n \frac{(w\tilde{\mathtt Q}_U)^k}{k}+\log(1-wz).
\ee
Then a straightforward calculation yields $\tilde {\mathtt \Omega}-\hbar w^n(1+\frac{1-\delta_{n,1}}{wz})\frac{1-wz}{z}\frac{\p}{\p z}  \in {\mathcal N}$. 
Let us apply the Baker-Campbell-Hausdorff formula to the left hand side of (\ref{Rex}).
\begin{equation}
\begin{split}\label{prott}
e^{w^{n+1}{\mathtt X}_U+\sum_{k=1}^n \frac{(w\tilde{\mathtt Q}_U)^k}{k}}&=
e^{\tilde {\mathtt \Omega}-\log(1-wz)}\\
&=e^{\tilde {\mathtt \Omega}+\frac{1}{2}\left[\tilde {\mathtt \Omega},\log(1-wz)\right]+\dots}\,\frac{1}{1-wz},
\end{split}
\end{equation}
where by $\dots$ we denote a combination of the nested commutators of $\log(1-wz)$ and $ \tilde {\mathtt \Omega}$.  
For the simplest commutator we have
\be
\left[\tilde {\mathtt \Omega},\log(1-wz)\right]\in{\mathcal N}- \frac{\chi}{z}.
\ee
 Again, straightforward caclulations show that $\left[{\mathcal N},\log(1-wz)\right] \subset {\mathcal N}$,  $\left[{\mathcal N},\tilde {\mathtt \Omega}\right] \subset {\mathcal N}$,
and
\begin{equation}
\begin{split}
\left[\tilde {\mathtt \Omega},\left[\tilde {\mathtt \Omega},\log(1-wz)\right]\right] &\in {\mathcal N},\\
\left[\log(1-wz),\left[\tilde {\mathtt \Omega},\log(1-wz)\right]\right] &\in {\mathcal N}.
\end{split}
\end{equation}
By induction all higher nested commutators belong to ${\mathcal N}$, and hence to $z^{-1} {\mathcal D}_-$. The statement of the lemma follows from the comparison of (\ref{Rex}) and (\ref{prott}).
\end{proof}

Consider the KS operator
\be
1-w \tilde{\mathtt Q}_U-\frac{\chi w}{2}
=(1-wz)\left(1-\hbar\left(\frac{w}{z^n}\frac{\p}{\p z}-\frac{wn}{2z^{n+1}}-\frac{w^2}{2z^n}\sum_{j=1}^{n-1}(wz)^j\right)\right).
\ee
From (\ref{Rex}) one immediately has
\begin{equation}\label{Slozhnop}
\begin{split}
&e^{w^{n+1}{\mathtt X}_U+\sum_{k=1}^n \frac{(w\tilde{\mathtt Q}_U)^k}{k}}\left(1-w \tilde{\mathtt Q}_U-\frac{\chi w}{2}\right)\\
&=e^{-\chi\frac{\p}{\p z}+\delta_{n,1}\hbar\frac{w}{z}\frac{\p}{\p z}-\frac{\chi}{2z}+\dots}
\left(1-\delta_{n,1}\hbar\frac{w}{z}\frac{\p}{\p z}+\frac{\chi}{2z}+\dots\right)\\
&=e^{-\chi\frac{\p}{\p z}+\dots}
\end{split}
\end{equation}
where by $\dots$ we denote the terms from $z^{-1}{\mathcal D}_-$. Hence, from the uniqueness of the quantum spectral curve operator we have
\be\label{Plog}
{\mathtt P}_U=-\frac{1}{\chi}\log\left(e^{w^{n+1}{\mathtt X}_U+\sum_{k=1}^n \frac{(w\tilde{\mathtt Q}_U)^k}{k}}\left(1-w \tilde{\mathtt Q}_U-\frac{\chi w}{2}\right)\right).
\ee
From the commutation relation between $\tilde{\mathtt Q}_U$ and ${\mathtt X}_U$ it follows, that this operator is indeed of the form given by the first equation of (\ref{PQredd}). We will derive $R'(z)$ from the comparison of these two expressions. 

We start from a simple implication of the Baker-Campbell-Hausdorff formula. 
\begin{lemma}\label{commlem}
If $[a,b]=1$, then for any function $f$
\be
e^{a+f'(b)}= e^{a} e^{f(b)-f(b-1)}
\ee
\end{lemma}
\begin{proof}
\be
e^{a} e^{f(b)-f(b-1)}=e^{-f(b)}e^{a}e^{f(b)}=e^{e^{-f(b)}a e^{f(b)}}=e^{a+f'(b)}.
\ee
\end{proof}

Let
\be
F(z):=\frac{1}{w^{n+1}}\sum_{k=1}^\infty \frac{\left(-\chi w\right)^k}{k}B_k\left(\frac{1}{2}-\frac{z}{\chi}\right),
\ee
where $B_n$ are the Bernoulli polynomials, satisfying
\be
B_n(x+1)-B_n(x)=n x^{n-1}.
\ee
\begin{proposition}\label{Theor1}
For the deformation of the monomial potential (\ref{nmondef}) the quantum spectral curve operator is given by
\be
{\mathtt P}_U=\frac{1}{\hbar}\left(F(\tilde{\mathtt Q}_U)-\frac{1}{w^{n+1}}\sum_{k=1}^n \frac{(w\tilde{\mathtt Q}_U)^k}{k}-{\mathtt X}_U\right).
\ee
\end{proposition}
\begin{proof}
We apply Lemma \ref{commlem}. Namely, let $a:=w^{n+1} {\mathtt X}_U +\sum_{k=1}^n \frac{(w\tilde{\mathtt Q}_U)^k}{k}$ and $b:=-\frac{\tilde{\mathtt Q}_U}{\chi}-\frac{1}{2}$. Then $[a,b]=1$ and we can rewrite the left hand side of (\ref{Slozhnop}) as
\be
e^a(1+\chi wb)=e^a\,  e^{\sum_{k=1}^\infty (-1)^{k+1}\frac{(\chi w b)^k}{k}}.
\ee
To apply Lemma \ref{commlem} we have to solve the difference equation
\be
e^{f(b)-f(b-1)}=e^{\sum_{k=1}^\infty (-1)^{k+1}\frac{(\chi w b)^k}{k}}.
\ee
The solution is given by a sum of the Bernoulli polynomials: 
\be
 f(b)=\sum_{k=1}^\infty (-1)^{k+1}\frac{(\chi w)^k}{k(k+1)} B_{k+1}(b+1) +2\pi i m x
\ee
for some $m \in {\mathbb Z}$. Note that the series $f(x)$ is unique up to any periodic function $f(x)$ such that $f(x) = f(x + 1)$. However, if we require the coefficient of each power of $w$
to be a polynomial in $x$, we can only add a linear term. To find $f'$ one can apply $B'_k(x)=kB_{k-1}(x)$:
\be
f'(b)=\sum_{k=1}^\infty (-1)^{k+1}\frac{(\chi w )^k}{k} B_{k}(b+1)+2 \pi i m.
\ee
From the comparison of the result with case $w=0$ we conclude, that the linear term must vanish, $m=0$. Hence
\be\label{Pmon}
{R}'(z)=-F(z)-\frac{1}{w^{n+1}}\log(1-wz).
\ee
This completes the proof of the proposition.
\end{proof}

Using (\ref{Plog}) we can write the equation for the modified wave function as (\ref{pochqq}).
The classical spectral curve
\be
- w^{-n-1}\left( \log \left(1-w y\right)+  \sum_{k=1}^n\frac{w^k}{k} y^k\right)=x
\ee
reduces to (\ref{rspinc}) at $w=0$. Being exponentiated, it can be rewritten as
\be
e^{\sum_{k=1}^n\frac{w^k}{k}{y}^k}\left(1-w{y}\right)=e^{-w^{n+1}x}.
\ee
\begin{remark}
These classical and quantum spectral curves are similar to the curves for simple Hurwitz numbers \cite{ZhouQSC} and r-spin Hurwitz numbers \cite{MSS} (for $r=n+1>2$). In particular,  for $n=1$, after change of variables $x\mapsto -(1+x)/w^2$, $y\mapsto (1-y)/w$ the classical spectral curve coincides with the Lambert curve $x=y^{-y}$.
\end{remark}

Let us consider integral representation  (\ref{Waved}) of the wave function. Namely, from Proposition \ref{Theor1} we see that
\be
\frac{1}{\hbar}\left(R(\varphi)-V(\varphi)\right)=\sum_{k=1}^\infty \frac{(-\chi w)^k}{k(k+1)}B_{k+1}\left(\frac{1}{2}-\frac{\varphi}{\chi}\right)+\frac{1}{\chi}\sum_{k=1}^n \frac{w^k\varphi^{k+1}}{k(k+1)}+\dots,
\ee
where by $\dots$ we denote terms independent of $\varphi$.
Comparing it to Stirling's expansion of the gamma function
\be\label{Gammaas}
\Gamma(z+1/2-n)=\sqrt{2\pi}z^{z-n}e^{-z}e^{\sum_{k=1}^\infty \frac{B_{k+1}(n+1/2)}{(k+1)kz^k}},
\ee
valid for large values of $|z|$ with $|\arg(z)|<\pi$ and an arbitrary finite complex number $n$, we conclude that
\be
e^{\frac{1}{\hbar}\left(R(\varphi)-V(\varphi)\right)}=\alpha\, e^{\frac{1}{\chi}\left(\sum_{k=1}^n \frac{w^k\varphi^{k+1}}{k(k+1)}+\varphi\log(-\chi w)\right)}\frac{\Gamma\left(\frac{1}{2}+\frac{\varphi}{\chi}-\frac{1}{\chi w}\right)}{\Gamma\left(\frac{1}{2}-\frac{1}{\chi w}\right)}.
\ee
So, the wave function is given by
\be\label{QRT}
\Psi_U=\alpha \frac{e^{-\frac{S_0}{\hbar}-S_1}}{\sqrt{2 \pi \hbar}}\int_{\gamma(U_0)}d \varphi \, e^{\frac{1}{\hbar}V'(z)\varphi+\frac{1}{\chi}\left(\sum_{k=1}^n \frac{w^k\varphi^{k+1}}{k(k+1)}+\varphi\log(-\chi w)\right)}\frac{\Gamma\left(\frac{1}{2}+\frac{\varphi}{\chi}-\frac{1}{\chi w}\right)}{\Gamma\left(\frac{1}{2}-\frac{1}{\chi w}\right)}.
\ee
For $n>1$ the distingushed basis is given by (\ref{canbss}):
\be\label{disb_dd}
\check{\Phi}^U_k=\alpha \frac{e^{-\frac{S_0}{\hbar}-S_1}}{\sqrt{2 \pi \hbar}}\int_{\gamma(U_0)}d \varphi \,\varphi^{k-1}\, e^{\frac{1}{\hbar}V'(z)\varphi+\frac{1}{\chi}\left(\sum_{k=1}^n \frac{w^k\varphi^{k+1}}{k(k+1)}+\varphi\log(-\chi w)\right)}\frac{\Gamma\left(\frac{1}{2}+\frac{\varphi}{\chi}-\frac{1}{\chi w}\right)}{\Gamma\left(\frac{1}{2}-\frac{1}{\chi w}\right)}.
\ee
The case with $n=1$ will be considered in Section \ref{S_DifK1}.

Simplest operators of the Heisenberg-Virasoro constraints (\ref{Stringd}) for deformation (\ref{nmondef}) are given by
\begin{equation}
\begin{split}\label{strings}
\widehat{J}^U_1&=\sum_{k=n+1}^\infty \frac{w^{k-n-1}}{k}\frac{\p}{\p t_{k}},\\
\widehat{L}_{-1}^U&= \widehat{L}_{-n-1}-w\widehat{L}_{-n}-\frac{1}{\hbar}\frac{\p}{\p t_1} -\frac{w^{n+1}}{24}.
\end{split}
\end{equation}
They annihilate the tau-function $\tau_U$.


\subsection{Dilaton equation}

In addition to the Heisenberg-Virasoro constraints tau-function of the deformed GKM $\tau_U$ and associated point of the Sato Grassmannian ${\mathcal W}_U$ satisfy some other, independent constraints. Let
\be
E_w:=\sum_{j=1}^{n+1}w_j\frac{\p}{\p w_j}
\ee
be the Euler operator. Let us introduce operators
\begin{equation}
\begin{split}\label{genks}
\Theta_1&=z\frac{\p}{\p z} - E_w + \sum_{j=1}^{n-1} (n-j) b_{j}\frac{\p}{\p b_{j}}+\frac{n+2}{\hbar }(zV'(z)-V(z)),\\
\Theta_2&=\hbar \frac{\p}{\p \hbar}+\frac{1}{\hbar}(V(z)-zV'(z))
\end{split}
\end{equation}
with $\left[\Theta_1,\Theta_2\right]=0$. We have
\begin{equation}
\begin{split}
\left[\Theta_1,{\mathtt X}_U\right]&=(n+1){\mathtt X}_U,\\
\left[\Theta_1,\tilde{\mathtt Q}_U\right]&=-(n+1)\tilde{\mathtt Q}_U
\end{split}
\end{equation}
and
\begin{equation}
\begin{split}
\left[\Theta_2,{\mathtt X}_U\right]&=0,\\
\left[\Theta_2,\tilde{\mathtt Q}_U\right]&=\tilde{\mathtt Q}_U.
\end{split}
\end{equation}

\begin{proposition}
The operators $\Theta_k$ for $k=1,2$ stabilize the point of the Sato Grassmannian for the deformed GKM, $\Theta_k\cdot {\mathcal W}_U \subset  {\mathcal W}_U$.
\end{proposition}

\begin{proof}
For the potential (\ref{gendefor}) we have
\be
\Theta_1\cdot U(z) = n U(z)
\ee
and, more generally, 
\be\label{comc}
\Theta_1\cdot V^{(k)}(z)= (n+2-k) V^{(k)}(z).
\ee
The statement of proposition follows directly from the application of the operators to the basis (\ref{phidef}) and from the integration by parts.
\end{proof}

\begin{remark}
Operators $\Theta_k$ contain derivatives with respect to $w_j$, $b_j$, and $\hbar$ so they are not Kac-Schwarz operators. 
\end{remark}

We derive the {\em dilaton equation} associated with $\Theta_1$ for the deformation of the monomial potential using Theorem \ref{taufr}. Similar idea was used by Guo and Wang in derivation of the dilation equation for the linear Hodge integrals in \cite{Wang2}.
Let 
\be
v_k:=(n+2)[z^{k}]\left(z V'(z) -V(z)\right)=(n+2)\hbar[z^{k}]\int^z \eta U(\eta) d\eta.
\ee
\begin{proposition}
\be\label{dile}
\left(\widehat{L}_0  +E_w-\frac{1}{\hbar}\sum_{k=n+2}^\infty  v_k \frac{\p}{\p t_{k}}+\frac{n^2+2n}{24}\right)\cdot \tau_U({\bf t})= \left(E_w \cdot \log {C}_U\right)\tau_U({\bf t}).
\ee
\end{proposition}
\begin{proof}
By construction $\tau_{U_0}$ does not depend on $w_i$. Hence $E_w\cdot \tau_{U_0}=0$. Combining it with Proposition \ref{Lem_virpol} one gets
\be
\left(E_w+(n+1) \widehat{L}_0^{U_0}\right)\cdot \tau_{U_0}=0.
\ee
From Theorem \ref{taufr} it follows that
\be
\left({C}_U\,  \widehat{ G} \left(E_w+(n+1)\widehat{L}_0^{U_0}\right)  \widehat{ G}^{-1} {C}_U^{-1} \right)\cdot \tau_{U}=0.
\ee
Let us find the operator acting in this formula. First, we conjugate $E_w+(n+1) \widehat{L}_0^{U_0}$ with the Virasoro part of (\ref{defconjop}). From the definition of $f(z)$ it follows that for a monomial $U_0$ coefficients $a_k$ are homogeneous (and symmetric) functions of $w_k$ of degree $k$. Hence
\be
\left[E_w+ \widehat{L}_0,\sum_{k\in \z_{>0}} a_k \widehat{L}_k\right]=0
\ee
and
\be\label{firstconj}
e^{\sum_{k \in {\z}_{>0}} a_k \widehat{L}_k }  \left(E_w+\widehat{L}_0-\frac{1}{\hbar}\frac{\p}{\p t_{n+2}}+\frac{n^2+2n}{24}\right) e^{-\sum_{k \in {\z}_{>0}} a_k \widehat{L}_k } =\\
=E_w+\widehat{L}_0-\frac{1}{\hbar}\sum_{k=n+2}^\infty \tilde{v}_k \frac{\p}{\p t_{k}}+\frac{n^2+2n}{24},
\ee
where
\be
\tilde{v}_k:=[z^{k}]f(z)^{n+2}.
\ee
Now we conjugate it with $e^{\sum_{k \in {\z}_{>0}}  u_k \widehat{J}_k} $ in (\ref{defconjop}). The last two terms of the last line in (\ref{firstconj}) commute with this operator, so we only need
\be
e^{\sum_{k \in {\z}}  u_k \widehat{J}_k}   \left(E_w+\widehat{L}_0\right)   e^{- \sum_{k \in {\z}}  u_k \widehat{J}_k} = E_w+\widehat{L}_0 +\frac{1}{\hbar} \sum_{k=n+2}^\infty \left( \tilde{v}_k -v_k\right)\frac{\p}{\p t_{k}}.
\ee
Since ${C}_U$ does not depend on ${\bf t}$, only $E_w$ can contribute in conjugation with  ${C}_U$, which concludes the proof. 
\end{proof}
In particular, if Conjecture \ref{ConC} is true, then the right hand side of (\ref{dile}) vanishes.  Vice versa, if the right hand side of (\ref{dile}) vanishes, then Conjecture \ref{ConC} is true for the deformation of the monomial potential.


\subsection{Deformation of Kontsevich model}\label{DefK}
The simplest example of the deformed GKM is given by deformation of the Kontsevich model with $n=1$ and $U_0=z$, considered in Section \ref{Sec_Kon}.
In this case (\ref{gendefor}) yields
\be\label{genKond}
U=\frac{1}{\hbar}\frac{z}{(1-w_1z)(1-w_2z)}=\frac{1}{\hbar}\sum_{k=0}^\infty h_k(w_1,w_2)z^{k+1},
\ee 
where $h_k$ are complete symmetric functions. Then
\begin{equation}
\begin{split}
\tilde{\mathtt Q}_U&=    z+ \hbar\left(\frac{(1-w_1z)(1-w_2z)}{z}\frac{\p}{\p z} - \frac{1-w_1w_2z^2}{2z^2}\right)\\
&={\mathtt Q}_{KW}+\hbar\left(\left(w_1w_2z-w_1-w_2\right)\frac{\p}{\p z}+\frac{w_1w_2}{2}\right).
\end{split}
\end{equation}
If $w_1\neq0$, $w_2\neq 0$ and $w_1\neq w_2$, 
\be\label{bdefK}
{\mathtt X}_U=\hbar \int_0^z U(\eta) d\eta =\frac{w_1\log(1-w_2z)-w_2\log(1-w_1z)}{w_1w_2(w_1-w_2)}.
\ee

\begin{remark}
We can always rewrite the operator $\tilde{\mathtt G}$ in (\ref{reloper}) as
\be
\tilde{\mathtt G}= e^{\sum_{k \in {\z}} a_k {\mathtt l}_k } e^{\frac{1}{\hbar}\left(\tilde{S}_0(z)-S_0(f^{-1}(z))\right)}.
\ee
Then, for deformation of the  Kontsevich model using the reduction constraints and Virasoro constraints (\ref{monomvir}) we can rewrite the relation of Theorem \ref{taufr} 
in terms of Virasoro operators only:
\be
\tau_U=C_U e^{\sum_{k \in {\z}_{>0}} \tilde{a}_k \widehat{L}_k } \tau_{KW}.
\ee
\end{remark}

The Heisenberg-Virasoro constraints are described by Proposition \ref{lem_defvir}, and the simplest operators, that describe reduction and string equations, are given by
\begin{equation}
\begin{split}
{\widehat{J}}_1^U&=\frac{\p}{\p t_2} +\sum_{k=3}^\infty \frac{h_{k-2}(w_1,w_2)}{k}\frac{\p}{\p t_k},\\
{\widehat{L}}_{-1}^U&=\widehat{L}_{-2}-(w_1+w_2)\widehat{L}_{-1}+w_1w_2 \widehat{L}_{0} -\frac{1}{\hbar}\frac{\p}{\p t_1}-\frac{w_1^2+w_2^2-w_1w_2}{24}.
\end{split}
\end{equation}
The Heisenberg-Virasoro constraints completely specify the deformed tau-function, so that the correlation functions can be obtained by their solution in the naive topological recursion way.

From (\ref{bdefK}) it follows, that there are two special cases:
\begin{itemize}
\item $w_1=w_2=w$. In this case
\begin{equation}
\begin{split}
{\mathtt X}_U&=\frac{1}{w^2(1-wz)}+\frac{1}{w^2}\log(1-wz),\\
\tilde{\mathtt Q}_U&= z+\hbar\left(\frac{(1-wz)^2}{z}\frac{\p}{\p z}-\frac{1-w^2z^2}{2z^2}\right).
\end{split}
\end{equation}
By Lemma \ref{lem_roots} it can be reduced to the next case by the linear change and a shift of variables.

\item $w_2=0$ (or $w_1=0$). The model depends on one deformation parameter $w$. We consider this case in the next section.

\end{itemize}


\subsection{Simplest deformation of the Kontsevich model}\label{S_DifK1}
Let us consider
\be\label{Wlinear}
U=\frac{1}{\hbar}\frac{z}{1-wz}
\ee
with
\be\label{Vforlinear}
V(z)=\frac{(1-wz)\log(1-wz)}{w^3}+\frac{2z-w z^2}{2 w^2}=\sum_{k=3}^\infty \frac{z^kw^{k-3}}{k(k-1)}.
\ee
This case is the simplest example of the family of the deformed potentials, considered in Section \ref{defmon}.
Then
\begin{equation}
\begin{split}\label{Simpbq}
{\mathtt X}_U&=-\frac{z}{w}-\frac{\log(1-wz)}{w^2}=\frac{z^2}{2}+\frac{1}{w^2}\sum_{k=3}^\infty \frac{(wz)^k}{k},\\
\tilde{\mathtt Q}_U&=z+\hbar\left(\left(\frac{1}{z}-w\right)\frac{\p}{\p z} -\frac{1}{2z^2}\right)\\
&={\mathtt Q}_{KW}-\hbar w \frac{\p}{\p z}.
\end{split}
\end{equation}

Equation (\ref{pochqq}) for the modified wave function reduces to 
\be\label{qscdefk}
e^{w^{2}\hat{x} +w\hat{y}}\left(1-w\hat{y}-\hbar \frac{w^{3}}{2}\right)\cdot \tilde{\Psi}=\tilde{\Psi}.
\ee
 The classical spectral curve 
\be
x = -w^{-2}\left( \log \left(1-w y\right)+wy \right)
\ee
reduces to (\ref{KWcsc}) at $w=0$. It can be rewritten as
\be
e^{-w^2 x}=(1-wy)e^{wy}.
\ee

\begin{remark}
We expect that the same quantum spectral curve can be obtained by the conjugation of the KS operators
for the single simple Hurwitz generating function, described in \cite{A1}. 
\end{remark}
The wave function is given by (\ref{QRT}),
\be\label{psidefk}
\Psi_U=\alpha \frac{e^{-\frac{S_0}{\hbar}-S_1}}{\sqrt{2 \pi \hbar}}\int_{\gamma(U_0)}d \varphi \, e^{\frac{1}{\hbar}V'(z)\varphi+\frac{1}{w^2\hbar}\left( \frac{w\varphi^{2}}{2}+\varphi\log(-\hbar w^3)\right)}\frac{\Gamma\left(\frac{1}{2}+\frac{\varphi}{\hbar w^2}-\hbar w^3\right)}{\Gamma\left(\frac{1}{2}-\hbar w^3\right)}.
\ee
In this case we can also find a simpler integral expression for the wave function. Namely, after Fourier  transform, the differential-difference equation (\ref{qscdefk}) reduces to a first order differential equation. Let
\be
{\tilde{\Psi}}^*:=e^{\frac{wx^2}{2\hbar}}\tilde{\Psi}.
\ee
From equation (\ref{qscdefk}) one has
\be\label{qqq}
e^{w\hat{y}}\left(1-w\hat{y}+w^2 \hat{x}-\hbar \frac{w^{3}}{2}\right)\tilde{\Psi}^*=\tilde{\Psi}^*.
\ee
With the ansatz 
\be
\tilde{\Psi}^*=\frac{1}{\sqrt{2\pi\hbar}} \int d \varphi\, e^{\frac{1}{\hbar}(x\varphi-A(\varphi))}
\ee
the equation (\ref{qqq}) leads to the differential equation for $A(\varphi)$
\be
e^{-w\varphi}=1-w\varphi-\frac{\hbar w^3}{2}+w^2\frac{\p A}{\p \varphi}
\ee
with a general solution
\be
A(\varphi)=-\frac{e^{-w\varphi}-1+w\varphi -w^2\varphi^2/2}{w^3}+\frac{\hbar w \varphi}{2}+A_0.
\ee
Here $A_0$ does not depend on $\varphi$. Hence
\be\label{Psisimm}
\tilde{\Psi}=e^{-\frac{wx^2}{2\hbar}-\frac{A_0}{\hbar}}\int_{{\gamma}(U_0)}  d \varphi\, e^{\frac{1}{\hbar}\left(x\varphi+\frac{e^{-w\varphi}-1+w\varphi -w^2\varphi^2/2}{w^3}-\frac{\hbar w \varphi}{2}\right)}.
\ee
Then it is easy to see that 
\be
e^{-\frac{S_0}{\hbar}-S_1-\frac{wx^2}{2\hbar}-\frac{A_0}{\hbar}}\int_{{\gamma}(U_0)}  d \varphi \,e^{\frac{1}{\hbar}\left(x\varphi+\frac{e^{-w\varphi}-1+w\varphi -w^2\varphi^2/2}{w^3}-\frac{\hbar w \varphi}{2}\right)}  \in {\mathbb C}((z^{-1}))[[{\bf w}]].
\ee
From Lemma \ref{lemma_unique} it follows that this is indeed the wave function, and (\ref{Psisimm}) is a modified wave function. 
\begin{remark}
The same expression can be derived from (\ref{psidefk}), if one uses the standard integral expression for the gamma-function
\be
\Gamma(z)=\int_{\rr} d \varphi \, e^{\varphi z -e^\varphi}.
\ee
\end{remark}
At $w=0$ the modified wave function (\ref{Psisimm}) reduces to
\be
e^{-\frac{1}{\hbar}\left.A_0\right|_{w=0}} \frac{1}{\sqrt{2\pi\hbar}} \int_{\tilde{\gamma}(U_0)}  d \varphi\, e^{\frac{1}{\hbar}(x\varphi-\frac{\varphi^3}{3!})}
\ee
and from comparison with (\ref{defpsi}) we see that $\left.A_0\right|_{w=0}=0$.

Higher basis vectors can be obtained by the action of the operator ${\tilde{\mathtt Q}_U}$ with the help of integration by parts
\be\label{basisdkw}
\Phi_k^U=e^{-\frac{S_0}{\hbar}-S_1-\frac{wx^2}{2\hbar}-\frac{A_0}{\hbar}}\int_{{\gamma}(U_0)}  d \varphi\,\left(\frac{1-e^{-w\varphi}}{w}\right)^{k-1} e^{\frac{1}{\hbar}\left(x\varphi+\frac{e^{-w\varphi}-1+w\varphi -w^2\varphi^2/2}{w^3}-\frac{\hbar w \varphi}{2}\right)}.
\ee
This basis is of the form (\ref{goodbas}), which makes it suitable for construction of the matrix integral. Let
\be\label{measyr1}
\left[d \mu_w (\Phi)\right]:=\frac{1}{(2\pi)^\frac{N}{2}w^{\frac{N(N-1)}{2}}\prod_{k=1}^N k!}\Delta(\varphi)\Delta(e^{-w\varphi})\left[d { U}\right]\prod_{i=1}^N e^{-\frac{1}{2}w \varphi_i}d \varphi_i
\ee
be the measure on the space of the Hermitian matrices. At $w=0$ it reduces to the flat measure (\ref{flatmeas}), for $w\neq 0$ it is proportional to the flat measure, and the coefficient of proportionality is given by the exponential of the double trace potential \cite{ARP}.
\begin{remark}
Matrix integrals with the measure (\ref{measyr1}) naturally appear in description of the enumerative geometry invariants associated to the moduli spaces of Riemann surfaces, including simple Hurwitz numbers \cite{MS,ARP}, and Gromov-Witten invariants of ${\mathbb P}^1$ \cite{AP1,BR}. It is closely related to a matrix model for the linear Hodge integrals \cite{A1}.
\end{remark}
Using Harish-Chandra-Itzykson-Zuber formula one gets
 \begin{proposition}
 For $U=\frac{z}{1-wz}$ the tau-function in the Miwa parametrization is given by the matrix integral
\be\label{intm}
\tau_U\left(\left[\Lambda^{-1}\right]\right)=\tilde{c}\, \mathcal{C}^{-1}e^{-\frac{w}{2\hbar} \Tr V'(\Lambda)^2 } \int \left[d \mu_w (\Phi)\right] e^{\frac{1}{\hbar}\Tr\left(V'(\Lambda)\Phi+\frac{e^{-w\Phi}-1+w\Phi -w^2\Phi^2/2}{w^3}\right)}.
\ee
\end{proposition}
Here ${\mathcal C}$ is given by (\ref{NormalC}) for the potential (\ref{Vforlinear}), and $\tilde{c}$ is some normalization factor that does not depend on $\Lambda$.


\section{Deformed GKM and  Hodge integrals}\label{Sec:Hodge}

In this section we will discuss the relation between the deformed GKM and Hodge integrals. In particular, we prove that the deformation of the Kontsevich model, considered in Sections \ref{DefK}-\ref{S_DifK1}, describes triple Hodge integrals with imposed Calabi-Yau condition, which for certain values of parameters degenerate to linear Hodge integrals.

\subsection{KP tau-function for linear Hodge integrals}\label{S_Inters}

For linear Hodge integrals we consider the generating function:
\be\label{tildegf}
 {\mathcal F}^H_{g,n}({\bf T};u)= \sum_{a_1,\ldots,a_n\ge 0}\left<\Lambda_g(-u^2)  \tau_{a_1}\tau_{a_2}\cdots \tau_{a_n} \right>_g  \frac{\prod T_{a_i}}{n!}
\ee
and
\be\label{tildegf1}
Z({\bf T};u)=\exp\left(\sum_{g=0}^\infty \sum_{n=0}^\infty {\mathcal F}_{g,n}^H\right).
\ee
\begin{remark}
In this section we put $\hbar=1$ to simplify the connection with the existing literature \cite{Kaza,A1,Wang2,Wang1}. Dependence on $\hbar$ can be easily restored with the help of the dimensional constraint (\ref{dimc}).
\end{remark}
Consider the operator
\be\label{Dopp}
D=-(1+uz)^2 z^{-1} \frac{\p}{\p z}.
\ee
Then $\varphi_k:= D^k \cdot z^{-1}$ are polynomials in $z^{-1}$. On identification of $\varphi_k$ with $T_k$ and $z^{-k}$ with $k t_k$ one gets
a linear change of variables
\begin{equation}
\begin{split}\label{Kazatr}
T_0({\bf t})&=t_1,\\
T_{k}({\bf t})&=\sum_{m\geq 1} \left(m\,u^2 t_m+2(m+1)u\, t_{m+1}+(m+2)t_{m+2}\right) \frac{\p}{\p t_m}
 T_{k-1}({\bf t})\\
 &=\left(u^2 (\widehat{L}_0-\frac{ t_1^2}{2})+2u\,\widehat{L}_{-1}+\widehat{L}_{-2}\right)T_{k-1}({\bf t}),
\end{split}
\end{equation}
such that $T_{k}({\bf t})=(2k+1)!! t_{2k+1}+O(u)$. It allowed Kazarian to relate the generating function (\ref{tildegf1}) to the KP hierarchy \cite{Kaza}:
\begin{theorem*}[Kazarian]
\be\label{gf}
\tau_{H}({\bf t};u)=Z({\bf T}({\bf t});u).
\ee
is a tau-function of the KP hierarchy (identically in u).
\end{theorem*}
For $u=0$ only $\psi$-classes survive in (\ref{tildegf}), and $Z({\bf T};u)$ reduces to 
\be
\tau_{KW}:=\exp\left(\sum_{g=0}^\infty \sum_{n=1}^\infty \hbar^{2g-2+n}F_{g,n}\right),
\ee
where
\be
F_{g,n}:=\sum_{a_1,\ldots,a_n\ge 0}\<\tau_{a_1}\tau_{a_2}\cdots\tau_{a_n}\>_g\frac{\prod T_{a_i}}{n!}.
\ee
On substitution $T_k=(2k+1)!!\, t_{2k+1}$ this generating function yields the profound Kontsevich--Witten (KW) tau-function  \cite{Konts,W}  discussed in Section \ref{Sec_Kon}. KdV integrability also follows from the Kazarian theorem if one puts $u=0$.

It is known that linear Hodge integrals can be expressed through the intersection numbers of the $\psi$-classes \cite{M,FP}
\be\label{Givop}
Z({\bf T};u)
=\widehat{R}_u \cdot \tau_{KW}, 
\ee
where
\be
\widehat{R}_u=-\sum_{k=1}^\infty \frac{B_{2k}u^{4k-2}}{2k(2k-1)}\widehat{W}_k.
\ee
Here
\be
\widehat{W}_k=-\sum_{m}\tilde{T}_m\frac{\p}{\p T_{m+2k-1}}+\frac{1}{2}\sum_{m=0}^{-2k}(-1)^{l+1}\tilde{T}_m\tilde{T}_{-2k-m}+\frac{1}{2}\sum_{m=0}^{2k-2}(-1)^m\frac{\p^2}{\p T_m \p T_{2k-m-2}}.
\ee
The so-called {\em dilaton shift} of the variables is given by $\tilde{T}_k=T_k-\delta_{k,1}$. Here $B_{2k}$ are the Bernoulli numbers 
\be
\frac{xe^x}{e^x-1}=1+\frac{x}{2}+\sum_{k=1}^{\infty}\frac{B_{2k}x^{2k}}{(2k)!}.
\ee
Operator $\widehat{R}_u$ does not belong to the $\GL(\infty)$ symmetry group of the KP hierarchy. It belongs to Givental's group, acting on the space of semi-simple cohomological field theories  \cite{Giv1,Giv2}.

In \cite{A2} it was conjectured that the KW tau-function $\tau_{KW}$  and the Hodge tau-function  $\tau_H$ are related by a certain element of the Virasoro subgroup of the $\GL(\infty)$ group (see Section \ref{Virsec}). In \cite{A1} this conjecture was proved with the help of the Sato Grassmannian description. 
Corresponding group element was constructed up to a ${\bf t}$-independent factor (which was conjectured to be equal to one). 
Moreover, with the help of the Virasoro constrains, satisfied by the KW tau-function, an infinite dimensional family of the group elements of the Heisenberg--Virasoro group, providing equivalent relations between two tau-functions, was described. Using the conjugation of the KS operators of the KW tau-function, in \cite{A1} we constructed the KS operators for the tau-function $\tau_H$ and the Heisenberg-Virasoro constraints, which govern $\tau_H$: 
\begin{equation}
\begin{split}\label{foc}
\widehat{J}_m^H \cdot \tau_H&=0,\,\,\,\,\, m\geq 1\\
\widehat{L}_m^H \cdot \tau_H&=0,\,\,\,\,\, m\geq -1.
\end{split}
\end{equation}


\subsection{KS algebra for the Hodge tau-function}

Guo, Liu, and Wang \cite{Wang1, Wang2} constructed explicitly a particularly convenient representative of the family of the group operators from \cite{A1}. Their construction is based on the direct relation between the Givental operator (\ref{Givop}) and an element of the Heisenberg-Virasoro group. Using this relation, they have proved that the multiplicative constant is indeed $1$, as it was conjectured in \cite{A1}. This particular choice of the representative of the group elements allowed the authors of \cite{Wang1, Wang2} to find a convenient basis in the space of the constraints (\ref{foc}).

In particular, they found the string equation 
\be
V_{-1}^{(H)} \cdot \tau_H({\bf t};u)=0,
\ee
where
\be
V_{-1}^{(H)}:=\widehat{L}_{-2}+2u\widehat{L}_{-1}+u^2\widehat{L}_{0}-\frac{u^2}{24}-\sum_{j=1}^\infty(-u)^{j-1} \frac{\p}{\p t_j}.
\ee
This operator was obtained by the conjugation of the $\widehat{L}_{-1}^{KW}$ for the KW tau-function. Let us denote the point of the Sato Grassmannian, associated to the tau-function (\ref{gf}), by ${\mathcal W}_H$. The KS operator, associated with the string equation (with the inverse sign) can also be obtained by conjugation of ${\mathtt Q}_{KW}$. It is equal to
\be
{\mathtt V}_H:=-\left({\mathtt l}_{-2}+2u{\mathtt l}_{-1}+u^2{\mathtt l}_0 -\sum_{j=1}^\infty (-u)^{j-1} z^{j}\right).
\ee
It satisfies $V_H\cdot {\mathcal W}_H \subset {\mathcal W}_H$. Using (\ref{virw}) for the operators ${\mathtt l}_j$, one gets
\be
{\mathtt V}_H=\left(\frac{1+uz}{z}\right)^2z \frac{\p}{\p z}-\frac{1}{2z^2}+\frac{u^2}{2}+\frac{z}{1+uz}.
\ee
The second KS operator can also be extracted from the spectral parameter transformation $\eta$ of \cite{Wang2}, namely the KS operator ${\mathtt X}_{KW}=z^2/2$ for the Kontsevich-Witten tau-function transforms to
\be
{\mathtt X}_H=\frac{1}{u^2}\left(\log(1+uz)-\frac{uz}{1+uz}\right).
\ee
These operators satisfy the commutation relation $[{\mathtt V}_H,{\mathtt X}_H]=1$.

Let us stress that there is a certain arbitrariness in this choice of the variable transformations (\ref{Kazatr}). Namely, as we discussed in Section \ref{Virsec} the algebra $\sll(2)$ generates the automorphisms of KP hierarchy given by the linear changes of the variables. Let us find an operator which simplifies the change of variables. We act on the Hodge tau-function by the $GL(\infty)$ group element
\be
{\tau}_{\tilde{H}}({\bf t};u):=e^{-u \widehat{L}_1}\cdot \tau_H({\bf t};u).
\ee
As the tau-function $\tau_H({\bf t};u)$ belongs to ${\mathbb C}[[{\bf t},u]]$, the left hand side is well-defined and ${\tau}_{\tilde{H}}({\bf t};u)\in {\mathbb C}[[{\bf t},u]]$.
As it is described in Section \ref{Virsec} by equation (\ref{Gtr}) this operator acts on any function by the linear change of variables
\be
{\tau}_{\tilde{H}}({\bf t};u)=\tau_H({\bf t}^{-1}({\bf t});u),
\ee
where the varsities ${\bf t}^{(-1)}$ are given by (\ref{Tchant}).

The combination of this linear transformation with Kazarian's transformation (\ref{Kazatr}) leads to 
\be\label{Dtild}
D_u=e^{u z^2 \frac{\p}{\p z}} D e^{-u z^2 \frac{\p}{\p z}}=-\frac{1-uz}{z} \frac{\p}{\p z}.
\ee
It is associated with the change of variables
\be\label{THo}
\tilde{T}_0({\bf t})=t_1,\\
\tilde{T}_{k}({\bf t})=\sum_{m\geq 1} \left((m+2)t_{m+2}-(m+1)u\, t_{m+1}\right) \frac{\p}{\p t_m}
 \tilde{T}_{k-1}({\bf t})\\
 =\left(\widehat{L}_{-2}-u\,\widehat{L}_{-1}-\frac{ t_1^2}{2}\right)\cdot \tilde{T}_{k-1}({\bf t})
\ee
which defines a tau-function of the KP hierarchy
\be
\tau_{\tilde{H}}({\bf t})=\left.Z({\bf T}; u)\right|_{T_k=\tilde{T}_{k}({\bf t})}.
\ee
\begin{remark}
This change of variables is used by Kazarian in his recent work \cite{Kaza2}. He also proves that ${\tau}_{\tilde{H}}$, obtained by this change of variables from $e^{\tilde F}$, is a tau-function of the KP hierarchy.
\end{remark}

The KS operators for the tau-function $\tilde{\tau}_{H}$ can be obtained by conjugation of operators ${\mathtt V}_H$ and ${\mathtt X}_H$
\begin{equation}\label{almostcan}
\begin{split}
\tilde{\mathtt Q}_{\tilde{H}}&:= e^{-u l_1}\left( {\mathtt V}_H-\frac{u^2}{2}\right) e^{u l_1}=z+\left(\frac{1}{z}-u\right)\frac{\p}{\p  z} - \frac{1}{2z^2},\\
{\mathtt X}_{\tilde{H}}&:= e^{-u l_1}\, {\mathtt X}_H \, e^{u l_1}= -\frac{z}{u}-\frac{1}{u^2}\log(1-uz).
\end{split}
\end{equation}
These operators, up to the identification $w=u$, coincide with the operators for the simplest deformation of Kontsevich model, defined at $\hbar=1$ by (\ref{Simpbq}). According to Lemma \ref{deforlemma}, these constrains uniquely define a point of the Sato Grassmannian. Let us now restore the dependence on $\hbar$ in a way, consistent with deformed GKM:
\be
{\tau}_{\tilde{H}}({\bf t};u)=\exp\left({\sum_{g=0}^\infty \sum_{n=1}^\infty \hbar^{2g-2+n} {\mathcal F}_{g,n}^H(\tilde{\bf T}({\bf t}),u)}\right),
\ee
Now we have
\begin{theorem}\label{Th_Hodge}
Tau-function for linear Hodge integrals is given by the deformed GKM with potential $U=\frac{1}{\hbar}\frac{z}{1-wz}$
\be
{\tau}_{\tilde{H}}({\bf t};w)=\tau_U({\bf t}).
\ee
\end{theorem}
Hence, the results of Section \ref{Secdef}, in particular, of Section \ref{DefK}, are immediately applicable to ${\tau}_{\tilde{H}}$.


\subsection{Cubic Hodge integrals}

Let us consider the case of cubic Hodge integrals with an additional Calabi-Yau condition introduced in Section \ref{S_int_H}. Let
\be
D_{q,p}=-\frac{(1+\sqrt{p+q}z)(1+q z/\sqrt{p+q})}{z} \frac{\p}{\p z}.
\ee
Then the functions $\phi_k(z):=D_{q,p}^k\cdot \frac{1}{z}$
define a change of variables from $T_k^{q,p}$ to $t_k$ if we associate $k t_k$ with $z^{-k}$ and $\phi_k(z)$ with $T_k^{q,p}$. This change of variables is described by (\ref{changeof}).

For  $p\neq 0$ and $q\neq 0$ let us also consider
\be\label{fpq}
\frac{f(z)^2}{2}=\frac{p+q}{pq}\log\left(1+\frac{qz}{\sqrt{p+q}}\right)-\frac{1}{p}\log(1+\sqrt{p+q}z).
\ee
For $q=0$ it degenerates to
\be
\frac{f(z)^2}{2}=\frac{z}{\sqrt{p}}-\frac{1}{p}\log(1+\sqrt{p}z)
\ee
and for $p=0$ it degenerates to
\be
\frac{f(z)^2}{2}=\frac{1}{q}\log(1+\sqrt{q}z)-\frac{1}{\sqrt{q}}\frac{z}{1+\sqrt{q} z}.
\ee
For any $p+q\neq 0$ it determines an element of the Virasoro group $e^{\sum  a_k \widehat{L}_k }$ by (\ref{ffun}).

Let us also consider 
\be
v_k:=[z^k]\int_{0}^z (f(\eta)-y(\eta)) d x(\eta), 
\ee
where $x(z)=f(z)^2/2$ and
\be\label{ygr}
y(z) =\int^z \frac{d x(\eta)}{\eta},
\ee
which for $p\neq 0$ and $q\neq 0$ is given by
\be
y(z)  =\frac {\sqrt {p+q}}{p}   \left( \log  \left( 1+\sqrt {p+q}z \right)  -\log  \left(1+\frac{qz}{ \sqrt {p+q}} \right) \right).
\ee
For $q=0$ it reduces to
\be
y(z)=\frac {1}{\sqrt{p}} \log  \left( 1+\sqrt {p }z \right)
\ee
and for $p=0$ it degenerates to
\be
y(z)=\frac{z}{1+\sqrt{q}z}.
\ee

One of the main results of a companion paper \cite{H3_1} is the following 
\begin{theorem*}
Generating function of the triple Hodge integrals, satisfying the Calabi-Yau condition, in the variables (\ref{changeof}) is a tau-function of the KP hierarchy,
\be
\tau_{q,p}({\bf t})=Z_{q,p}({{\bf T}^{q,p}(\bf t)}).
\ee
It is related to the Kontsevich-Witten tau-function by an element of the Heisenberg-Virasoro group 
\be\label{RBY}
\tau_{q,p}({\bf t})=e^{\hbar^{-1} \sum_{k=4} v_k \frac{\p}{\p t_{k}}}e^{\sum_{k\in\z_{>0}} a_k \widehat{L}_k} \cdot  \tau_{KW}({\bf t}).
\ee
\end{theorem*}

\begin{remark}
Note that if $p\neq 0$ and $q\neq 0$, then $\tau_{q,p}({\bf 0})\neq 1$, because from the results of Faber and Pandharipande \cite{FP} one has
\be
\int_{\overline{\cal M}_{g}}  \Lambda_g (-q) \Lambda_g (-p) \Lambda_g (\frac{pq}{p+q})=\frac{1}{2}\left(\frac{q^2p^2}{q+p}\right)^{g-1}
\frac{B_{2g}}{2g}\frac{B_{2g-2}}{2g-2}\frac{1}{(2g-2)!}
\ee
for $g\geq 2$.
\end{remark}

Operators (\ref{Dopp}) and (\ref{Dtild}) are the specifications of $D_{q,p}$,
\be
D=D_{u^2,0},\,\,\,\,\,\,\,\,\,\,\,{D}_{-u}=D_{0,u^2}.
\ee
These operators are associated with the reduction to linear Hodge integrals at $p=0$ and $q=0$ correspondingly.

For the deformed Kontsevich model, considered in Section \ref{DefK} let us choose a parametrization 
\be
U=\frac{1}{\hbar}\frac{z}{(1+\sqrt{p+q}z)(1+q z/\sqrt{p+q})}
\ee
or, equivalentely
\be
V=\frac{(p+q)(qz\sqrt{p+q}+p+q)\ln(1+\frac{qz}{\sqrt{p+q}})-q^2(1+\sqrt{p+q}z)\ln(1+\sqrt{p+q}z)}{pq^2\sqrt{p+q}}-\frac{z}{q}.
\ee

Now we are ready to prove Theorem \ref{H3U}.
\begin{proof}[{\bf Proof of Theorem \ref{H3U}}]
Theorem follows from the comparison of equation (\ref{RBY}) with Theorem \ref{taufr}.
It is easy to see that the Virasoro parts of the Heisenberg-Virasoro group elements coincide. Hence it remains to compensate the difference of the translations.
The difference is described by the function
\be
\int_{0}^z (f(\eta)-y(\eta)d x(\eta)- \int_{0}^z (f(\eta)-\eta) dx(\eta)=\int_{0}^z (\eta-y(\eta)d x(\eta)
\ee
This completes the proof.
\end{proof}
KS operators for the tau-function $\tau_{q,p}({\bf t})$ are
\be
{\mathtt X}_{q,p}={\mathtt X}_{U}\\
\tilde{\mathtt Q}_{q,p}=\tilde{\mathtt Q}_{U}-z+y(z)=y(z)+\frac{1}{U}\frac{\p}{\p z}-\frac{U'}{2U^2}.
\ee
For $q=0$ these operators are related to the KS operator (\ref{almostcan}) by 
\be
\tilde{\mathtt Q}_{0,u^2}+u\, {\mathtt X}_{0,u^2}=\tilde{\mathtt Q}_{\tilde{H}}\Big|_{u\mapsto-u}.
\ee
We have a modification of Proposition \ref{lem_defvir}, namely the tau-function $\tau_{q,p}$ satisfies the constraints
\begin{equation}
\begin{split}
\widehat{J}_m^U\cdot \tau_{q,p} &=0 \,\,\,\,\,\,\,\,\,\,\,\, m\geq 1,\\
\widehat{\tilde{L}}_m^U\cdot \tau_{q,p}&=0 \,\,\,\,\,\,\,\,\,\,\,\, m\geq -1,
\end{split}
\end{equation}
where the Virasoro operators $\widehat{\tilde{L}}_m^U$ are obtained from $\widehat{{L}}_m^U$ by the shift of variables $t_k \mapsto t_k + \hbar^{-1}v_k$.


\subsection{General conjecture}

\begin{conjecture}
Deformed GKM with potential (\ref{nmondef}), after a shift of variables,  describes interesting enumerative geometry invariants in the $r$-spin case.
\end{conjecture}

We expect that change of variables from from $T_k$ to $t_k$ is always described by the operator $D_U=\frac{1}{U(z)}\frac{\p}{\p z}=\frac{\p}{\p x(z)}$. It would be interesting to find a ELSV-type relation between the GKM's with deformed potentials and generating functions of Hurwitz-type enumerative invariants. We expect that open versions of the corresponding Hodge integrals can be described by the insertion of the logarithmic term into the deformed GKM. Moreover, the corresponding points of the Sato Grassmannian should be completely fixed by the KS operators
\begin{equation}
\begin{split}
{\mathtt X}&=x(z)\\
\tilde{\mathtt Q}&=y(z)+\frac{1}{\hbar}\left(\frac{\p}{\p x(z)}-\frac{1}{2}\frac{\p}{\p x(z)}\log\left(\frac{\p x(z)}{\p z}\right)\right).
\end{split}
\end{equation}
These topics will be considered elsewhere.


\end{document}